\documentclass[a4paper,leqno]{article}

\usepackage[right=2.4cm,left=2.4cm,top=2.2cm,bottom=2.8cm]{geometry}
\usepackage[latin1]{inputenc}
\usepackage{amsmath,amssymb,amsthm,
	verbatim,
	fancyhdr,amsfonts
}

\usepackage[format=plain,labelfont=it,textfont=it]{caption}
\usepackage{needspace}
\usepackage{float}
\usepackage{tikz}
\usetikzlibrary{decorations.pathreplacing}

\usepackage[pdftex]{lscape}
\usepackage{amsmath}
\usepackage{amsthm}
\usepackage{amsfonts,amsmath,amssymb}
\usepackage{enumerate,verbatim,color}
\usepackage{epsfig}
\usepackage{textcomp}
\usepackage{graphicx}
\usepackage{amsfonts}
\usepackage{mathrsfs}
\usepackage{upgreek}
\usepackage{mathtools}
\usepackage{pdfpages}

\newcommand{\virgolette}[1]{``#1''}
\newtheorem{teorema}{Theorem}[section]
\newtheorem{coro}[teorema]{Corollary}
\newtheorem{lemma}[teorema]{Lemma}
\newtheorem{prop}[teorema]{Proposition}
\newtheorem{osss}[teorema]{Remark}

\newtheorem*{oss}{Remark}
\theoremstyle{definition}
\newtheorem{defi}[teorema]{Definition}
\theoremstyle{remark}
\newtheorem*{assumption1}{\bf Assumptions I}
\newtheorem*{assumption2}{\bf Assumptions II}
\newtheorem*{boundaries}{\bf Boundary conditions}
\newtheorem*{notations}{\bf Notations}
\newcommand{\iii}{{\, \vert\kern-0.25ex\vert\kern-0.25ex\vert\, }}

\DeclareMathOperator{\spn}{span}
\DeclareMathOperator{\Imm}{Im}

\newcommand{\ffi}{\varphi}

\newcommand{\AL}{\mathcal{A}}

\newcommand{\CC}{\mathcal{C}}
\newcommand{\LL}{\mathcal{L}}
\newcommand{\MM}{\mathcal{M}}

\newcommand{\NN}{\mathcal{N}}
\newcommand{\KK}{\mathcal{K}}

\newcommand{\Di}{\mathcal{D}}

\newcommand{\Z}{\mathbb{Z}}
\newcommand{\N}{\mathbb{N}}
\newcommand{\R}{\mathbb{R}}
\newcommand{\Q}{\mathbb{Q}}
\newcommand{\C}{\mathbb{C}}

\newcommand{\Hi}{\mathscr{H}}

\newcommand{\Gi}{\mathscr{G}}

\newcommand{\dd}{\partial}

\newcommand{\ra}{\rangle}
\newcommand{\la}{\langle}

\newcommand{\G}{\Gamma}

\begin{document}

\date{}
\title{Global exact controllability of bilinear quantum systems on compact graphs and energetic controllability}
 \author{Alessandro \textsc{Duca}\footnote{Universit\'e 
 Grenoble Alpes, CNRS, Institut Fourier, F-38000 Grenoble, France, email: { 
 alessandro.duca@univ-grenoble-alpes.fr}} }
%
\maketitle

\begin{abstract}The aim of this work is to study the controllability of the bilinear Schr\"odinger equation on 
compact graphs. In particular, we consider the equation $(\ref{mainx1})$ $i\dd_t\psi=-\Delta\psi+u(t)B\psi$ in the 
Hilbert space $L^2(\Gi,\C)$, with $\Gi$ being a compact graph. The Laplacian $-\Delta$ is equipped with 
self-adjoint boundary conditions, $B$ is a bounded symmetric operator and $u\in L^2((0,T),\R)$ with $T>0$. We 
provide a new technique leading to the global exact controllability of the $(\ref{mainx1})$ in $D(|\Delta|^{s/2})$ 
with $s\geq 3$. Afterwards, we introduce the {\virgolette{energetic controllability}}, a weaker notion of 
controllability useful when the global exact controllability fails.
In conclusion, we develop some applications of the main results involving for instance star graphs.
\end{abstract}

\section{Introduction}\label{intro}

\noindent
In quantum mechanics, any state of a closed system is mathematically represented by a wave function $\psi$ in the 
unit sphere of a Hilbert space $\Hi$. 
We consider the evolution of a particle confined in a network shaped as compact graph $\Gi$ (see Figure 
\ref{fig:1}) and subjected to an external field which plays the role of control. \begin{figure}[H]
	\centering
	\includegraphics[width=\textwidth-150pt]{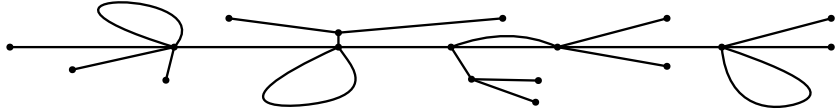}
	\caption{A compact graph is a one-dimensional domain composed by finite vertices (points) connected by edges 
(segments) of finite lengths.}\label{fig:1}
\end{figure}
A standard choice for such setting is to represent the action of the field by an operator $B$ and its intensity by 
a real function $u$. We also impose that $\Hi$ is $L^2(\Gi,\C)$. The evolution of $\psi$ is modeled by the 
bilinear 
Schr\"odinger equation in $L^2(\Gi,\C)$
\begin{equation}\label{mainx1}\tag{BSE}\begin{split}
\begin{cases}
i\dd_t\psi(t)=A\psi(t)+u(t)B\psi(t),\ \ \ \ \ \ \ \ &t\in(0,T),\\
\psi(0)=\psi_0,\ &\ \ \ \ T>0.\\
\end{cases}
\end{split}
\end{equation}
The Laplacian $A=-\Delta$ is equipped with self-adjoint boundary conditions, $B$ is a bounded symmetric operator 
and $u\in L^2((0,T),\R)$. In this context, the well-posedness of the $(\ref{mainx1})$ can be deduced by the seminal 
work on bilinear systems $\cite{ball}$ by Ball, Mardsen and Slemrod where they show the existence of the 
unitary propagator $\G_t^u$ generated by $A+u(t)B.$

\vspace{0.5mm}

The aim of this work is to study the controllability of the $(\ref{mainx1})$ according to the structure of the 
graph $\Gi$, the choice of control field $B$ and the boundary conditions defining the domain of $A$. 

\vspace{0.5mm}

The controllability of finite-dimensional quantum systems modeled by the $(\ref{mainx1})$, when $A$ and $B$ are 
$N\times N$ Hermitian matrices, is well-known for being  
linked to the rank of the Lie algebra spanned by $A$ and $B$ (see $\cite{basso,corona}$). Nevertheless the Lie 
algebra rank condition can not be used for infinite-dimensional quantum systems (see $\cite{corona}$). 

\vspace{0.5mm}

The {global approximate controllability} of the $(\ref{mainx1})$ was proved with different techniques in 
literature. We refer to $\cite{milo,nerse2}$ for Lyapunov techniques, while we cite $\cite{ugo2,ugo3}$ for 
adiabatic arguments and $\cite{nabile,ugo}$ for Lie-Galerking methods.

\vspace{0.5mm}

The {exact controllability} of infinite-dimensional quantum systems is in general a more delicate matter. When 
we consider the linear Schr\"odinger equation, the controllability and observability properties are reciprocally 
dual. Different results were developed by addressing directly or by duality the control problem with different 
techniques: multiplier methods $\cite{lion,cagn}$, 
microlocal analysis $\cite{rauch,burqa,lollo}$ and 
Carleman estimates $\cite{44,kaska,330}$. In any case, when one considers graphs type domains, a complete theory 
is 
far from being formulated. Indeed, the interaction between the different components of a graph may generate 
unexpected phenomena (see $\cite{wave}$). 

\vspace{0.5mm}

The {bilinear Schr\"odinger equation} is well-known for not being exactly controllable in the Hilbert space 
where it is defined when $B$ is a bounded operator and $u\in L^2((0,T),\R)$ with $T>0$ (even though it is 
well-posed in such space). This result was proved by Turinici in \cite{torino} by exploiting the techniques 
from 
the work \cite{ball}. 
As a consequence, the exact controllability of bilinear quantum systems can not be addressed with the classical 
techniques valid for the linear Schr\"odinger equation and weaker notions of controllability are necessary.
The turning point for this kind of studies has been the idea of controlling the equation in subspaces of $D(A)$ 
introduced by Beauchard in $\cite{be1}$. Following this approach, different works were developed for the bilinear 
Schr\"odinger equation \eqref{mainx1} with $\Gi=(0,1)$ and $A=-\Delta_D$ the Dirichlet Laplacian:
$$D(-\Delta_D)=H^2((0,1),\C)\cap H^1_0((0,1),\C)),\ \ \ \ \ -\Delta_D\psi:=-\Delta\psi,\ \ \ \ \forall\psi\in 
D(-\Delta_D).$$

\noindent
For instance, in $\cite{laurent}$, Beauchard and Laurent proved the {well-posedness} and the {local exact 
controllability} of the bilinear Schr\"odinger equation in $H^{s}_{(0)}:=D(|-\Delta_D|^{s/2})$ for $s= 3$. 
For the {global exact controllability} in $H^3_{(0)}$, we refer to $\cite{beauchard2017,mio2}$, while we 
mention $\cite{morgane1,morganerse2}$ for {simultaneous exact controllability} results in $H^3_{(0)}$ and 
$H^4_{(0)}$.

\vspace{0.5mm}

Studying the controllability of the bilinear Schr\"odinger equation \eqref{mainx1} on compact graphs presents an 
additional problem. In particular, when we consider $(\lambda_k)_{k\in\N^*}$, the ordered sequence of eigenvalues 
of $A$, it is possible to show that there exists $\MM\in\N^*$ such that
\begin{equation}\label{g13}\begin{split}
\inf_{{k\in\N^*}}|\lambda_{k+\MM}-\lambda_k|>0\\
\end{split}\end{equation}
(as ensured in $\cite[Lemma\ 2.4]{mio3}$). Nevertheless, the uniform {spectral gap} 
$\inf_{{k\in\N^*}}|\lambda_{k+1}-\lambda_k|> 0$ is only valid when $\Gi=(0,1)$. This hypothesis is crucial for the 
techniques developed in the works $\cite{laurent,mio2,morgane1}$, which can not be directly applied without 
imposing further assumptions. 

\vspace{0.5mm}

As far as we know, the bilinear Schr\"odinger equation on compact graphs has only been studied in $\cite{mio3}$. 
There, the author ensures that, if there exist $C>0$ and suitable $\tilde d\geq 0$ such that
\begin{align}\label{weak_gap}|\lambda_{k+1}-\lambda_k|\geq {C}{k^{-\tilde d}},\ \ \ \ \ \ \forall 
k\in\N^*,\end{align}
then the {global exact controllability} of the $(\ref{mainx1})$ can be guaranteed in some subspaces of 
$L^2(\Gi,\C)$.

\vspace{4mm}

{\noindent \bf \underline{Novelties of the work: Global exact controllability}}

\vspace{2mm}

\noindent
The aim of this work is to present a new technique ensuring the {global exact controllability} of the 
$(\ref{mainx1})$ in different frameworks from the ones considered in \cite{mio3}. Here, we focus on discussing few 
interesting
applications of our result involving {star graphs} composed by any number of edges. The general outcome is 
postponed to the 
next section (Theorem \ref{global}) in order to avoid further technicalities at this moment.

\vspace{0.5mm}

Let $\Gi$ be a {star graph} with $N\in\N^*$ edges $\{e_j\}_{j\leq N}$. We denote by $v$ the internal vertex of 
$\Gi$ and by $\{v_j\}_{j\leq N}$ the set of the external vertices such that $v_j\in e_j$ for every $j\leq N$. Each 
edge $e_j$ with $j\leq N$ is equipped with a coordinate going from $0$ to the length of the edge $L_j$. We set 
the 
coordinate $0$ in $v_j$. 
\begin{figure}[H]
	\centering
	\includegraphics[width=\textwidth-150pt]{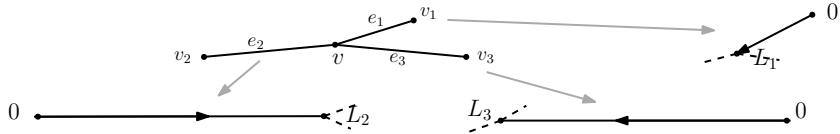}
	\caption{Parametrization of a star graph with $N=3$ edges.}\label{parametrizzazione}
\end{figure}

\noindent
We consider functions $f:=(f^1,...,f^N):\Gi\rightarrow \C$ so that $f^j:(0,L_j)\rightarrow \C$ for every $j\leq N$. 
We denote $L^2(\Gi,\C)=\prod_{j\leq N}L^2((0,L_j),\C)$ the Hilbert space equipped with the norm $\|\cdot\|_{L^2}$.
The controllability result that we present is guaranteed when the lengths $\{L_j\}_{j\leq N}$ satisfy suitable 
assumptions introduced in the following definition.

\begin{defi}
	Fixed $N\in\N^*$, we define $\AL\LL(N)$ such as the set of $\{L_j\}_{j\leq N}\in (\R^+)^N$ so that 
the numbers $\{1,L_1,...,L_N\}$ are {linearly independent} over $\Q$ and all the ratios 
$L_k/L_j$ are 
{algebraic irrational numbers}.
\end{defi}
The set $\AL\LL(N)$ with $N\in\N^*$ contains the uncountable set of $\{L_{j}\}_{j\leq 
N}\in(\R^+)^N$ such
that each $L_j$ can be written in the form $t\widetilde L_j$ where all the ratios $\widetilde L_j/\widetilde L_k$ 
are algebraic irrational numbers and $t$ is a transcendental number. For instance, $\{\pi\sqrt{2},\pi\sqrt{3}\}$ 
belongs to $\AL\LL(2)$.
\begin{defi}
Let $\G_T^u$ be the unitary propagator associated to the \eqref{mainx1} with $u\in L^2((0,T),\R)$ and $T>0$. The 
$(\ref{mainx1})$ is globally exactly controllable in $D(|A|^\frac{s}{2})$ with $s\geq 3$ when, for every 
$\psi^1,\psi^2\in D(|A|^\frac{s}{2})$ such that $\|\psi^1\|_{L^2}=\|\psi^2\|_{L^2}$, there exist $T>0$ and $u\in 
L^2((0,T),\R)$ such that $\G_T^u\psi^1=\psi^2.$
\end{defi}

We are finally ready to present two interesting results obtained from the techniques of this work.

\needspace{3\baselineskip}\begin{teorema}\label{bim.1}  
		Let $\Gi$ be a star graph. Let $D(A)$ be the set of functions $f\in \prod_{j\leq N} H^2(e_j,\C)$ such that:
	\begin{itemize}
		\item $f^j(0)=0$ for every $j\leq N$ (Dirichlet boundary conditions in the external vertices 
$\{v_j\}_{j\leq N}$);
		 \item $f^j(L_j)=f^k(L_k)$ for every $j,k\leq N$ and $\sum_{j\leq 
N}\dd_x f^j(L_j)=0$ {(Neumann-Kirchhoff} conditions in $v$).
\end{itemize}
Let the operator $B$ be such that:
\begin{align*}B:\psi=(\psi^1,...,\psi^N)\in 
L^2(\Gi,\C)\longmapsto\big((x-L_1)^4\psi^1(x),0,...,0\big).\end{align*} There 
exists $\CC\subset (\R^+)^N$ countable such that, 
for every $\{L_j\}_{j\leq N}\in\AL\LL(N)\setminus \CC$, the $(\ref{mainx1})$ is {globally exactly controllable} in 
$D\big(|A|^\frac{4+\epsilon}{2}\big)$ for every $\epsilon>0.$ 
\end{teorema}

A similar result to Theorem \ref{bim.1} is the following. Here, the \eqref{mainx1} is considered on a 
generic star 
graph equipped with Neumann boundary conditions on the external vertices instead of the Dirichlet ones.

\begin{teorema}\label{bim.2}
Let $\Gi$ be a star graph. Let $D(A)$ be the set of functions $f\in \prod_{j\leq N} H^2(e_j,\C)$ such that:
	\begin{itemize}
		\item $\dd_x f^j(0)=0$ for every $j\leq N$ (Neumann boundary conditions in the external vertices 
$\{v_j\}_{j\leq N}$);
		 \item $f^j(L_j)=f^k(L_k)$ for every $j,k\leq N$ and $\sum_{j\leq 
N}\dd_x f^j(L_j)=0$ {(Neumann-Kirchhoff} conditions in $v$).
\end{itemize}
Let the operator $B$ be such that: 
\begin{align*}B:\psi=(\psi^1,...,\psi^N)\in L^2(\Gi,\C)\longmapsto\Big(\big(5x^6-24x^5 L_1 
+45 x^4 L_1^2-40 x^3 L_1^3 +15 x^2 L_1^4 -L_1^6\big)\psi^1(x),0,...,0\Big).\end{align*} There exists $\CC\subset 
(\R^+)^N$ countable such that, 
for every $\{L_j\}_{j\leq N}\in\AL\LL(N)\setminus \CC$, the $(\ref{mainx1})$ is {globally exactly controllable} in 
$D\big(|A|^\frac{5+\epsilon}{2}\big)$ for every $\epsilon>0.$ 
\end{teorema}

Theorem \ref{bim.1} and Theorem \ref{bim.2} are deduced from an abstract global exact controllability 
result 
stated in the next section 
in Theorem \ref{global}. Such theorem presents hypotheses on $(\lambda_k)_{k\in\N^*}$ and on the operator $B$  
so that the controllability of the bilinear Schr\"odinger equation \eqref{mainx1} is 
guaranteed for a general $\Gi$.

\vspace{0.5mm}

As it is common for this type of outcomes, the global exact controllability of the \eqref{mainx1} can be 
ensured by extending a local result following from the solvability of a suitable 
\virgolette{moment problem} (an example can be found in \eqref{mome_generic}).
In order to prove Theorem \ref{bim.1} and Theorem \ref{bim.2}, we start by studying assumptions on 
$(\lambda_k)_{k\in\N^*}$ and on the operator $B$ leading to the solvability of such 
moment problem. In a second moment, we prove the abstract global exact controllability result of 
Theorem \ref{global}.
Afterwards, we ensure the validity of the spectral assumptions considered 
in Theorem \ref{global} for suitable star graphs.
In conclusion, we validate the remaining hypotheses when the operator $B$ is defined as in Theorem 
\ref{bim.1} or Theorem \ref{bim.2}.

\vspace{0.5mm}

The main novelty of Theorem \ref{bim.1} and Theorem \ref{bim.2} is the validity of the controllability results 
when 
$\Gi$ is a star graph with any number of edges. In fact, the techniques developed in the existing work \cite{mio3} 
only allow to consider star graphs with at most $4$ edges (see \cite[Proposition\ 3.3]{mio3}). In addition, the 
controllability is guaranteed even though the control field $B$ only acts on one edges of the graph, which is due 
to the choice of the lengths $\{L_{j}\}_{j\leq N}$ in $\AL\LL(N)$. About this fact, if some ratios $L_k/ L_1$ are 
rationals, then the spectrum of the operator $A$ presents multiple eigenvalues and there exist eigenfunctions of 
$A$ vanishing in $e_1$ (we refer to Remark \ref{pallosso} for further details on this fact). As a consequence, the 
dynamics of the bilinear Schr\"odinger equation \eqref{mainx1} stabilizes such eigenfunctions since the control 
operator $B$ only acts on $e_1$. This is an obvious obstruction to the controllability which underlines the 
importance of choosing suitable lengths for the edges of the graph in this kind of problems.

\vspace{4mm}

{\noindent \bf \underline{Novelties of the work: Energetic controllability}}

\vspace{2mm}

\noindent
In the spirit of the results provided in $\cite{corobeau}$, we 
introduce a weaker notion of controllability: the {{energetic controllability}}.

\begin{defi}\label{energiaaa}
Let $(\varphi_{k})_{k\in\N^*}$ be an orthonormal system of $L^2(\Gi,\C)$ (not necessarily complete) composed by 
eigenfunctions of $A$ and $(\mu_{k})_{k\in\N^*}$ be the corresponding eigenvalues. Let $\G_T^u$ be the unitary 
propagator associated to the \eqref{mainx1} with $u\in L^2((0,T),\R)$ and $T>0$.
	The $(\ref{mainx1})$ is energetically controllable in $(\mu_k)_{{k\in\N^*}}$ if, for every $m,n\in\N^*$, there 
exist $T>0$ and $u\in L^2((0,T),\R)$ so that $\G_T^u\varphi_m=\varphi_n.$
\end{defi}

The energetic controllability guarantees the controllability of specific energy levels of the quantum system 
$i\dd_t\psi=A\psi$ in $L^2(\Gi,\C)$ via the external field $u(t)B$. An application of the abstract energetic 
controllability result, which is presented in Section \ref{controlloenergia} (in Theorem 
$\ref{globalenergetic}$), is the following theorem.

\begin{teorema}\label{bim2}
	Let $\Gi$ be a star graph with edges of equal length $L$. Let $D(A)$ be defined such as in Theorem 
$\ref{bim.1}$. 
	Let the operator $B$ be such that: 
\begin{align*}B:\psi=(\psi^1,...,\psi^N)\in L^2(\Gi,\C)\longmapsto\big((x-L)^2\psi^1(x),0,...,0\big).\end{align*}
	The $(\ref{mainx1})$ is energetically controllable in $\big(\frac{k^2\pi^2}{4L^2}\big)_{k\in\N^*}.$ 
\end{teorema}
Theorem $\ref{bim2}$ is valid although the spectrum of $A$ presents multiple eigenvalues and then the global exact 
controllability from Theorem $\ref{global}$ is not satisfied (also $\cite[Theorem\ 3.2]{mio3}$ is not guaranteed).
In addition, the energetic controllability is ensured with respect to all the energy levels of the quantum system 
$i\dd_t\psi=A\psi,$ since the eigenvalues of $A$ are $\big(\frac{k^2\pi^2}{4L^2}\big)_{k\in\N^*}$ (without 
considering their multiplicity).

\vspace{0.5mm}

The energetic controllability is useful when it is not possible to fully characterize the spectrum of $A$ because 
of the complexity of the graph $\Gi$. By studying the structure of $\Gi$, it is possible to explicit some 
eigenvalues $(\mu_k)_{k\in\N^*}$ and verify if the system is energetically controllable in $(\mu_k)_{k\in\N^*}$. 
In Section $\ref{energia}$, we discuss some examples where the result is satisfied, {\it e.g } graphs containing 
loops as in Figure \ref{cappio} (a loop is an edge of the graph which is connected from both extremes to the same 
vertex).

\begin{figure}[H]
	\centering
	\includegraphics[width=\textwidth-150pt]{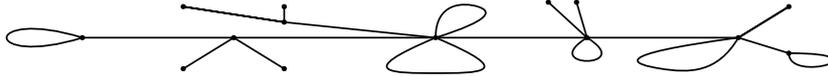}
	\caption{Example of compact graph containing more loops.}
	\label{cappio}
\end{figure}

{\noindent \bf \underline{Scheme of the work}}\\

\noindent
In {Section $\ref{accurate}$}, we present the main assumptions adopted in the work, the well-posedness of the 
$(\ref{mainx1})$ in specific subspaces of $L^2(\Gi,\C)$ and the abstract global exact controllability result in 
Theorem \ref{global}.

\noindent
In Section $\ref{trigonometric}$, we study the solvability of the \virgolette{moment problem} under the hypotheses 
of Theorem \ref{global}.

\noindent
In Section \ref{proofglobal}, we prove Theorem \ref{global} by extending a local exact 
controllability result provided in Proposition \ref{localteorema}. To the purpose, we use the outcomes developed 
in Section 
$\ref{trigonometric}$ and Appendix $\ref{approximate}$.

\noindent
In Section \ref{generic_application}, we study the \eqref{mainx1} on star graphs and we prove Theorem \ref{bim.1} 
and Theorem \ref{bim.2}.

\noindent
In Section \ref{controlloenergia}, we discuss the energetic controllability of the \eqref{mainx1} and we prove 
Theorem \ref{bim2}.

\noindent
In Appendix $\ref{approximate}$, we present the global approximate controllability of the bilinear Schr\"odinger 
equation.

\noindent
In {Appendix $\ref{numeri}$}, we study some spectral results adopted in the work.

\section{The bilinear Schr\"odinger equation on compact graphs}\label{accurate}
\subsection{Preliminaries} 
Let $\Gi$ be a compact graph composed by $N\in\N^*$ edges $\{e_j\}_{j\leq N}$ of lengths $\{L_j\}_{j\leq N}$ and 
$M\in\N^*$ vertices $\{v_j\}_{j\leq M}$. For every vertex $v$, we denote $N(v):=\big\{l \in\{1,...,N\}\ |\ v\in 
e_l\big\}$ and $n(v):=|N(v)|.$
We call $V_e$ and $V_i$ the external and the internal vertices of $\Gi$, {\it i.e.}
\begin{equation*}V_e:=\big\{v\in\{v_j\}_{ j\leq M}\ |\ \exists ! e\in\{e_j\}_{j\leq N}: v\in 
e\big\},\ \ \ \ \ V_i:=\{v_j\}_{ j\leq M}\setminus V_e.\end{equation*}
We study graphs equipped with a metric, which parametrizes each edge $e_j$ with a coordinate going from $0$ to its 
length $L_j$. A graph is compact when it is composed by a finite number of vertices and edges of finite lengths. We 
consider functions $f:=(f^1,...,f^N):\Gi\rightarrow \C$ with domain a compact metric graph $\Gi$ so that 
$f^j:e_j\rightarrow \C$ for every $j\leq N$. We denote
$$\Hi=L^2(\Gi,\C)=\prod_{j\leq N}L^2(e_j,\C).$$
The Hilbert space $\Hi$ is equipped with the norm $\|\cdot\|_{L^2}$ induced by the scalar product
$$\la\psi,\ffi\ra_{L^2}:=\sum_{j\leq N}\la\psi^j,\ffi^j\ra_{L^2(e_j,\C)}=\sum_{j\leq 
N}\int_{e_j}\overline{\psi^j}(x)\ffi^j(x)dx,\ \ \ \  \ \ \forall \psi,\ffi\in\Hi.$$
For $s>0$, we define the spaces $$H^s=H^s(\Gi,\C):=\prod_{j=1}^N H^s(e_j,\C),$$
$$h^s=\Big\{(x_j)_{j\in\N^*}\subset{\C}\ 
\big|\ \sum_{j=1}^{\infty}|j^s x_j|^2<\infty\Big\}.$$
We equip $h^s$ with the norm 
$\big\|(x_j)_{j\in\N^*}\big\|_{(s)}=\big(\sum_{j=1}^\infty|j^s x_j|^2\big)^\frac{1}{2}$ for every 
$(x_j)_{j\in\N^*}\in h^s$. 

\vspace{0.5mm}

Let $f=(f^1,...,f^N):\Gi\rightarrow\C$ be smooth and $v$ be a vertex of $\Gi$ connected once to an edge $e_j$ with 
$j\leq N$. When the coordinate parametrizing $e_j$ in the vertex $v$ is equal to $0$ (resp. $L_j$), we denote
\begin{align}\label{NK1}\dd_x f^j(v)=\dd_xf^j(0),\ \ \  \ \ \ \ \ \big(\text{resp.}\ \dd_x 
f^j(v)=-\dd_xf^j(L_j)\big).\end{align}
When $e_j$ is a loop and it is connected to $v$ in both of its extremes, we use the notation 
\begin{align}\label{NK2}\dd_x 
f^j(v)=\dd_xf^j(0)-\dd_xf^j(L_j).\end{align} When $v$ is an external vertex and then $e_j$ is the only edge 
connected to $v$, we call $\dd_x f(v)=\dd_x f^j(v).$

\vspace{0.5mm}

In the bilinear Schr\"odinger equation $(\ref{mainx1})$, we consider the Laplacian $A$ being self-adjoint and we 
denote $\Gi$ as {quantum graph}. From now on, when we introduce a quantum graph $\Gi$, we implicitly define on 
$\Gi$ a self-adjoint Laplacian $A$. Formally, $D(A)$ is characterized by the following boundary conditions.

\needspace{3\baselineskip}
\begin{boundaries}
Let $\Gi$ be a compact quantum graph. 
\begin{itemize}
\item[($\Di$)] A vertex $v\in V_e$ is equipped with {Dirichlet} boundary conditions when $f(v)=0$ for every $f\in 
D(A)$.

\item[($\NN$)] A vertex $v\in V_e$ is equipped with {Neumann} boundary conditions when $\dd_xf(v)=0$ for every 
$f\in D(A)$.
 \item[($\NN\KK$)] A vertex $v\in V_i$ is equipped with {Neumann-Kirchhoff} boundary conditions when every $f\in 
D(A)$ is 
continuous in $v$ and $\sum_{j\in N(v)}\dd_x f^j(v)=0$ (we refer to \eqref{NK1} and \eqref{NK2} for further 
details on the notation).

\end{itemize}
\end{boundaries}

\needspace{1\baselineskip}
\begin{notations}Let $\Gi$ be a compact quantum graph. 
\begin{itemize}

\item The graph $\Gi$ is said to be equipped with ($\Di$) (resp. ($\NN$)) when every $v\in V_e$ is equipped with 
($\Di$) (resp. ($\NN$)) and every $v\in V_i$ with ($\NN\KK$). 

\item The graph $\Gi$ is said to be equipped with ($\Di$/$\NN$) when every $v\in V_e$ is equipped with ($\Di$) or 
($\NN$), while every $v\in V_i$ with ($\NN\KK$). 
\end{itemize}

\end{notations}

When the boundary conditions described above are satisfied, the Laplacian $A$ is self-adjoint (see \cite[Theorem\ 
3]{kuk} for further details) and admits purely discrete spectrum (see \cite[Theorem\ 18]{kuk}). We denote by 
$(\lambda_k)_{k\in\N^*}$ the ordered sequence of eigenvalues of $A$ and we define a Hilbert basis of $\Hi$:  
\begin{equation*}\Phi:=(\phi_k)_{k\in\N^*}\end{equation*}composed by corresponding eigenfunctions.
 From $\cite[Lemma\ 2.3]{mio3}$, there exist $C_1,C_2>0$ so that
	\begin{equation}\label{interessante}C_1{k^2}\leq\lambda_k\leq C_2k^2,\ \ \ \ \ \forall k\geq 2.\end{equation}

\noindent
Let $[r]$ be the entire part of a number $r\in\R$. For $s>0$, we denote
\begin{equation*}
\begin{split}
H^s_{\NN\KK}:=\Big\{&\psi\in H^s(\Gi,\C)\ \Big|\ \dd_x^{2n}\psi\text{ is continuous in }v,\ 
 \forall n\in\N,\ n<\big[({s+1})/{2}\big],\ \forall v\in V_i;\\
& \sum_{j\in N(v)}\dd_{x}^{2n+1}\psi^j(v)=0,\ \forall n\in\N,\ n<\big[{s}/{2}\big],\  \forall v\in V_i\Big\},\\
H^s_{\Gi}&:=D(A^{{s}/{2}}),\ \ \ \ \ \ \ \ \ \ \ \ \ \ 
\|\cdot\|_{(s)}:=\|\cdot\|_{H^s_{\Gi}}=\Big(\sum_{k\in\N^*}\big|k^s\la\cdot,\phi_k\ra_{L^2}\big|^2\Big)^{\frac{1}{2
} } .\\
\end{split}
\end{equation*}

We introduce the main assumptions adopted in the manuscript by considering 
$(\mu_{k})_{k\in\N^*}\subseteq(\lambda_{k})_{k\in\N^*}$ an ordered sequence of some eigenvalues of $A$ and the 
corresponding eigenfunctions $$\upvarphi:=(\varphi_{k})_{k\in\N^*}\subseteq(\phi_k)_{k\in\N^*}.$$
Let $\eta>0$, $a\geq 0$, $I:=\{(j,k)\in(\N^*)^2:j< k\}$ and $\widetilde\Hi:=\overline{\spn\{\varphi_k\ |\ 
k\in\N^*\}}^{\ L^2}.$

\smallskip

\needspace{1\baselineskip}
\begin{assumption1}[$\upvarphi,\eta$]
	The bounded symmetric operator $B$ satisfies the following conditions.

	\begin{enumerate}
		\item There exists $C>0$ such that $\big|\la\ffi_k,B\ffi_1\ra_{L^2}\big|\geq\frac{C}{k^{2+\eta}}$ for 
every 
$k\in\N^*.$
		\item For every $(j,k),(l,m)\in I$ such that $(j,k)\neq(l,m)$ and $\mu_j-\mu_k-\mu_l+\mu_m=0,$ it holds 
$$\la\ffi_j,B\ffi_j\ra_{L^2}-\la\ffi_k,B\ffi_k\ra_{L^2}-\la\ffi_l,B\ffi_l\ra_{L^2}+\la\ffi_m,B\ffi_m\ra_{L^2}\neq 
0.$$
			\end{enumerate}

\end{assumption1}

\begin{assumption1}[$\eta$]
	The couple $(A,B)$ satisfies Assumptions I$(\Phi,\eta)$ with $\Phi$ a Hilbert basis of $\Hi$ made by 
eigenfunctions of $A$.
\end{assumption1}	
\vspace{1mm}

The first condition of Assumptions I($\upvarphi,\eta$) (resp. Assumptions I($\eta$)) quantifies how much $B$ mixes 
the eigenspaces associated to the eigenfunctions $(\ffi_k)_{k\in\N^*}$ (resp. $(\phi_k)_{k\in\N^*}$). This 
assumption is crucial for the controllability. Indeed, when $B$ stabilizes such spaces, also $\G_T^u$ does the 
same 
and we can not expect to obtain controllability results. The second hypothesis is used to decouple some 
eigenvalues 
resonances appearing in the proof of the approximate controllability that we use in order to prove our main 
results. 

\needspace{3\baselineskip}
\begin{assumption2}[$\upvarphi,\eta,a$] Let $Ran(B|_{H^2_{\Gi}\cap \widetilde \Hi})\subseteq H^2_{\Gi}\cap 
\widetilde \Hi$ and one of the following points be satisfied.
	\begin{enumerate}
		
	\item When $\Gi$ is equipped with ($\Di$/$\NN$) and $a+\eta\in(0, 3/2)$, there exists 
	$d\in[\max\{a+\eta,1\},3/2)$ so that $$Ran\big(B|_{H_{\Gi}^{2+d}\cap \widetilde \Hi}\big)\subseteq H^{2+d}\cap 
H^2_{\Gi}\cap \widetilde \Hi.$$

	\item When $\Gi$ is equipped with ($\NN$) and $a+\eta\in(0, 7/2)$, there exist $d\in[\max\{a+\eta,2\},7/2)$ and 
$d_1\in(d,7/2)$ such that $$Ran\big(B|_{ H^{d_1}_{\NN\KK}\cap \widetilde \Hi}\big)\subseteq H^{d_1}_{\NN\KK}\cap 
\widetilde 
\Hi,\ \ \ \ \ \ \ Ran\big(B|_{H_{\Gi}^{2+d}\cap \widetilde \Hi}\big)\subseteq H^{2+d}\cap H^{1+d}_{\NN\KK}\cap 
H^2_{\Gi}\cap \widetilde \Hi.$$

	\item When $\Gi$ is equipped with ($\Di$) and $a+\eta\in(0, 5/2)$, there exists $d\in[\max\{a+\eta,1\},5/2)$ 
such that $$Ran\big(B|_{H_{\Gi}^{2+d}\cap \widetilde \Hi}\big)\subseteq H^{2+d}\cap H^{1+d}_{\NN\KK}\cap 
H^2_{\Gi}\cap 
\widetilde \Hi.$$ If $d\geq 2$, then there exists $d_1\in(d,5/2)$ such that $Ran\big(B|_{H^{d_1}\cap \widetilde 
\Hi}\big)\subseteq H^{d_1}\cap \widetilde \Hi.$
	
	\end{enumerate}
\end{assumption2}
\begin{assumption2}[$\eta,a$]
	The couple $(A,B)$ satisfies Assumptions II$(\Phi,\eta,a)$ with $\Phi$ a Hilbert basis of $\Hi$ made by 
eigenfunctions of $A$.
\end{assumption2}	
\vspace{1mm}

We introduce Assumptions II($\upvarphi,\eta,a$) and Assumptions II($\upvarphi,\eta$) since the choice of the 
boundary conditions defining $D(A)$ affects the definition of the spaces $H^s_\Gi=D(|A|^\frac{s}{2})$ with $s>0$. 
For this reason, we have to calibrate the regularity of the control potential $B$ according to such choice.

\subsection{Well-posedness}
Now, we cite $\cite[Theorem\ 4.1]{mio3}$ where the {well-posedness} of the bilinear Schr\"odinger equation 
$(\ref{mainx1})$ is ensured in $H^s_\Gi$ with suitable $s\geq 3$.

\begin{prop}{$\cite[Theorem\ 4.1]{mio3}$}\label{laura}
Let $\Gi$ be a compact quantum graph and $(A,B)$ satisfy Assumptions II$(\eta,\tilde d)$ with $\eta>0$ and $\tilde 
d\geq 0$. For any $T>0$ and $u\in L^2((0,T),\R)$, the flow of the ($\ref{mainx1}$) is unitary in $\Hi$ and, for any 
initial data $\psi^0\in H^{2+d}_{\Gi}$ with $d$ from Assumptions II$(\eta,\tilde d)$, there
exists a unique mild solution of ($\ref{mainx1}$) in
$H^{2+d}_{\Gi}$, i.e. a function $\psi\in C^0([0,T],H^{2+d}_{\Gi})$ such that
\begin{equation}\label{form}
\psi(t,x)=e^{-iAt}\psi^0(x)-
i\int_0^t e^{-iA(t-s)}u(s)B\psi(s,x)ds,\ \ \ \ \ \ \ \ \ \ \ \forall t\in[0,T].
\end{equation}
\end{prop}
\begin{osss}\label{wellene}
Let $\upvarphi:=(\varphi_{k})_{k\in\N^*}\subseteq(\phi_k)_{k\in\N^*}$ be an orthonormal system of $\Hi$ made by 
eigenfunctions of $A$ and $\widetilde\Hi:=\overline{\spn\{\varphi_k\ |\ k\in\N^*\}}^{\ L^2}.$ We notice that the 
statement of Proposition $\ref{laura}$ can be ensured in $\widetilde\Hi$ as the propagator $\G^u_t$ 
preserves the 
space $H^2_{\Gi}\cap\widetilde\Hi$ when $B:H^2_{\Gi}\cap\widetilde\Hi\longrightarrow H^2_{\Gi}\cap\widetilde\Hi$. 
Thus, if $(A,B)$ satisfies Assumptions II$(\upvarphi,\eta,\tilde d)$ with $\eta>0$ and $\tilde d\geq 0$, then, for 
every $\psi^0\in H^{2+d}_{\Gi}\cup\widetilde\Hi$ with $d$ from Assumptions II$(\upvarphi,\eta,\tilde d)$ and $u\in 
L^2((0,T),\R)$, there exists a unique mild solution of the ($\ref{mainx1}$) in $H^{2+d}_{\Gi}\cup\widetilde\Hi.$
\end{osss}

\subsection{Abstract global exact controllability result}\label{globalcon}

In the following theorem, we provide the main abstract result of the work regarding the global exact 
controllability of the $(\ref{mainx1})$.
\begin{teorema}\label{global} \needspace{30mm}
Let $\Gi$ be a compact quantum graph and $(\lambda_k)_{k\in\N^*}$ be the ordered sequence of 
eigenvalues of 
$A$. Let the following hypotheses be satisfied.
\begin{itemize}
\needspace{10mm} \item The eigenvalues $(\lambda_k)_{k\in\N^*}$ are simple. There exists an entire function 
$G\in L^\infty(\R,\R)$  such that there exist 
$J,I>0$ such that $|G(z)|\leq J e^{I|z|}$ for every $z\in\C$ and the numbers $\{\pm\sqrt{\lambda_k}\}_{k\in\N^*}$ 
are simple zeros of $G$. In addition, there exist $\tilde d\geq 0$ and $C>0$ such that 
\begin{equation}\label{stimabasso}\Big|G'\big(\pm\sqrt{\lambda_k}\big)\Big|\geq \frac{C}{k^{1+\tilde d}},\ \  \  \ 
\ \ \ \ \ \ \ \ 
\ \ \ \forall k\in\N^*.\end{equation}
\item The couple $(A,B)$ satisfy Assumptions I$(\eta)$ and Assumptions II$(\eta,\tilde d)$ for $\eta>0$.
	\end{itemize}
	
	\noindent
	Then, the $(\ref{mainx1})$ is globally exactly controllable in $H^{2+d}_{\Gi}$ with $d$ from 
Assumptions II$(\eta,\tilde d)$.
		\end{teorema}
		\begin{proof}
		 See Section \ref{proofglobal}.\qedhere
		\end{proof}

The proof of Theorem \ref{global} is obtained by extending a local exact controllability result which follows from 
the solvability of a suitable moment problem in a specific subspace of $\ell^2(\C)$. In other words, we 
ensure that, 
for every $(x_{k})_{k\in\N^*}$ in such subspace, there exists $u\in L^2((0,T),\R)$ with $T>0$ such that 
\begin{equation}\begin{split}\label{mome_generic}
	{x_{k}}=-i\int_{0}^Tu(\tau)e^{i(\lambda_k-\lambda_1)\tau}d\tau,\ \ \ \ \ \ \ \forall k\in\N^*.
	\\
	\end{split}\end{equation} 
	In the next section, we develop a new technique leading to such result when the hypotheses of Theorem 
\ref{global} are satisfied. In particular, the entire function $G$ is used to construct the control $u$ satisfying 
\eqref{mome_generic} thanks to the Paley-Wiener's Theorem. The lower bound \eqref{stimabasso} provides the 
regularity of the subspace of $\ell^2(\C)$ where we can consider $({x_{k}})_{k\in\N^*}$ in the moment problem 
\eqref{mome_generic}.
	
	\vspace{0.5mm}
	
	In Section \ref{generic_application}, we show how to construct an entire function $G$ satisfying the 
hypothesis of Theorem \ref{global} when $\Gi$ is an appropriate star graph. In such framework, it is possible to 
see the numbers $\{\pm\sqrt{\lambda_k}\}_{k\in\N^*}$ as the zeros of a specific function that we use to 
define $G$.

\begin{oss}
	When $\Gi=(0,1)$,  Ingham's type theorems lead to the solvability of \eqref{mome_generic} for sequences 
$({x_{k}})_{k\in\N^*}\in\ell^2(\C)$. Such techniques are valid thanks to the {spectral gap} 
$\inf_{{k\in\N^*}}|\lambda_{k+1}-\lambda_k|> 0$ as explained in \cite[Appendix\ B]{laurent}. When $\Gi$ is a 
compact graph and the weak spectral gap \eqref{weak_gap} is satisfied with suitable $\tilde d>0$, such result can 
be ensured for $(x_{k})_{k\in\N^*}\in h^{\tilde d(\MM-1)}$ with $\MM$ from \eqref{g13} (see \cite[Appendix\ 
B]{mio3}).
\end{oss}

\section{Trigonometric moment problems}\label{trigonometric}
\subsection{On the solvability of the moment problem}\label{trigonometric_intro}
The aim of this section is to prove the next proposition which ensures the solvability of the moment problem 
\eqref{mome_generic} when the hypotheses of Theorem \ref{global} are valid.

\begin{prop}\label{entire}
	Let $(\lambda_k)_{k\in\N^*}\subset \R^+$ be an ordered sequence of pairwise distinct numbers so that there 
exist $\delta,C_1,C_2>0$ and $\MM\in\N^*$ such 
that $\inf_{{k\in\N^*}}\big|\sqrt{\lambda_{k+\MM}}-\sqrt{\lambda_k}\big|\geq\delta \MM$ and
	\begin{align}\label{ultimo}C_1 k^2 
\leq\lambda_k\leq C_2 k^2,\ \ \ \ \ \ \ \ \ \ \ \forall k\geq 2.\end{align}
Let $G$ be an entire function so that 
$\{\pm\sqrt{\lambda_k}\}_{k\in\N^*}$ are its simple zeros, $G\in L^\infty(\R,\R)$ and there exist $J,I>0$ such that 
$|G(z)|\leq J e^{I|z|}$ for every $z\in\C.$
	If there exist $\tilde d\geq 0$ and $C>0$ such that $\big|G'(\pm\sqrt{\lambda_k})\big|\geq 
\frac{C}{k^{1+\tilde 
d}}$ for every $k\in\N^*,$
	then for $T>2\pi/\delta$ and for every $(x_k)_{k\in\N^*}\in h^{\tilde d}(\N^*,\C)$ with $x_1\in\R$, there 
exists $u\in L^2((0,T),\R)$ so that
\begin{equation}\label{momem}\begin{split}
{x_{k}}=\int_{0}^Tu(\tau)e^{i(\lambda_k-\lambda_1) \tau}d\tau,\ \ \ \ \ \ \forall k\in\N^*.\\
	\end{split}\end{equation}
	\end{prop}
	The proof of Proposition \ref{entire} is provided in Section \ref{proofen} by gathering some technical results 
developed below. We suggest to the uninterested reader to skip it and pass to the Section \ref{proofglobal}.

\begin{osss}\label{entirecoro}
	Let the hypotheses of Proposition \ref{entire} be guaranteed. From Proposition \ref{entire}, for 
$T>2\pi/\delta$ and for every 
$(x_k)_{k\in\N^*}\in h^{\tilde d}(\N^*,\C)$ with $x_1\in\R$, there exists $u\in L^2((0,T),\R)$ such that 
\begin{equation*}\begin{split}
{x_{k}}=\int_{0}^Tu(\tau)e^{-i(\lambda_k-\lambda_1) \tau}d\tau,\ \ \ \ \ \  \ \ \ \ \ \forall k\in\N^*.\\
	\end{split}\end{equation*}
Indeed, for every 
$(x_k)_{k\in\N^*}\in h^{\tilde d}(\N^*,\C)$ with $x_1\in\R$, there exists $u\in L^2((0,T),\R)$ such that $
\overline{x_{k}}=\int_{0}^Tu(\tau)e^{i(\lambda_k-\lambda_1) \tau}d\tau$ for every $k\in\N^*.$
The claim is proved by conjugating the last expression.
\end{osss}

\subsection{Families of functions in a Hilbert space}
Let $\Z^*=\Z\setminus\{0\}$. We denote by $\la \cdot,\cdot\ra_{L^2(0,T)}$ the scalar product in $L^2((0,T),\C)$ 
with $T>0.$
\begin{defi}
	Let $(f_k)_{k\in\Z^*}$ be a family of functions in $L^2((0,T),\C)$ with $T>0$. The family $(f_k)_{k\in\Z^*}$ 
is 
said to be minimal if and only if $f_k\not\in\overline{\spn\{f_j:j\neq k\}}^{\, L^2}$ for every $k\in\Z^*.$ 
\end{defi}
	
	\begin{defi} A biorthogonal family to $(f_k)_{k\in\Z^*}\subset L^2((0,T),\C)$ is a sequence of functions 
$(g_k)_{k\in\Z^*}$ in $L^2((0,T),\C)$ such that $\la f_k, g_j\ra_{L^2(0,T)}=\delta_{k,j}$ for every $k,j\in\Z^*.$
	\end{defi}
	
    \begin{osss}\label{unique_bio}
    When $(f_k)_{k\in\Z^*}$ is minimal, there exists an unique biorthogonal family $(g_k)_{k\in\Z^*}$ to 
$(f_k)_{k\in\Z^*}$ belonging to $X:=\overline{\spn\{f_j:j\in\Z^*\}}^{\, L^2}$. Indeed, $(g_k)_{k\in\Z^*}$ can be 
constructed by setting $$g_k=(f_k-\pi_kf_k)\|f_k-\pi_kf_k\|^{-2}_{L^2(0,T)},\ \ \ \ \ \ \ \forall k\in\Z^*$$ where 
$\pi_k$ is the orthogonal projector onto $\overline{\spn\{f_j:j\neq k\}}^{\, L^2}$. The unicity of 
$(g_k)_{k\in\Z^*}$ follows as, for any biorthogonal family $(g^1_k)_{k\in\Z^*}$ in $X$, we have $\la g_k 
-g_k^1,f_j\ra_{L^2(0,T)}=0$ for every $j,k\in\Z^*$, which implies $g_k 
=g_k^1$ for every $k\in\Z^*.$
       \end{osss}
    
    \begin{osss}\label{minimality}
If a sequence of functions $(f_k)_{k\in\Z^*}\subset L^2((0,T),\C)$ admits a biorthogonal family 
$(g_k)_{k\in\Z^*}$, then it is minimal. Indeed, if we assume that there exists $k\in\Z^*$ such that 
$f_k\in\overline{\spn\{f_j:j\neq k\}}^{\, L^2}$, then the relations $\la 
f_k,g_j\ra_{L^2(0,T)}=0$ for every $j\in\Z^*\setminus\{k\}$ would imply $\la 
f_k,g_k\ra_{L^2(0,T)}=0$ which is absurd.
    \end{osss}

\begin{defi}
Let $(f_k)_{k\in\Z^*}$ be a family of functions in $L^2((0,T),\C)$ with $T>0$. The family $(f_k)_{k\in\Z^*}$ is a 
Riesz basis of $\overline{\spn\{f_j:j\in\Z^*\}}^{\, L^2}$ if and only if it is isomorphic to an orthonormal system.
\end{defi}

\begin{osss}\label{unique_bio_riesz}    
    Let $(f_k)_{k\in\Z^*}$ be a Riesz basis of $X:=\overline{\spn\{f_j : j\in\Z^*\}}^{\, L^2}$. The sequence 
$(f_k)_{k\in\Z^*}$ is also minimal and the biorthogonal family to $(f_k)_{k\in\Z^*}$ can be uniquely defined in $X$ 
thanks to Remark \ref{unique_bio}. The biorthogonal family to $(f_k)_{k\in\Z^*}$ forms a Riesz basis of $X$.
\end{osss}

Now, we provide an important property on the Riesz basis proved in $\cite[Appendix\ B.1]{laurent}$.

\begin{prop}{$\cite[Appendix\ B; Proposition\ 19]{laurent}$}\label{inequality_prop}
\,\,Let $(f_k)_{k\in\Z^*}$ be a family of functions in $L^2((0,T),\C)$ with $T>0$ and let $(f_k)_{k\in\Z^*}$ be a 
Riesz basis of $\overline{\spn\{f_k : k\in\Z^*\}}^{\, L^2}$. There exist $C_1,C_1>0$ such that
\begin{equation*}C_1\|{\bf x}\|^2_{\ell^2}\leq\int_0^{ T}\Big|\sum_{k\in\Z^*}x_kf_k\Big|^2ds\leq 
	C_2\|{\bf x}\|^2_{\ell^2},\ \ \  \ \ \  \ \forall {\bf x}:=(x_k)_{k\in N^*}\in\ell^2(\Z^*,\C).\end{equation*}
 \end{prop}

\subsection{Riesz basis made by divided differences of exponentials}\label{divided_difference}
Let ${\bf {\upnu}}=(\nu_k)_{k\in\Z^*}\subset\R$ be an ordered sequence of pairwise distinct numbers such that 
there exist $\MM\in\N^*$ and $\delta>0$ such that 
\begin{equation}\begin{split}\label{gapp11}
\inf_{\{k\in\Z^*\ :\ k+\MM\neq 0 \}}|\nu_{k+\MM}-\nu_k|\geq\delta \MM.\\
\end{split}\end{equation}\needspace{2mm}
\noindent The last relation yields that it does not exist $\MM$ consecutive $k,k+1\in\Z^*$ such that 
$|\nu_{k+1}  -  \nu_{k}| < \delta$ and then, there exist some $j\in\Z^*\setminus\{-1\}$ such that 
$|\nu_{j+1}-\nu_{j}|\geq\delta$. This fact leads to a partition of $\Z^*$ in subsets $\{E_m\}_{m\in\Z^*}$ that we 
construct as follows. We denote by 
$(l_m)_{m\in\Z^*}\subseteq\Z^*\setminus\{-1\}$ the ordered
sequence of all the numbers such that 
$$|\nu_{l_m+1}-\nu_{l_m}|\geq\delta.$$ We also add the value $-1$ to such 
sequence when $|\nu_{1}-\nu_{-1}|\geq\delta$. We denote by
$\{E_{m}\}_{m\in\Z^*}$ the sets:
\begin{equation}\begin{split}\label{partitionpalle}
E_{-1}=\Big\{k\in \Z^*\ :\ &l_{-1}+1\leq k\leq 
l_{1}\Big\},\\
E_m=\Big\{k\in \Z^*\ :\ l_{m}+1&\leq k\leq l_{m+1}\Big\},\ \ \ \ \ \forall 
m\in\Z^*\setminus\{-1\}.
\end{split}\end{equation}
The partition of $\Z^*$ in subsets $\{E_m\}_{m\in\Z^*}$ also defines an equivalence 
relation in $\Z^*$:
$$k\sim n\text{ if and only if there exists }m\in\Z^*\text{ such that }k,n\in E_m.$$
Now, 
$\{E_m\}_{m\in\Z^*}$ are the equivalence classes corresponding to such relation and $|E_m|\leq \MM-1$ thanks to 
\eqref{gapp11}. Let $s(m)$ 
be the smallest element of $E_m$. For every ${\bf {x}}:=(x_k)_{k\in\Z^*}\subset\C$ and $m\in\Z^*,$ we define 
$${\bf 
{x}}^m:=(x^m_l)_{l\leq |E_m|},\ \ \ \ \  : \ \ \  \ \ x_l^m=x_{s(m)+(l-1)},\ \ \ \ \ \forall l\leq |E_m|.$$
In other words, ${\bf x}^m$ is the vector in $\C^{|E_m|}$ composed by those elements of ${\bf x}$ with indices in 
$E_m$.
For every $m\in\Z^*$, we denote $F_m({\bf {\upnu}}^m):\C^{|E_m|}\rightarrow \C^{|E_m|}$ the matrix with components
\begin{equation*}
\begin{split}
F_{m;j,k}({\bf {\upnu}}^m):=\begin{cases}
\prod_{\underset{ l\leq k}{l\neq j}}(\nu^m_j-\nu_l^m)^{-1},\ \ \ \ \ & j\leq k,\\
1,\ \ \ \ \ \ \ \ \ \ \ \ \ \ \ \ & j=k=1,\\
0,\ \ \ \ \ \ \ \ \ \ \ \ \ \ \ \ & j>k,\\
\end{cases}\ \ \ \ \ \ \ \ \ \ \ \ \ \ \forall j,k\leq {|E_m|}.
\end{split}
\end{equation*}For each $k\in\Z^*$, there exists $m(k)\in\Z^*$ such that $k\in E_{m(k)}$, while $s(m(k))$ 
represents the smallest element of $E_{m(k)}$. Let $F({\bf {\upnu}})$ be the infinite matrix acting on ${\bf 
x}=(x_k)_{k\in\Z^*}\subset\C$ as follows
$$\big(F({\bf {\upnu}}){\bf x}\big)_k=\Big(F_{m(k)}({\bf {\upnu}}^{m(k)}){\bf x}^{m(k)}\Big)_{k-s(m(k))+1},\ \ \ \ 
\  \ \ \ \ \ \forall\, k\in\Z^*.$$
We consider $F({\bf {\upnu}})$ as the operator on $\ell^2(\Z^*,\C)$ defined by the action\ above and with domain
$$H({\bf {\upnu}}):=D(F({\bf {\upnu}}))=\left\{{\bf  x}:=(x_k)_{k\in\Z^*}\in\ell^2(\Z^*,\C)\ :\ F({\bf 
{\upnu}}){\bf  x}\in\ell^2(\Z^*,\C)\right\}.$$

\begin{osss}\label{limitatelo}
Each matrix $F_m({\bf {\upnu}}^m)$ with $m\in\Z^*$ is invertible and we call $F_m({\bf {\upnu}}^m)^{-1}$ its 
inverse. Now, $F({\bf {\upnu}}):H({\bf {\upnu}})\rightarrow Ran(F({\bf {\upnu}}))$ is invertible and $F({\bf 
{\upnu}})^{-1}:Ran(F({\bf {\upnu}}))\rightarrow H({\bf {\upnu}})$ is so that, for ${\bf x}\in Ran(F({\bf 
{\upnu}}))$,
$$\big(F({\bf {\upnu}})^{-1}{\bf x}\big)_k=\Big(F_{m(k)}({\bf {\upnu}}^{m(k)})^{-1}{\bf 
x}^{m(k)}\Big)_{k-s(m(k))+1},\ \ \ \ \  \ \ \ \ \forall k\in\Z^*.$$
\end{osss}

Let $F_{m(k)}({\bf {\upnu}}^{m(k)})^*$ be the transposed matrix of $F_{m(k)}({\bf {\upnu}}^{m(k)})$ for every 
$m\in\Z^*$. Let $F({\bf {\upnu}})^*$ be the infinite matrix so that, for every ${\bf 
x}=(x_k)_{k\in\Z^*}\subset\C$, 
$$\big(F({\bf {\upnu}})^*{\bf x}\big)_k=\Big(F_{m(k)}({\bf {\upnu}}^{m(k)})^*{\bf x}^{m(k)}\Big)_{k-s(m(k))+1},\ \ 
\ \ \  \ \ \ \ \forall k\in\Z^*.$$

\begin{osss}\label{aggiunto}
	When $H({\bf {\upnu}})$ is dense in $\ell^2(\Z^*,\C)$, we can consider $F({\bf {\upnu}})^*$ as the unique 
adjoint operator of $F({\bf {\upnu}})$ in $\ell^2(\Z^*,\C)$ with domain $H({\bf {\upnu}})^*:=D(F({\bf 
{\upnu}})^*)$.
	As in Remark $\ref{limitatelo}$, we define $(F({\bf {\upnu}})^*)^{-1}$ the inverse operator of $F({\bf 
{\upnu}})^*:H({\bf {\upnu}})^*\rightarrow Ran(F({\bf {\upnu}})^*)$ and $(F({\bf {\upnu}})^*)^{-1}=(F({\bf 
{\upnu}})^{-1})^*$.
	\end{osss}

	\vspace{2mm}
Let ${\bf e}$ be the sequence of functions in $L^2((0,T),\C)$ with $T>0$ so that
${\bf e}:=(e^{i\nu_k (\cdot)})_{k\in\Z^*}.$
We denote by ${\bf \Xi}$ the so-called {divided differences} of the family $(e^{i\nu_kt})_{k\in\Z^*}$ such that
\begin{align}\label{xi_riesz}{\bf \Xi}:=(\xi_k)_{k\in\Z^*}=F({\bf {\upnu}})^*{\bf e}.\end{align}

In the following theorem, we rephrase a result of Avdonin and Moran $\cite{avdd}$, which is also proved by 
Baiocchi, Komornik and Loreti in $\cite{balocchi}$. 
\begin{teorema}[Theorem\ 3.29; $\cite{wave}$]\label{riesz}
	Let $(\nu_k)_{k\in\Z^*}$ be an ordered sequence of pairwise distinct real numbers satisfying $(\ref{gapp11})$. 
If $T>2\pi/\delta$, then $(\xi_k)_{k\in\Z^*}$ forms a Riesz Basis in the space $\overline{\spn\{\xi_k:k\in\Z^*\}}^{ 
L^2}.$
\end{teorema}

\begin{osss}\label{minimality1}
Let Theorem \ref{riesz} be valid. As ${\bf \Xi}$ is a Riesz basis, it is minimal in 
$\overline{\spn\{\xi_k:k\in\Z^*\}}^{ L^2}$ and it admits a biorthogonal family ${\bf u}:=(u_k)_{k\in\Z^*}$ thanks 
to Remark \ref{minimality}. Now, it is possible to see that $$\la u_k,\xi_j\ra_{L^2(0,T)}=\big\la u_k,\big(F({\bf 
{\upnu}})^*{\bf e}\big)_j\big\ra_{L^2(0,T)}=\big\la\big(F({\bf {\upnu}}){\bf 
u}\big)_k,e^{i\nu_jt}\big\ra_{L^2(0,T)},\ \ \ \ \ \forall 
j,k\in\Z^*.$$
The last relation ensures that $F({\bf {\upnu}}){\bf u}$ is a  biorthogonal family to ${\bf e}$, which is then 
minimal.
\end{osss}
	
	\subsection{Auxiliary results}
	In this subsection, we provide few intermediate results required in the proof of Proposition \ref{entire}. 
	
	\noindent
	In Lemma \ref{linden} and Lemma \ref{biorthogonal_estimation}, we consider a suitable sequence of 
numbers ${\bf 
\upnu}:=(\nu_k)_{k\in\Z^*}$. We construct and characterize a biorthogonal family to 
$\{e^{i\nu_kt}\}_{k\in\Z^*}$ by using a function $G$ defined as in Proposition \ref{entire}.

	\noindent
In Lemma \ref{entire1}, the previous results lead to specific estimations on $ F_{m}({\bf \upnu}^m)$ with 
$m\in\Z^*$.

\noindent
In Lemma \ref{entire2} and Lemma \ref{entire3}, we consider a sequence ${\bf \Theta}:=(\theta_k)_{k\in\Z^*}$ and 
another one defined as ${\bf 
\upnu}:=(\nu_k)_{k\in\Z^*}=\big(sgn{(\theta_k)}\sqrt{|\theta_k|}\big)_{k\in\Z^*}$. We use the estimations on $ 
F_{m}({\bf \upnu}^m)$ with 
$m\in\Z^*$ provided by 
Lemma 
\ref{entire1} in order to study $ F_{m}({\bf \Theta}^m)$ 
with $m\in\Z^*$ and on the domain of the operator $ F({\bf \Theta})$.

\noindent
In the proof of Proposition \ref{entire}, we respectively denote by ${\bf \Theta}$ and ${\bf \upnu}$ the sequences 
obtained by reordering $\{\pm\lambda_{k}\}_{k\in\N^*}$ and $\{\pm\sqrt{\lambda_{k}}\}_{k\in\N^*}$ and we use Lemma 
\ref{entire3} to ensure the statement.

	\begin{lemma}\label{linden}
	Let ${\bf {\upnu}}:=(\nu_k)_{k\in\Z^*}$ be an ordered sequence of pairwise distinct real numbers. Let $G$ be an 
entire function such that $G\in L^\infty(\R,\R)$. Let exist $J,I>0$ such that $|G(z)|\leq J 
e^{I|z|}$ for 
every $z\in\C.$
	Denoted $G_k(z):={G(z)}{(z-\nu_k)^{-1}}$ with $z\in\C$ for every $k\in \Z^*$, there exists $C>0$ such that
	\begin{equation*} \|G_k\|_{L^2(\R,\R)}\leq C,\ \ \ \ \  \forall k\in\Z^*.\end{equation*}
	\end{lemma}
	\begin{proof}
	 We know that there exists $M>0$ so that $|G(x)|\leq M$ for every $x\in \R$. It implies that, for every 
$k\in\Z^*$, there exists $C_1>0$ (not depending on $k$) so that 
$$\|G_k\|_{L^2(\R)}^ 2=\int_\R \overline{G_k(x)}G_k(x)\,dx=\int_\R \frac{\overline{G(x)}G(x)}{(x-\nu_k)^ 2}\,dx\leq 
\int_{|x-\nu_k|\leq 1}{\overline{G(x)}G(x)}{(x-\nu_k)^{- 2}}\,dx+M^2C_1.$$
Now, from $\cite[p.\ 82;\ Theorem\ 11]{giovane}$, we have $|G(x+iy)|\leq Me^ {I |y|}$ for $x,y\in\R$. The Cauchy 
Integral Theorem ensures the existence of $C_2>0$ (not depending on $k$) so that $$\int_{|x-\nu_k|\leq 
1}\frac{\overline{G(x)}G(x)}{(x-\nu_k)^2}dx\leq\int_0^\pi\big|\overline{G(\overline{\nu_k+e^ {i\theta}})}G(\nu_k+e^ 
{i\theta})\big|d\theta\leq M^2 \int_0^\pi e^{2I\sin(\theta)}\,d\theta\leq 
M^2C_2.\qedhere$$
	\end{proof}

	\begin{lemma}\label{biorthogonal_estimation}
	Let ${\bf {\upnu}}:=(\nu_k)_{k\in\Z^*}$ be an ordered sequence of pairwise distinct real numbers satisfying 
$(\ref{gapp11})$ with $\delta>0$. Let $G$ be an entire function such that $G\in L^\infty(\R,\R)$. Let exist 
$J,I>0$ such that $|G(z)|\leq J e^{I|z|}$ for every $z\in\C.$ 
	Let $(\xi_k)_{k\in\Z^*}\subset L^2((0,T),\C)$ be defined in \eqref{xi_riesz} for 
${T}>\max\{2\pi/\delta,2I\}$. If 
$({\nu_k})_{k\in\Z^*}$ are simple zeros of $G$ such that there exist $\tilde d\geq 0$, $C>0$ such that 
\begin{equation}\label{11x}\big|G'(\nu_k)\big|\geq \frac{C}{|k|^{1+\tilde d}},\ \ \ \  \ \ \ \ \  \ \ \forall 
k\in\Z^*,\end{equation}
	then there exists $(w_k)_{k\in\Z^*}$ an unique biorthogonal family to $(e^{i\nu_kt })_{k\in\Z^*}$ in 
$\overline{\spn\{\xi_k:k\in\Z^*\}}^{ L^2}$ satisfying the following property. There exists $C_1>0$ such that 
$\|w_k\|_{L^2(0,T)}\leq{C_1}{|k|^{1+\tilde d}}$ for every $k\in\Z^*.$
	\end{lemma}
\begin{proof}
For every $k\in\Z^*$, we define $G_k(z):={G(z)}{(z-\nu_k)^{-1}}.$
Thanks to the Paley-Wiener's Theorem $\cite[Theorem\ 3.19]{wave}$, for every $k\in\Z^*$, there exists $f_k\in 
L^2(\R,\R)$ with support in $[-I, I]$ such that
$$G_k(z)=\int_{-I}^{I}e^{izs}f_k(s)ds=\int_{-T/2}^{T/2}e^{izs}f_k(s)ds=\int_0^{{T}}e^{izt}e^{-iz\frac{{T}}{2}}
f_k(t-T/2)dt.$$
	For $j,k\in\Z^*$ and $c_k:=G'(\nu_k)$, we name $v_k(t):=e^{i\nu_k\frac{{T}}{2}}\overline{f_k}(t-T/2)$. Since 
$G_k(\nu_j)=\delta_{k,j}G'(\nu_k),$
$$\la v_k, e^{i\nu_j(\cdot)}\ra_{L^2(0, T)}=G_k(\nu_j)=\delta_{k,j}G'(\nu_k)=\delta_{k,j}c_k.$$
	Thus, the sequence $(v_k)_{k\in\Z^*}$ is biorthogonal to $(c_k^{-1}e^{i\nu_k(\cdot)})_{k\in\Z^*}$. 
Thanks to the Plancherel's identity and to Lemma \ref{linden}, there exists $C_1>0$ such that 
	\begin{equation}\label{22x} \|v_k\|_{L^2(0, T)}=\|G_k\|_{L^2(\R,\R)}\leq C_1,\ \ \ \ \  \forall 
k\in\Z^*.\end{equation}
The family $(e^{i\nu_k(\cdot)})_{k\in\Z^*}$ is minimal in $X:=\overline{\spn\{\xi_k:k\in\Z^*\}}^{ L^2}$ thanks to 
Remark \ref{minimality1} and then, $(c_k^{-1}e^{i\nu_k(\cdot)})_{k\in\Z^*}$ is minimal in $X$. Defined $\pi_X$ 
the orthogonal projector onto $X$, we see that, for every $j,k\in\Z^*$,
$$\la \pi_Xv_k, c_j^{-1}e^{i\nu_j(\cdot)}\ra_{L^2(0, T)}=\la v_k, c_j^{-1}e^{i\nu_j(\cdot)}\ra_{L^2(0, 
T)}=\delta_{k,j}.$$	
The last relation and Remark \ref{unique_bio} imply that $(\pi_X v_k)_{k\in\Z^*}$ is the unique biorthogonal 
family to $(c_k^{-1}e^{i\nu_k(\cdot)})_{k\in\Z^*}$ in $X$. We denote $w_k=(\overline{c_k})^{-1}\pi_X v_k$ for every 
$k\in\Z^*$ and $(w_k)_{k\in\Z^*}$ is the unique biorthogonal family to $(e^{i\nu_k(\cdot)})_{k\in\Z^*}$ in $X$. 
In conclusion, thanks to \eqref{11x} and \eqref{22x}, there exists $C_2>0$ such that $$\|w_k\|_{L^2(0,T)}\leq 
\|\pi_Xv_k\|_{L^2(0,T)} |c_k|^{-1}\leq \|v_k\|_{L^2(0,T)} |G'(\nu_k)|^{-1}\leq C_2 |k|^{1+\tilde d},\ \ \ \ \ \ 
\forall k\in\Z^* .\qedhere$$
\end{proof}

	\begin{lemma}\label{entire1}
	Let ${\bf {\upnu}}:=(\nu_k)_{k\in\Z^*}$ be an ordered sequence of pairwise distinct real numbers satisfying 
$(\ref{gapp11})$. Let $G$ be an entire function such that $G\in L^\infty(\R,\R)$. Let exist $J,I>0$ such 
that 
$|G(z)|\leq J e^{I|z|}$ for every $z\in\C.$ If $({\nu_k})_{k\in\Z^*}$ are simple zeros of $G$ such that there exist 
$\tilde d\geq 0$, $C>0$ such that 
$|G'(\nu_k)|\geq \frac{C}{|k|^{1+\tilde d}}$ for every $k\in\Z^*$, then there exists $C>0$ so that 
$$Tr\Big(F_{m}({\bf \upnu}^m)^*F_{m}({\bf \upnu}^m)\Big)\leq C\min\{|l|\in E_m\}^{2(1+{\tilde d})},\ \ \ \ \ \ 
\forall m\in\Z^*$$
where the matrices $ F_{m}({\bf \upnu}^m)$ are defined in Section \ref{divided_difference}.
\end{lemma}

\begin{proof}
Let ${T}>\max(2\pi/\delta,2I)$ with $\delta>0$ from \eqref{gapp11}. Thanks to Theorem $\ref{riesz}$, the 
sequence of functions ${\bf \Xi}=(\xi_{k})_{k\in\Z^*}:=(F({\bf 
\upnu})^*{\bf  e})$ forms a Riesz basis of $X:=\overline{\spn\{\xi_k\ :\ k\in\Z^*\}}^{ L^2}.$
	We call ${\bf {u}}:=( u_k)_{k\in\Z^*}$ the corresponding biorthogonal sequence to ${\bf  \Xi}$ in $X$, which 
is 
also a Riesz basis of $X$ (see Remark \ref{unique_bio_riesz}). From Remark $\ref{limitatelo}$, the matrix $F({\bf 
\upnu})$ is invertible and $(F({\bf \upnu})^*)^{-1}=(F({\bf \upnu})^{-1})^*$ thanks to Remark \ref{aggiunto}. 
	Let ${\bf w}=(w_k)_{k\in\Z^*}$ be the biorthogonal family in $X$ to $(e^{i\nu_k(\cdot)})_{k\in\Z^*}$ defined 
in 
Lemma \ref{biorthogonal_estimation}. For every $j,k\in\Z^*$, $$\delta_{k,j}=\la w_k,e^{i\nu_j(\cdot)}\ra_{L^2(0, 
T)}=\big\la w_k,\big((F({\bf {\upnu}})^*)^{-1}{\bf \Xi}\big)_j\big\ra_{L^2(0, T)}
	=\big\la \big(F({\bf {\upnu}})^{-1}{\bf w}\big)_k,\xi_j\big\ra_{L^2(0, T)},$$ implying ${\bf w}=F({\bf 
\upnu}){\bf {u}}$.
	Let $f(t)=\sum_{k\in\Z^*} \xi_k x_k$ with ${\bf x}:=(x_k)_{k\in\Z^*}\in\ell^2(\Z^*,\C)$. As ${\bf \Xi}$ and 
${\bf u}$ are reciprocally biorthogonal, there holds $x_{k}=\la {u}_k,f\ra_{L^2(0, T)}$ for every $k\in \Z^*$. 
From Proposition \ref{inequality_prop}, there exist $C_1,C_2>0$ such that
\begin{equation}\label{ccc1}C_1\|{\bf x}\|^2_{\ell^2}\leq\|f\|_{L^2(0,T)}^2\leq 
	C_2\|{\bf x}\|^2_{\ell^2}.\end{equation}
	 For $k\in\Z^*$, we call $m(k)\in\Z^*$ the number such that $k\in E_{m(k)}$. Thanks to Lemma 
\ref{biorthogonal_estimation} and $(\ref{ccc1})$, there exist $C_3,C_4>0$ such that, for every $k\in \Z^*$, we 
have 
	\begin{equation}\label{casinaccio}\begin{split}
	\Big|\big(F({\bf \upnu}){\bf x}\big)_k\Big|&=\Big|\big(F({\bf \upnu})\big(\la {u}_l,f\ra_{L^2(0, 
T)}\big)_{l\in\Z^*}\big)_k\Big|=|\la  
w_k,f\ra_{L^2(0, T)}|\leq {\|w_k\|_{L^2(0, T)}\|f\|_{L^2(0, T)}}\\
&\leq C_2^\frac{1}{2}\|w_k\|_{L^2(0, T)}{\|{\bf x}\|_{\ell^2}}\leq C_3|k|^{1+\widetilde{d}}{\|{\bf 
x}\|_{\ell^2}}\leq 
C_4\min_{l\in E_{m(k)}}|l|^{1+\widetilde{d}}{\|{\bf x}\|_{\ell^2}}.\\
	\end{split}\end{equation}
	 For every $m\in\Z^*$, we denote by $s(m)$ the smallest element of $E_m$ and, for every $j\leq |E_m|$, we 
consider $${\bf x}\in\ell^2(\Z^*,\C) \ \ \  \ :\ \ \ \ x_{s(m)+j-1}=1,\ \ \  \ \ \ x_l=0,\ \ \  \ \ \forall 
l\in\Z^*\setminus\{ s(m)+j-1\}.$$ For every $j,n\leq |E_m|$, we use the sequence ${\bf x}$ in the identity 
\eqref{casinaccio} with $k=s(m)+n-1$ and we obtain $|(F_{m;n,j}({\bf \upnu}^m))|\leq C_4\min_{l\in 
E_m}|l|^{1+\widetilde{d}}$, which leads to the statement.  \qedhere
\end{proof}

\begin{osss}\label{double_gap}
Let ${\bf \Theta}:=(\theta_k)_{k\in\Z^*}$ be an ordered sequence of pairwise distinct real numbers and let the 
sequence ${\bf 
\upnu}:=(\nu_k)_{k\in\Z^*}=
	\big(sgn{(\theta_k)}\sqrt{|\theta_k|}\big)_{k\in\Z^*}$ satisfy the identity $(\ref{gapp11})$ with $\MM\in\N^*$ 
and $\delta>0$. As $\inf_{\underset{k+\MM\neq 
0}{k\in\Z^*}}|\nu_{k+\MM}-\nu_k|\geq\delta\MM\min_{\underset{\nu_k\neq 0}{k\in\Z^*}}(|\nu_k|,1)$, we have 
$$\underset{\underset{k+\MM\neq 0}{k\in\Z^*}}{\inf}|\theta_{k+\MM}-\theta_k|=\underset{\underset{k+\MM\neq 
0}{k\in\Z^*}}{\inf}\big||\nu_{k+\MM}|-|\nu_{k}|\big|\big||\nu_{k+\MM}|+|\nu_{k}|\big|\geq\underset{\underset{
\nu_k\neq 0}{k\in\Z^*}}{\min}(|\nu_k|,1)\delta\MM.$$ Now, both the sequences ${\bf {\Theta}}$ and ${\bf 
\upnu}$ satisfy $(\ref{gapp11})$ with respect to the same 
$\delta':=\min_{\overset{k\in\Z^*}{\nu_k\neq 0}}\{|\nu_k|,1\}\delta$ and $\MM$. This fact ensures that ${\bf 
{\Theta}}$ and ${\bf \upnu}$ induce the definition of the same equivalence classes $\{E_m\}_{m\in\Z^*}$ of $\Z^*$ 
introduced in the relations \eqref{partitionpalle}. Thus, we can use $\{E_m\}_{m\in\Z^*}$ in order to define the 
matrices $F_m({\bf \Theta}^m)$ and $F_m({\bf \upnu}^m)$ for every $m\in\Z^*$ and the operators $F({\bf \Theta})$ 
and $F({\bf \upnu})$ on $\ell^2(\Z^*,\C)$.
\end{osss}

\begin{lemma}\label{entire2}
	Let ${\bf \Theta}:=(\theta_k)_{k\in\Z^*}$ be an ordered sequence of pairwise distinct real numbers such that 
$(\nu_k)_{k\in\Z^*}=
	\big(sgn{(\theta_k)}\sqrt{|\theta_k|}\big)_{k\in\Z^*}$ satisfies $(\ref{gapp11})$. Let exist $C_1,C_2>0 $ such 
that
	\begin{equation}\label{ricorda} C_1 |k| \leq|\nu_k|\leq C_2 |k|,\ \ \ \ \ \ \ \ \ \ \ \forall k\in\Z^*,\ \ 
\nu_k\neq 0.\end{equation}
	Let $G$ be an entire function so that $(\nu_k)_{k\in\Z^*}$ are its simple zeros, $G\in L^\infty(\R,\R)$ and 
there exist $J,I>0$ such that $|G(z)|\leq J e^{I|z|}$ for every $z\in\C.$ If there exist $\tilde d\geq 0$ and $C>0$ 
such that $|G'(\nu_k)|\geq \frac{C}{|k|^{1+\tilde d}}$ for every $k\in\Z^*$, then there exists $C>0$ so that 
$$Tr\Big(F_{m}({\bf \Theta}^m)^*F_{m}({\bf \Theta}^m)\Big)\leq C\min\{|l|\in E_m\}^{2{\tilde d}},\ \ \ \ \ \forall 
m\in\Z^*$$ where the matrices $ F_{m}({\bf \Theta}^m)$ are defined as in Remark \ref{double_gap} from the sequence 
${\bf \Theta}$.
\end{lemma}

\begin{proof} 
	We notice $|\theta_{l}-\theta_k|\geq\min\{|\nu_l|,|\nu_{k}|\}|\nu_{l}-\nu_{k}|$ for every $l,k\in \Z^*.$
	Let $m\in\Z^*$ and $I\subseteq E_m$ so that $I\neq \emptyset
	$. Now, $|I|\leq|E_m|\leq\MM-1$ and
	\begin{equation*}\begin{split}&\prod_{j,k\in I}|\theta_k-\theta_j|\geq\underset{\underset{\nu_l\neq 0}{l\in 
I}}{\min} |\nu_l|^{|I|}\prod_{j,k\in I}|\nu_k-\nu_j|\geq C_1\underset{\underset{\nu_l\neq 0}{l\in I}}{\min}\ 
|\nu_l|\prod_{j,k\in I}|\nu_k-\nu_j|\\
		\end{split}\end{equation*}
for $C_1=\min_{\underset{\nu_l\neq 0}{l\in \Z^*}}( |\nu_l|^{\MM-2},1)$. Thus, there exists $C_2>0$ so that, for 
every $m$ and $j,k\in E_m$, we have $$\big|F_{m;j,k}({\bf {\bf {\Theta}}}^m)\big|\leq C_2\big|F_{m;j,k}({\bf 
\upnu}^m)\big|\min\big\{|\nu_l|^{-1}:\ l\in E_m,\ \nu_l\neq 0\big\}.$$ In conclusion, thanks to $(\ref{ricorda})$ 
and Lemma 
$\ref{entire1}$, there exists $C_3>0$ such that
\begin{equation*}\begin{split}
	Tr\Big(F_{m}({\bf {\bf {\Theta}}}^m)^*F_{m}({\bf {\bf {\Theta}}}^m)\Big)&\leq 
C_2^2\underset{\underset{\nu_l\neq 0}{l\in E_m}}{\min}|\nu_l|^{-2}\ {Tr\Big(F_{m}({\bf \upnu}^m)^*F_{m}({\bf 
\upnu}^m)\Big)}\leq C_3\underset{{l\in E_m}}{\min}|l|^{2{\tilde d}}.\qedhere
	\end{split}\end{equation*}
\end{proof}

\begin{lemma}\label{entire3}
	Let the hypotheses of Lemma \ref{entire2} be verified. Then, there holds $H({\bf {\bf \Theta}})\subseteq 
h^{\tilde d}(\Z^*,\C)$
	where $H({\bf {\bf \Theta}})$ is the domain of the operator $F({\bf {\bf \Theta}})$ defined as in Remark 
\ref{double_gap} from the sequence ${\bf \Theta}$.
\end{lemma}
\begin{proof}  Let $\rho(M)$ be the spectral radius of a matrix M and let $\iii M\iii=\sqrt{\rho(M^*M)}$ be its 
euclidean norm. We consider Lemma \ref{entire2} and, since $\big(F_{m}({\bf 
{\Theta}}^m)^*F_{m}({\bf {\Theta}}^m)\big)$ is positive-definite, there exists $C>0$ such that $$\iii F_m({\bf 
{\Theta}}^m)\iii^2=\rho\big(F_{m}({\bf {\Theta}}^m)^*F_{m}({\bf {\Theta}}^m)\big)\leq Tr\Big(F_{m}({\bf {\bf 
{\Theta}}}^m)^*F_{m}({\bf {\bf {\Theta}}}^m)\Big)\leq
C\underset{{l\in E_m}}{\min}|l|^{2{\tilde d}},\ \ \ \ \ \ \ m\in\Z^*.$$
	In conclusion, we obtain $h^{\tilde{d}}(\Z^*,\C)\subset H({\bf {\Theta}})$ since, for every ${\bf 
x}=(x_k)_{k\in\Z^*}\in h^{{\tilde d}}(\Z^*,\C)$,
	\begin{equation*}\begin{split}
	&\|F({\bf {\Theta}})
	{\bf x}\|_{\ell^2}^2\leq\sum_{m\in \Z^*}\iii F_m({\bf {\Theta}}^m)\iii
	^2\sum_{l\in E_m} |x_l|^2\leq C\sum_{m\in \Z^*}\underset{{l\in E_m}}{\min}|l|^{2{\tilde d}}
	\sum_{l\in E_m} |x_l|^2\leq C\|{\bf x}\|_{h^{\tilde d}}^2<+\infty.\qedhere
	\end{split}
	\end{equation*}
	
\end{proof}
\begin{osss}\label{nonme}
When Lemma $\ref{entire3}$ is satisfied with respect to the sequence ${\bf {\Theta}}=(\theta_k)_{k\in\Z^*}$ with 
$\tilde 
d\geq 0$, we have $H({\bf {\Theta}})\supseteq h^{\tilde d}(\Z^*,\C)$ that is dense in $\ell^2(\Z^*,\C)$. Thanks to 
Remark $\ref{aggiunto}$, we consider $F({\bf {\Theta}})^*$ as the unique adjoint operator of $F({\bf {\Theta}})$. 
As $Tr(F_{m}({\bf {\bf {\Theta}}}^m)^*F_{m}({\bf {\bf {\Theta}}}^m))=Tr(F_{m}({\bf {\bf {\Theta}}}^m)F_{m}({\bf 
{\bf {\Theta}}}^m)^*)$ for every $m\in\Z^*$, the techniques from the proof of Lemma $\ref{entire2}$ 
lead to $H({\bf {\Theta}})^*=D(F({\bf {\Theta}})^*)\supseteq h^{\tilde d}(\Z^*,\C)$.
\end{osss}

	\subsection{Proof of Proposition \ref{entire}}\label{proofen}
	\begin{proof}[Proof of Proposition \ref{entire}]Let us introduce the following sequences:
	$${\bf \Theta}:=(\theta_k)_{k\in\Z^*}\ \ \ \ \  \ :\ \ \ \  \ \theta_k=-\lambda_k,\ \ \ \ \ \forall k>0;\ \ \ 
\ 
\  \ \theta_k=\lambda_{-k},\ \ \ \ \ \forall k< 0;$$
	$${\bf {\upnu}}:=(\nu_k)_{k\in\Z^*}\ \ \ \ \  \ :\ \ \ \  \ \nu_k=-\sqrt{\lambda_k},\ \ \ \ \ \forall k>0;\ \ \ 
\ \  \ \nu_k=\sqrt{\lambda_{-k}},\ \ \ \ \ \forall k< 0;$$
	$${\bf \upalpha}:=(\alpha_k)_{k\in\Z^*\setminus\{-1\}}\ \ \ \ \ \ \ :\ \ \ \ \ \ 
\alpha_k=-\lambda_k+\lambda_1,\ \ \ \ \forall  k>0; \ \ \ \ \ \ \alpha_k=\lambda_{-k}-\lambda_1,\ \ \ \ \forall k< 
-1.$$
	We consider $\MM'\in\N^*$ and $\delta'>0$ so that ${\bf \Theta}$ and ${\bf {\upnu}}$ satisfy $(\ref{gapp11})$ 
with respect to $\MM'$ and $\delta'$ (as in Remark \ref{double_gap}), while
	\begin{equation}\label{diocan}\inf_{\{k\in\Z^*\setminus\{-1\}\ :\ k+\MM'\in\Z^*\setminus\{-1\} 
\}}|\alpha_{k+\MM'}-\alpha_k|\geq\delta' \MM'.\end{equation} Let $\{E_m\}_{m\in\Z^*}$ be the equivalence classes 
in 
$\Z^*$ defined by ${\bf \Theta}$ and ${\bf {\upnu}}$. Let us assume that 
$-1\in 
E_{-1}$. Now, $\{E_m\}_{m\in\Z^*\setminus\{-1\}}\cup \{E_{-1}\setminus\{-1\}\}$ are the equivalence classes in 
$\Z^*\setminus\{-1\}$ defined by $(\ref{diocan})$. 
Thanks to \eqref{ultimo}, the sequence ${\bf \upnu}$ satisfies \eqref{ricorda} and
Remark $\ref{nonme}$ is valid with respect to the sequences 
${\bf \Theta}$ and 
${\bf {\upnu}}$. Thus, $H({\bf {\Theta}})^*\supseteq h^{\tilde d}(\Z^*,\C).$
	We define as in Section \ref{divided_difference}, the operator $F({\bf {\upalpha}})$ in 
$\ell^2(\Z^*\setminus\{-1\},\C)$ from the sequence ${\bf \upalpha}$ and we notice that, for every $m\neq -1$, we 
have $F_{m}({\bf {\upalpha}}^m)=F_{m}({\bf {\Theta}}^m)$ and $F_{m}({\bf {\upalpha}}^m)^*=F_{m}({\bf 
{\Theta}}^m)^*$. Thus, as in Lemma \ref{entire3} and Remark $\ref{nonme}$, there hold
	$$H({\bf {\upalpha}})\supseteq h^{\tilde d}\big(\Z^*\setminus\{-1\},\C\big),\ \ \ \ \ \ \ \ \ \ \ \ \  H({\bf 
{\upalpha}})^*\supseteq h^{\tilde d}\big(\Z^*\setminus\{-1\},\C\big).$$
		We define ${\bf e}:=(e^{i\alpha_k (\cdot)})_{k\in\Z^*\setminus\{-1\}}$ and ${\bf 
\Xi}:=(\xi_k)_{k\in\Z^*\setminus\{-1\}}=F({\bf \upalpha})^*{\bf e}.$ When $T>2\pi/\delta'$, Theorem $\ref{riesz}$ 
ensures that $(\xi_k)_{k\in\Z^*\setminus\{-1\}}$ is a Riesz Basis of the space 
$X:=\overline{\spn\{\xi_k:k\in\Z^*\setminus\{-1\}\}}^{ L^2}.$ Now, the map $$M:g\in X\longmapsto \big(\la 
\xi_k,g\ra_{L^2(0,T)}\big)_{k\in\Z^*\setminus\{-1\}}\in \ell^2\big(\Z^*\setminus\{-1\},\C\big)$$ is invertible 
thanks to 
Proposition \ref{inequality_prop}. 
	Denoted $\tilde{X}:=M^{-1}\circ F({\bf \upalpha})^*\big( h^{\widetilde d}(\Z^*\setminus\{-1\},\C)\big)$, the 
map $$\big(F({\bf \upalpha})^*\big)^{-1}\circ M:g\in \tilde{X}\longmapsto \big(\la {\bf 
e},g\ra_{L^2(0,T)}\big)_{k\in\Z^*\setminus\{-1\}}\in h^{\tilde d}\big(\Z^*\setminus\{-1\},\C\big)$$ is invertible. 
Thus, for 
every $(x_k)_{k\in\Z^*\setminus\{-1\}}\in h^{\tilde d}\big(\Z^*\setminus\{-1\},\C\big)$, there exists $u\in 
X\subseteq L^2((0,T),\C)$ 
such that 
\begin{equation}\label{pie}{x_{k}}=\int_{0}^Tu(\tau)e^{-i\alpha_k \tau}d\tau,\ \ \ \ \ \ \ \forall 
k\in\Z^*\setminus\{-1\}.\end{equation}
The last relation can be used to ensure the solvability of the moment problem \eqref{momem} for 
$u\in L^2((0,T),\C)$. So, we need to proceed as follows in order to ensure the result for a function $u\in 
L^2((0,T),\R).$

Given ${\bf x}=(x_k)_{k\in\N^*}\in h^{\tilde d}(\N^*,\C)$ so that $x_1\in\R$, we introduce $(\tilde x_k)_{k\in 
\Z^*\setminus\{-1\}}\in 
h^{\tilde d}(\Z^*\setminus\{-1\},\C)$ so that $\tilde x_k=x_{k}$ for $k> 0$, while $\tilde x_k=\overline x_{-k}$ 
for $k<-1$. Thanks to $(\ref{pie})$ and to the 
definition of ${\bf \upalpha}$, there exists $u\in 
X\subseteq L^2((0,T),\C)$ so that
\begin{equation}\label{SIAM}\begin{split}
	\int_0^Tu(s)e^{i(\lambda_k-\lambda_1) s}ds=x_k=\int_0^T\overline{u}(s)e^{ i(\lambda_k-\lambda_1) s}ds,\ \ \ \ \ 
\ \ \ & k\in\N^*\setminus\{1\},\\
\int_0^T\overline{u}(s)e^{-i(\lambda_k-\lambda_1) s}ds=\overline{x_k}=\int_0^Tu(s)e^{ -i(\lambda_k-\lambda_1) 
s}ds,\ \ \ \ \ 
\ \ \ & k\in\N^*\setminus\{1\},\\
\end{split}\end{equation}
while $\int_0^Tu(s)ds=x_1.$ The relations \eqref{SIAM} and the fact that $x_{1}\in\R$ imply $\la 
\Imm(u),e^{i\alpha_k(\cdot)}\ra_{L^2(0,T)}=0$ for every $k\in\Z^*\setminus\{-1\}$ and then
\begin{align}\label{immu}\la \Imm(u),\xi_k\ra_{L^2(0,T)}=\Big(F({\bf \upalpha})^*\big(\la 
\Imm(u),e^{i\alpha_l(\cdot)}\ra_{L^2(0,T)}\big)_{l\in\Z^*\setminus\{-1\}}\Big)_k=0,\ \ \ \  \ \ \forall 
k\in\Z^*\setminus\{-1\}.\end{align}
From Remark \ref{minimality1}, the family $\{e^{i\alpha_kt}\}_{k\in\Z^*\setminus\{-1\}}$ is minimal in $X$ and 
$X=\overline{\spn\{e^{i\alpha_kt}:k\in\Z^*\setminus\{-1\}\}}^{ L^2}.$ 
Now, we recall that $-\alpha_k=\alpha_{-k}$ for every $k\in\Z^*\setminus\{\pm 1\}$ and $\alpha_1=0$. For every 
$u\in X$, we have
$$\overline{u}\in \overline{\spn\{e^{-i\alpha_kt}:k\in\Z^*\setminus\{-1\}\}}^{ 
L^2}=\overline{\spn\{e^{i\alpha_kt}:k\in\Z^*\setminus\{-1\}\}}^{ 
L^2}=X,$$
which implies that $\Imm(u)=\frac{u-\overline{u}}{2i}\in X$. We call ${\bf {v}}:=( v_k)_{k\in\Z^*\setminus\{-1\}}$ 
the biorthogonal sequence to ${\bf  \Xi}$ in $X$ which is 
also a Riesz basis of $X$ (see Remark \ref{unique_bio_riesz}). Now, for every $u\in X$, we have 
$\Imm(u)\in X$ and \begin{align}\label{decomposizione}\Imm(u)=\sum_{k\in\Z^*\setminus\{-1\}}v_k c_{k},\ \ \ \  
\ \text{for}\ \ \ \ \ c_k=\la \Imm(u),\xi_k\ra_{L^2(0,T)},\ \ \ \ \forall k\in\Z^*\setminus\{-1\}.\end{align} In 
conclusion, when $u$ satisfies
\eqref{SIAM}, the identities \eqref{immu} and \eqref{decomposizione} yield that $\Imm(u)=0$ and then $u$ is real. 
\qedhere
\end{proof}
\begin{osss}\label{lowerlevel} 
The hypotheses on the function $G$ in 
Proposition \ref{entire} can be rewritten in terms of the sequence $\{\pm{\lambda_{k}}\}_{k\in\N^*}$ 
rather than
$\{\pm\sqrt{\lambda_{k}}\}_{k\in\N^*}$. In this case, one can ensure the solvability of the moment 
problem \eqref{momem} for sequences in $h^{1+\tilde 
d}(\N^*,\C)$ for every $T>0$. Indeed, if Lemma \ref{entire1} is valid for ${\bf 
\upnu}=\{\pm{\lambda_{k}}\}_{k\in\N^*}$, then $H(\upnu)^*\supseteq h^{1+\tilde 
d}(\N^*,\C)$ as for Lemma \ref{entire3} and the solvability of \eqref{momem} in $h^{1+\tilde 
d}(\N^*,\C)$ is 
guaranteed by the techniques adopted in the proof of Proposition 
\ref{entire}.
Nevertheless, such result 
presents some disadvantages which make it unsuitable 
for our purposes as we explain in Remark \ref{nextlevel}.
\end{osss}

\section{Proof of Theorem \ref{global}}\label{proofglobal}
\subsection{Local exact controllability}\label{profgloballocal}
For $\varepsilon,T>0$, we introduce the following 
subspaces of $H^s_\Gi$ with $s>0$:
	$$O_{\varepsilon,T}^{s}:=\big\{\psi\in H_{\Gi}^{s}\big|\ \|\psi\|_{L^2}=1,\ \|\psi 
-\phi_1(T)\|_{(s)}<\varepsilon\big\}, \ \ \ \ \ \ \ \phi_1(T):=e^{-i\lambda_1 T}\phi_1.$$ 
\begin{defi}The \eqref{mainx1} is said to be locally exactly controllable in $O_{\varepsilon,T}^s$ for
$\varepsilon,T,s>0$ when, for every 
$\psi\in 
O_{\varepsilon,T}^{s}$, there exists $u\in L^2((0,T),\R)$ such that $\psi= \G^u_T\phi_1.$\end{defi}

\begin{prop}\label{localteorema}
Let the hypotheses of Theorem \ref{global} be satisfied. The \eqref{mainx1} 
is locally exactly controllable in $ O_{\varepsilon,T}^{s}$ for $\varepsilon>0$ sufficiently small, $T>0$ 
sufficiently large and $s=2+d$ with $d$ from 
Assumptions II$(\eta,\tilde d)$. \end{prop}
\begin{proof}
The result corresponds to the surjectivity of the map $\G_T^{(\cdot)}\phi_1: L^2((0,T),\R)\longrightarrow 
O_{\varepsilon,T}^{s}.$ Let us define $\phi_k(T):=e^{-i\lambda_k T}\phi_k$ with $k\in\N^*$. We decompose 
$\G_{T}^{(\cdot)}\phi_1=\sum_{k\in\N^*}{\phi_k(T)}\la \phi_k(T),\G_{T}^{(\cdot)}\phi_1\ra_{L^2}$ and we consider
$$\alpha:u\in L^2((0,T),\R)\longmapsto \big(\la \phi_k(T), \G_{T}^{u}\phi_1\ra_{L^2}\big)_{k\in\N^*}.$$ Now, 
$\G_T^u\phi_1\in H^s_\Gi$ for every $u\in L^2((0,T),\R)$ and $T>0$ thanks to Proposition \ref{laura}. Thus, 
$\alpha$ takes value in $Q^s:=\{{\bf x}:=(x_k)_{k\in\N^*}\in h^s(\C)\ |\ \|{\bf x}\|_{\ell^2}=1\}$.
Defined $\updelta:=(\delta_{k,1})_{k\in\N^*}$, we notice that $$\alpha(0)=\big(\la \phi_k(T), 
\G_{T}^{0}\phi_1\ra_{L^2}\big)_{k\in\N^*}=\big(\la \phi_k(T), \phi_1(T)\ra_{L^2}\big)_{k\in\N^*}=\updelta.$$
The local exact controllability of the \eqref{mainx1} in $O_{\varepsilon,T}^s$ is equivalent to the surjectivity of 
the function
$$\alpha:L^2((0,T),\R)\longrightarrow Q^s_{\varepsilon}:=\{{\bf x}:=(x_k)_{k\in\N^*}\in Q^s\ |\ \|{\bf 
x}-\updelta\|_{(s)}<\varepsilon\}.$$ 
Let $T_{\updelta}Q^s$ be the tangent space of $Q^s$ in the point $\alpha(0)=\updelta$. For every 
$f:[0,1]\rightarrow Q^s$ such that $f(0)=\updelta$ and $f'(0)={\bf x}\in h^s(\C)$, we have $0=(\dd_t\la 
f(t),f(t)\ra)(t=0)=2\Re(\la{\bf x},\updelta\ra_{\ell^2})$, which implies $$ T_{\updelta}Q^s=\{{\bf 
x}:=(x_k)_{k\in\N^*}\in h^s(\C)\ |\ ix_1\in\R\}.$$
Let $P$ be the orthogonal projector onto $T_{\updelta}Q^s$. We define $\widetilde Q^s:=\{{\bf 
x}:=(x_k)_{k\in\N^*}\in h^s(\C)\ |\ \|{\bf x}\|_{\ell^2}\leq 1\}$ and $\widetilde Q^s_{\varepsilon}:=\{{\bf 
x}:=(x_k)_{k\in\N^*}\in \widetilde Q^s\ |\ \|{\bf x}\|_{(s)}<\varepsilon\}.$ We consider the Fréchet derivative 
of $\alpha$ in $u=0$:
$$\gamma:v\in L^2((0,T),\R)\longmapsto (d_u\alpha(u=0))\cdot\, v\ \in T_{\updelta}Q^s.$$ 
We notice that $P\alpha:L^2((0,T),\R)\rightarrow \widetilde Q^s\cap T_{\updelta}Q^s$ and its Fréchet 
derivative in $u=0$ is $\gamma$. If $\gamma$ is surjective 
in $T_{\updelta}Q^s$, then the Generalized Inverse Function Theorem (\cite[Theorem\  1;\ p.\ 240]{Inv}) guarantees 
the existence of $\varepsilon>0$ sufficiently small so that $P\alpha$ is surjective in $\widetilde 
Q^s_{\varepsilon}\cap T_{\updelta}Q^s$. Now, for every ${\bf 
x}=\{x_{k}\}_{{k\in\N^*}}\in Q^s_\epsilon$, we have $P{\bf x}\in\widetilde Q_\epsilon^s$ and if there exists $u\in 
L^2((0,T),\R)$ such that $P{\bf x}=P\alpha(u)$, then
	\begin{align*}{\bf x}&=P{\bf x}+\sqrt{1-\|P{\bf 
x}\|_{\ell^2}^2}\updelta=P\alpha(u)+\sqrt{1-\|P\alpha(u)\|_{\ell^2}^2}\updelta =\alpha(u).\end{align*}
	The last relation and the Generalized Inverse Function Theorem imply that if $\gamma$ is 
surjective in $T_{\updelta}Q^s$, then $\alpha$ is surjective in $Q^s_{\epsilon}$ with $\epsilon>0$ sufficiently 
small.
Thus, we study the function $\gamma$ and, thanks to the Duhamel's formula provided in \eqref{form}, we notice that 
it is composed by the elements $$\gamma_{k}(v)=-i\Bigg\la e^{-i\lambda_k 
T}\phi_k,\int_0^Te^{-iA(T-\tau)}v(\tau)Be^{-i\lambda_1\tau}\phi_1d\tau\Bigg\ra_{L^2}=
	-i\int_{0}^Tv(\tau)e^{i(\lambda_k-\lambda_1)\tau}d\tau \la\phi_k,B\phi_1\ra_{L^2}.$$
	Proving the surjectivity of $\gamma$ corresponds to ensure the solvability of the following moment problem
	\begin{equation}\begin{split}\label{mome1}
	{x_{k}}\la\phi_j,B\phi_k\ra_{L^2}^{-1}=-i\int_{0}^Tu(\tau)e^{i(\lambda_k-\lambda_1)\tau}d\tau,\ \ \ \ \ \ \ 
\forall k\in\N^*
	\\
	\end{split}\end{equation} 
	for every $(x_{k})_{k\in\N^*}\in T_{\updelta}Q^s$. We notice that $\big(x_k 
\la\phi_k,B\phi_1\ra_{L^2}^{-1}\big)_{k\in\N^*}\in h^{s-2-\eta}=h^{d-\eta}\subseteq h^{\tilde d}$ thanks to the 
point {\bf 1.\,\,}of Assumptions I$(\eta)$. As $B$ is symmetric, we have $\la\phi_1,B\phi_1\ra_{L^2}\in\R$ and 
$i{x_{1}}\la\phi_1,B\phi_1\ra_{L^2}^{-1}\in\R.$ Thanks to \cite[Proposition\ 6.2]{wave}, there exist $\delta>0$ 
and 
$\MM\in\N^*$ such that
	$$\inf_{{k\in\N^*}}\big|\sqrt{\lambda_{k+\MM}}-\sqrt{\lambda_k}\big|\geq\delta \MM.$$
	The last relation and the identity \eqref{interessante} ensure that Proposition $\ref{entire}$ is satisfied 
and the solvability of 
$(\ref{mome1})$ is guaranteed in $\{(c_k)_{k\in\N^*}\in 
h^{\tilde d}(\C)\ |\ ic_1\in\R\}$ for $T>0$ large enough. As $\big(x_k 
\la\phi_k,B\phi_1\ra_{L^2}^{-1}\big)_{k\in\N^*}$ belongs to such space for every $(x_{k})_{k\in\N^*}\in 
T_{\updelta}Q^s$, the claim is proved. 
\qedhere\end{proof}
	\begin{osss}\label{timereversibility}
  We consider the unitary propagator $\widetilde \G_t^u$ generated by the time-dependent Hamiltonian $-A-u(T-t)B$ 
with $u\in L^2((0,T),\R)$ with $T>0$. As explained in $\cite[Section\ 2.3]{mio1}$, the propagator $\widetilde 
\G_t^u$ represents the reversed dynamics of the \eqref{mainx1} with the same $u$ and 
\begin{align}\label{inversepropagator}(\widetilde\G_T^u)(\G_T^u)=(\G_T^u)(\widetilde\G_T^u)=I.\end{align}
When the hypotheses of Theorem \ref{global} are satisfied, the local exact controllability results 
of Proposition \ref{localteorema} is also valid for the reversed dynamics. Indeed, the functions 
$(\phi_k)_{k\in\N^*}$ are eigenfunctions of $-A$ corresponding to the eigenvalues 
$(-\lambda_k)_{k\in\N^*}$. Let $s=2+d$ and $d$ from 
Assumptions II$(\eta,\tilde d)$. As in the proof of Proposition \ref{localteorema}, the local exact 
controllability of the 
reversed dynamics in the neighborhood of $H^s_\Gi$:
$$\widetilde O_{\varepsilon,T}^{s}:=\big\{\psi\in H_{\Gi}^{s}\big|\ \|\psi\|_{L^2}=1,\ \|\psi 
-e^{i\lambda_1 T}\phi_1\|_{(s)}<\varepsilon\big\},$$
with $\varepsilon,T>0$ can be ensured by proving the solvability of a moment problem of the form
	\begin{equation*}\begin{split}
	{x_{k}}\la\phi_j,B\phi_k\ra_{L^2}^{-1}=-i\int_{0}^Tu(\tau)e^{-i(\lambda_k-\lambda_1)\tau}d\tau,\ \ \ \ \ \ \ 
\forall k\in\N^*
	\\
	\end{split}\end{equation*} 
	for every $(x_{k})_{k\in\N^*}\in T_{\updelta}Q^s$. The result is proved as Proposition \ref{localteorema} 
thanks 
to Remark 
\ref{entirecoro}. Hence, for $\varepsilon>0$ 
sufficiently small and $T>0$ 
sufficiently large, we have that for every $\psi\in \widetilde O_{\varepsilon,T}^{s}$, there exists $u\in 
L^2((0,T),\R)$ such that $\psi= \widetilde \G^u_T\phi_1$. Finally, the last relation yields that
\begin{align}\label{reversedmoscio}\phi_1= \G^{\widetilde u}_T\psi\ \ \ \  \ \text{for}\ \ \ 
\  \ \widetilde 
u(t)=u(T-t).\end{align}

\end{osss}

\subsection{Global exact controllability}  
\begin{proof}[Proof of Theorem \ref{global}]
Let us consider $\psi_1,\psi_2\in H^{s}_{\Gi}$ be so that $\|\psi_1\|_{L^2}=\|\psi_2\|_{L^2}$. We assume that 
$\|\psi_1\|_{L^2}=\|\psi_2\|_{L^2}=1$, but the result is equivalently proved in the general case. Let 
$T,\varepsilon>0$ be so that Proposition \ref{localteorema} and Remark \ref{timereversibility} are valid. Thanks 
to 
Proposition $\ref{approx}$ and Corollary \ref{approxreversed}, there exist $T_1,T_2>0$, $u_1\in 
L^2((0,T_1),\R)$ and $u_2\in L^2((0,T_2),\R)$ such that $$\|\G^{u_1}_{T_1}\psi_1-e^{i\lambda_1 
T}\phi_1\|_{(s)}<{\varepsilon},\ \ \ \ \ \|\widetilde \G^{u_2}_{T_2}\psi_2-e^{-i\lambda_1 
T}\phi_1\|_{(s)}<{\varepsilon},\ \ \ 
\ \ \Longrightarrow\ \ \ \ \ \  \G^{u_1}_{T_1}\psi_1\in \widetilde O_{\varepsilon,T}^{s},\ 
\widetilde\G^{u_2}_{T_2}\psi_2\in 
 O_{\varepsilon,T}^{s}.$$ 
From Remark \ref{timereversibility} (relation \eqref{reversedmoscio}), there exists $u_3\in L^2((0,T),\R)$ such 
that 
$\G_T^{u_3}\G^{u_1}_{T_1}\psi_1=\phi_1.$ From Proposition \ref{localteorema}, there 
exist $u_4\in L^2((0,T),\R)$ such that $\G_T^{u_4}\phi_1=\widetilde \G^{u_2}_{T_2}\psi_2$ and then
$$\G_T^{u_4}\G_T^{u_3}\G^{u_1}_{T_1}\psi_1=\widetilde \G^{u_2}_{T_2}\psi_2\ \ \ \ \ \ \ 
\Longrightarrow \ \ \ \ \ \ \ \ \ 
\G^{\widetilde u_2}_{T_2}\G_T^{u_4}\G_T^{u_3}\G^{u_1}_{T_1}\psi_1= \psi_2,\ 
\ \  \ \ \text{for }\ \ \ \widetilde u_2(t)=u_2(T_2-t).$$
In conclusion, Theorem \ref{global} is proved since there exists $\widehat T>0$ and $ \widehat u\in 
L^2((0,\widehat 
T),\R)$ such that $\G_{\widehat T}^{\widehat u}\psi_1=\psi_2.$\qedhere

\end{proof}

	\begin{osss}\label{nextlevel}
As discussed in Remark \ref{lowerlevel}, one can ensure the 
solvability of \eqref{mome1} for sequences in $h^{1+\tilde 
d}(\N^*,\C)$ for every $T>0$ by changing the hypotheses on the function $G$.
This outcome leads to
the local exact 
controllability of Proposition \ref{localteorema} for any $T>0$. Nevertheless, the 
new 
conditions on $G$ and the solvability of the moment problem in 
$h^{1+\tilde 
d}(\N^*,\C)$, rather than $h^{\tilde 
d}(\N^*,\C)$, compel to impose stronger hypotheses on the problem.
Our choice allows us to treat a wide range of 
problems as Theorem \ref{bim.1} and 
Theorem 
\ref{bim.2} which would be otherwise out of reach.
We also point out that, when we extend the local exact controllability in order to prove Theorem \ref{global}, the 
property of being 
controllable in any 
positive time is lost.  	\end{osss}
 
\section{Bilinear quantum systems on star graphs}\label{generic_application}
In the current section, we study the global exact controllability when $\Gi$ is a star graph by applying Theorem 
\ref{global}. The result is obtained by providing a suitable entire function $G$ satisfying the hypotheses of the 
theorem. From now on, when we call $\Gi$ a star graph, we also consider it as a quantum graph.
\begin{teorema}\label{final} Let $\Gi$ be a star graph equipped with ($\Di$/$\NN$) made by edges of lengths 
$\{L_j\}_{j\leq N}\in\AL\LL(N)$. If the couple $(A,B)$ satisfies Assumptions I$(\eta)$ and Assumptions 
II$(\eta,\epsilon)$ for $\eta,\epsilon>0$, then the $(\ref{mainx1})$ is globally exactly controllable in 
$H^{s}_{\Gi}$ for $s=2+d$ and $d$ from Assumptions II$(\eta,\epsilon)$.
	\end{teorema}
\begin{proof}
	{\bf 1) Star graph equipped with $(\Di)$.} The boundary conditions ($\Di$) on $V_e$ imply that 
$\phi_k=(a^1_k\sin(\sqrt{\lambda_k} x),...,a_k^n\sin(\sqrt{\lambda_k} x))$ for each $k\in\N^*$ and suitable 
$\{a_k^l\}_{l\leq N}\subset \C$ such that $(\phi_k)_{k\in\N^*}$ is orthormal in $\Hi.$ The conditions ($\NN\KK$) in 
the internal vertex $v\in V_i$ ensure that, for every $k\in\N^*$, 
\begin{equation}\label{cot}\begin{split}\begin{cases}a^1_k\sin(\sqrt{\lambda_k}L_1)=...=a^N_k\sin(\sqrt{\lambda_k}
L_N),\\
	\sum_{l\leq N} a^l_k\cos(\sqrt{\lambda_k}L_l)=0,\end{cases} 
	\Longrightarrow \ \ \ \ \sum_{l=1}^N \cot(\sqrt{\lambda_k}L_l)=0.\end{split}\end{equation} We use the provided identities in order to construct an entire function satisfying the hypotheses of Theorem $\ref{global}$. To this purpose, we define an entire function $G$ and two maps $\tilde G$ and $H$ such that
	\begin{equation}\label{stimettina}\begin{split}
	&G(x):=\prod_{l\leq N}\sin({x} L_l)\sum_{l\leq N} \cot({x} L_l)
	\ \ \ \ \ \ \ \ \ \ \ \ G'(x)=-\tilde{G}(x)+H(x),\\
	&\tilde{G}(x):={\prod_{l\leq N}\sin({x} L_l)}\sum_{l\leq N} \frac{L_l}{\sin^2({x} L_l)}, \ \ \ \ \ \ \ \ \ H(x):=\frac{d}{dx}\big(\prod_{l\leq N}\cos({x} L_l)\big)\sum_{l\leq N} \cot({x} L_l).\end{split}\end{equation}
	The identities $(\ref{cot})$ and $(\ref{stimettina})$ imply that $H(\sqrt{\lambda_k})=0$ and 
$G'(\sqrt{\lambda_k})=-\tilde{G}(\sqrt{\lambda_k})$ for every $k\in\N^*.$ Now,
	\begin{equation}\label{stimetta}\begin{split}
	|\tilde{G}(x)|&= \frac{\prod_{l\leq N}|\sin({x} L_l)|\sum_{l\leq N} L_l\prod_{k\neq l}\sin^2({x} 
L_k)}{\prod_{l\leq N}\sin^2({x} L_l)}\geq{L^*} \sum_{l\leq N} {\prod_{k\neq l}|\sin({x} L_k)|},\\
	\end{split}\end{equation} 
	with $L^*:=\min_{l\leq N} L_l$. Thanks to $(\ref{stimetta})$, we refer to $\cite[Corollary\ A.10;\ (2)]{wave}$ 
(which contains a misprint as it is valid for 
$\lambda>\frac{\pi}{2}\max\{1/L_j\ :\ j\leq N\}$) and for 
every $\epsilon>0$, there exists $C_1>0$ such that \begin{align}\label{relo}|G'(\pm\sqrt{\lambda_k})|\geq{L^*} 
\sum_{l=1}^N{\prod_{j\neq l}|\sin(\sqrt{\lambda_k} L_j)|}\geq\frac{C_1}{(\sqrt{\lambda_k})^{1+\epsilon}},\ \ \ \ \ 
\forall k\in\N^*\ \ \ \ :\ \ \ \lambda_{k}>\frac{\pi}{2}\max\{L_j^{-1}\}_{j\leq N}.\end{align}
	\begin{osss}\label{cicciobastardo}
		For every $k\in\N^*$ and $j\leq N$, we have $|\phi_k^j(L_j)|\neq 0,$ otherwise the $(\NN\KK)$ conditions 
would ensure that $\phi_k^l(L_l)=\phi_k^m(L_m)=0$ with $l,m\leq N$ so that $\phi_k^l,\phi_k^m\not\equiv 0$ and 
there would be satisfied $a_k^l\sin( L_l\sqrt{\lambda_k})=a_k^m\sin( L_m\sqrt{\lambda_k})=0$ with $a_k^l,a_k^m\neq 
0$, which is absurd as $\{L_j\}_{j\leq N}\in\AL\LL(N).$
	\end{osss}

	\noindent
	 Remark $\ref{cicciobastardo}$ implies $|G'(\pm\sqrt{\lambda_k})|\neq 0$ for $k\in\N^*$. Thanks to 
$(\ref{interessante})$ and \eqref{relo}, there exists $C_2>0$ so that $$|G'(\pm\sqrt{\lambda_k})|\geq 
{C_2}{k^{-(1+\epsilon)}},\ \ \ \forall k\in\N^*.$$We notice that the spectrum of $A$ is simple. Indeed, if there 
would exist two orthonormal eigenfuctions $f$ and $g$ of $A$ corresponding to the same eigenvalue $\lambda$, then 
$h(x)=f(v)g(x)-g(v)f(x)$ would be another eigenfunction of $A$. Now, $h$ is an eigenfunction and $h(v)=0$, which 
is impossible thanks to Remark $\ref{cicciobastardo}$.

	  As $|\cos(zL_l)|\leq e^{L_l|z|}$ and $|\sin(zL_l)|\leq e^{L_l|z|}$ for every $l\leq N$ and $z\in\C$, we notice that $|G(z)|\leq Ne^{|z|\sum_{l=1}^NL_l}$ for every $z\in\C.$
Now, $G(\sqrt{\lambda_k})=0$ for every $k\in\N^*$ thanks to $(\ref{cot})$ and $G\in L^\infty(\R,\R)$.
	 
	\noindent
	In conclusion, the claim is achieved as Theorem $\ref{global}$ is valid with respect to the function $G$ when $\tilde d=\epsilon$.

	\smallskip
	\noindent
	{\bf 2) Generic star graph.} Let $I_1\subseteq \{1,...,N\}$ be the set of indices of those edges containing an 
external vertex equipped with $(\NN)$ and $I_2:=\{1,..,N\}\setminus I_1$. The proof follows from the techniques 
adopted in {\bf 1)} by considering Proposition $\ref{zuazua3}$ (rather than $\cite[Corollary\ A.10;\ 
(2)]{wave}$) and the entire map
	{{\begin{equation*}\begin{split}
	G(x):=&\prod_{l\in I_2}\sin({x} L_l)\prod_{l\in I_1}\cos({x} L_l)\Big(\sum_{l\in I_2} \cot({x} L_l)+\sum_{l\in I_1} \tan({x} L_l)\Big).\qedhere 
	\end{split}\end{equation*}}}\end{proof}

\begin{osss}\label{pallosso}
When $\Gi$ is a star graph equipped with ($\Di$) such that $L_2/L_1$ and $L_3/L_1$ are rationals,$$\exists 
n_2,n_3,m_2,m_3\in\N^*\ \ \ \ :\ \ \ \ {L_2}/{L_1}={n_2}/{m_2},\ \ \ \ \ \  \ \ \ {L_3}/{L_1}={n_3}/{m_3}.$$
The numbers $\{\mu_k\}_{k\in\N^*}$ with $\mu_k=\frac{k^2m_2^2 m_3^2\pi^2}{L_1^2}$ are eigenvalues of $A$ and they 
are 
multiple. Fixed $k\in\N^*$,
$$f_k=\big(-2\sin(\sqrt\mu_k x),\sin(\sqrt\mu_k x),\sin(\sqrt\mu_k x),0,...,0\big),\ \ \ \ \ \ \ \ 
g_k=\big(0,\sin(\sqrt\mu_k 
x),-\sin(\sqrt\mu_k x),0,...,0\big)$$
are reciprocally orthogonal eigenfunctions of $A$ corresponding to $\mu_k$. In addition, we notice that the 
sequence $\{g_k\}_{k\in\N^*}$ is composed by eigenfuctions vanishing in the edge $e_1$. The same kind of 
construction can be repeated when the star graph $\Gi$ is equipped with the general boundary conditions 
($\Di$/$\NN$).
\end{osss}

\needspace{3\baselineskip}
\begin{coro}\label{copia} Let $\Gi$ be a star graph equipped with ($\Di$/$\NN$). Let $\Gi$ satisfy the following conditions with $\widetilde N\in 2\N^*$ such that $\widetilde N\leq N$. 
	
	\begin{itemize}
		\item For $j\leq \widetilde N/2$, the two external vertices belonging to $e_{2j-1}$ and $e_{2j}$ are equipped with $(\Di)$ or $(\NN)$.
		\item The couples of edges $\{e_{2j-1},e_{2j}\}_{j\leq {\widetilde N}/{2}}$ are long $\{L_j\}_{j\leq {\widetilde N}/{2}}$, while the edges $\{e_j\}_{\widetilde N< j\leq N}$ measure $\{L_j\}_{\widetilde N< j\leq N}$. In addition, $\{L_j\}_{j\leq \frac{\widetilde N}{2}}\cup \{L_j\}_{\widetilde N< j\leq N}\in\AL\LL\big(\frac{\widetilde N}{2}+N-\widetilde N\big)$.
	\end{itemize}

\noindent	
If $(A,B)$ satisfies Assumptions I$(\eta)$ and Assumptions II$(\eta,\epsilon)$ for $\eta,\epsilon>0$, then the $(\ref{mainx1})$ is globally exactly controllable in $H^{s}_{\Gi}$ for $s=2+d$ and $d$ from Assumptions II$(\eta,\epsilon)$.
\end{coro}
		\begin{figure}[H]
	\centering
	\includegraphics[width=\textwidth-150pt]{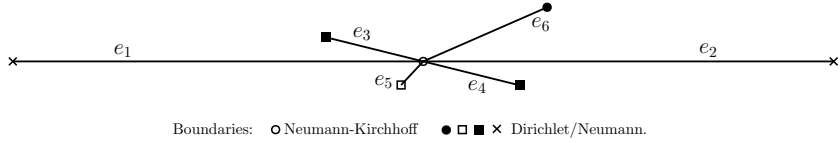}
	\caption{Example of graph described in Corollary $\ref{copia}$ with $\widetilde N=4$ and $N=6$.}
\end{figure}

\begin{proof}
	 Let $I_1\subseteq \{1,...,\widetilde N/2\}$ be the set of $j$ such that $e_{2j-1}$ and $e_{2j}$ contain two external vertices of $\Gi$ equipped with $(\NN)$ and $I_2:=\{1,..,\widetilde N/2\}\setminus I_1$.
	Let $I_3\subseteq \{\widetilde N+ 1,...,N\}$ be the set of $j$ such that $e_j$ contains an external vertex of $\Gi$ equipped with $(\NN)$ and $I_4:=\{\widetilde N+ 1,...,N\}\setminus I_3$. Let $$(\lambda_k^1)_{k\in\N^*}:=\Big(\frac{(2k-1)^2\pi^2}{4L_j^2}\Big)_{\underset{j\in I_1}{j,k\in\N^*}},\ \ \ \ \ \ \ (\lambda_k^2)_{k\in\N^*}:=\Big(\frac{k^2\pi^2}{L_j^2}\Big)_{\underset{j\in I_2}{j,k\in\N^*}}.$$ We notice that $(\lambda_k^1)_{k\in\N^*}\cup(\lambda_k^2)_{k\in\N^*}\subset (\lambda_k)_{k\in\N^*}$ are the only eigenvalues of $A$ corresponding to eigenfunctions vanishing in the internal vertex $v$. For every $f\in(\phi_k)_{k\in\N^*}$ of $A$ corresponding to an eigenvalue $\lambda\in (\lambda_k^1)_{k\in\N^*}$, $f$ is uniquely defined (up to multiplication for $\alpha\in\C$ such that $|\alpha|=1$) by the identities
$$f^{2j-1}(x)=-f^{2j}(x)=\sqrt{L_j^{-1}}\cos(\sqrt{\lambda} x),\ \ \ \ \ \ \ \ \ \ f^l \equiv 0,\ \ \ \ \ \ \ \ \forall l\in\{1,...,N\}\setminus\{2j-1,2j\}.$$
The same property is valid for $(\lambda_k^2)_{k\in\N^*}$ and then, the eigenvalues 
$(\lambda_k^1)_{k\in\N^*}\cup(\lambda_k^2)_{k\in\N^*}$ are simple. In conclusion, the discrete spectrum of $A$ is 
simple since, if there would exist a multiple eigenvalue 
$$\lambda\in (\lambda_k)_{k\in\N^*}\setminus\big((\lambda_k^1)_{k\in\N^*}\cup(\lambda_k^2)_{k\in\N^*}\big),$$
then there would exist two orthonormal eigenfuctions $f$ and $g$ corresponding to the same eigenvalue $\lambda$. Now, $h(x)=f(v)g(x)-g(v)f(x)$ would be another eigenfunction corresponding to $\lambda$ such that $h(v)=0$, which is impossible as it would imply that $\lambda\in (\lambda_k^1)_{k\in\N^*}\cup(\lambda_k^2)_{k\in\N^*}.$ Thus, $(\lambda_k)_{k\in\N^*}$ are simple eigenvalues.

\noindent
The remaining part of proof follows the one of Theorem $\ref{final}$ thanks to Proposition $\ref{zuazua3}$ by considering
\begin{equation*}\begin{split}
G(x):=&\prod_{l\in I_2\cup I_4}\sin({x} L_l)\prod_{l\in I_1\cup I_3}\cos({x} L_l)\Big(2\sum_{l\in I_2} \cot({x} L_l)+2\sum_{l\in I_1} \tan({x} L_l)+\sum_{l\in I_4} \cot({x} L_l)+\sum_{l\in I_3} \tan({x} L_l)\Big).\qedhere
\end{split}\end{equation*}
\end{proof}
\subsection{Proofs of Theorem \ref{bim.1} and of Theorem \ref{bim.2}}\label{applica}

\begin{proof}[Proof of Theorem $\ref{bim.1}$]
Theorem $\ref{bim.1}$ is proved such as $\cite[Theorem\ 1.2]{mio3}$ that is stated for $N=4$. The only difference 
between the two results is the fact that Theorem $\ref{bim.1}$ follows from Theorem $\ref{final}$ instead 
of $\cite[Proposition\ 3.3]{mio3}$, which is only valid for $N\leq 4$.\qedhere
\end{proof}

\begin{proof}[Proof of Theorem $\ref{bim.2}$]
	The conditions ($\NN$) in $V_i$ imply the existence, for every $k\in\N^*$, of $\{a_k^l\}_{l\leq N}\subset\C$ such that $\phi_k=(a^1_k\cos(x\sqrt{\lambda_k}),...,a_k^N\cos(x\sqrt{\lambda_k})).$	The coefficients $\{a_k^l\}_{l\leq N}\subset\C$ are so that $(\phi_k)_{k\in\N^*}$ forms a Hilbert basis of $\Hi$ and then
	\begin{equation}\label{equation1}1=\sum_{l\leq 
N}\int_{0}^{L_l}|a_k^l|^2\cos^2(x\sqrt{\lambda_k})dx=\sum_{l\leq 
N}|a_k^l|^2\Big(\frac{L_l}{2}+\frac{\sin(2L_l\sqrt{\lambda_k})}{4\sqrt{\lambda_k}}\Big).\end{equation} For every 
$k\in\N^*$, the ($\NN\KK$) boundary conditions in $V_i$ ensure 
\begin{equation}\label{tremor}\begin{split}a^1_k\cos&(\sqrt{\lambda_k}L_1)=...=a^N_k\cos(\sqrt{\lambda_k}L_N),\ \ \ \ \  \ \ \sum_{l\leq N} a^l_k\sin(\sqrt{\lambda_k}L_l)=0,\\
	&\sum_{l\leq N} \tan(\sqrt{\lambda_k}L_l)=0,\ \ \ \ \ \ \ \ \ \  \sum_{l\leq N}|a_k^l|^2{\sin(2L_l\sqrt{\lambda_k})}=0.\\ 
	\end{split}
	\end{equation} 
	The last identities and $(\ref{equation1})$ imply $1=\sum_{l=1}^N|a_k^l|^2{L_l}/{2}$. Thanks to $(\ref{tremor})$, we have $a_k^l=a_k^1\frac{\cos(\sqrt{\lambda_k} L_1)}{\cos(\sqrt{\lambda_k} L_l)}$ for $l\neq 1$ and $k\in\N^*$. Thus, $|a_k^1|^2\big(L_1+\sum_{l=2}^NL_l\frac{\cos^2(\sqrt{\lambda_k} L_1)}{\cos^2(\sqrt{\lambda_k} L_l)}\big)={2}$ for every $k\in\N^*$ and
	\begin{equation*}
	\begin{split}
	 &|a_k^1|^2={2\prod_{m=2 }^N\cos^2(\sqrt{\lambda_k} L_m)}{\Big(\sum_{j=1}^N L_j\prod_{m\neq j}\cos^2(\sqrt{\lambda_k} L_m)\Big)^{-1}}.\\
	\end{split}
	\end{equation*}

	\smallskip
	\noindent
	{\bf Validation of Assumptions I($\mathbf{3+\epsilon}$) with $\mathbf{\epsilon>0}$.} 
For every $k\in\N^*$, thanks to the relation $(\ref{tremor})$
$$\prod_{l\leq N}\cos(\sqrt{\lambda_k}L_l)\sum_{l\leq N} \tan(\sqrt{\lambda_k}L_l)=0,\ \ \ \ \Longrightarrow\ \ \ \ \ \ \sum_{l=1}^N\sin(\sqrt{\lambda_k} L_l)\prod_{m\neq l }\cos(\sqrt{\lambda_k} L_m)=0.$$ Thanks to the relation $(\ref{interessante})$ and Corollary $\ref{corozuazua}$, for every $\epsilon>0$, there exist $C_1,C_2>0$ such that,
\begin{equation}
	\label{oip}
	\begin{split}
	|a_k^1|&\geq \sqrt{\frac{2}{\sum_{l=1}^NL_l\cos^{-2}(\sqrt{\lambda_k} L_l)}}\geq\sqrt{\frac{2}{\sum_{l=1}^NL_lC_1^{-2}{\lambda_k^{1+\epsilon}}}}\geq \frac{C_2}{k^{1+\epsilon}},\ \ \ \ \ \ \  \ \forall k\in\N^*.\\
	\end{split}
	\end{equation}
	In addition, $\la \phi_1^l,(B\phi_k)^l\ra_{L^2(e_l,\C)}=0$ for $2\leq l\leq N$ and, for every $k\in\N^*$,
	\begin{equation}\begin{split}\label{ancora}
	\la \phi_1,B\phi_k\ra_{L^2}&=\la \phi_1^1,(B\phi_k)^1\ra_{L^2(e_1,\C)}=
	-\frac{120a_k^1a_1^1L_1^6}{(\sqrt{\lambda_k}+\sqrt{\lambda_1})^4}-\frac{120a_k^1a_1^1L_1^6}{(\sqrt{\lambda_k}-\sqrt{\lambda_1})^4}+o(\sqrt{\lambda_k}^{-5}).\\
	\end{split}\end{equation}
From the relations $(\ref{oip})$ and $(\ref{ancora})$, thanks to the relation $(\ref{interessante})$, for every $\epsilon>0$, there exists $C_3>0,$ such that for $k\in\N^*$ sufficiently large, \begin{align}\label{alin}|\la\phi_1,B\phi_k\ra_{L^2}|\geq {C_3}{k^{-(5+\epsilon)}}. \end{align}

\noindent
Now, it is possible to compute $a_k(\cdot)$ and $B_k(\cdot)$ with $k\in\N^* $, analytic functions in $\R^+$, so that $$a_k(L_1)^2=(a_k^1)^2,\ \ \ \  {a_1(L_1)a_k(L_1)}B_k(L_1)=\la\phi_1, B\phi_k\ra_{L^2}$$ and such that each ${a_1(\cdot)a_k(\cdot)}B_k(\cdot)$ is non-constant and analytic.
	Thus, each ${a_1(\cdot)a_k(\cdot)}B_k(\cdot)$ has discrete zeros $\tilde V_k\subset\R^+$ and $\tilde V=\bigcup_{k\in\N^*}\tilde V_k$ is countable. For every $\{L_l\}_{l\leq N}\in \AL\LL(N)$ so that $L_1\not\in \tilde V$, we have  $|\la\phi_1, B\phi_k\ra_{L^2}|\neq 0$ for every $k\in\N^*.$
Thus, the point {\bf 1.\,\,}of Assumptions I($3+\epsilon$) is ensured thanks to the relations $(\ref{alin})$ since, for every $\epsilon>0$, there exists $C_4>0$ such that $$|\la\phi_1,B\phi_k\ra_{L^2}|\geq {C_4}{k^{-(5+\epsilon)}},\ \ \ \ \forall k\in\N^*.$$

	Let $(k,j),(m,n)\in I,\ (k,j)\neq(m,n)$ for $I:=\{(j,k)\in(\N^*)^2:j< k\}$. We prove the validity of the point 
{\bf 2.\,\,}of Assumptions I($3+\epsilon$). As above, we compute $F_k(\cdot)$ with $k\in\N^*$, analytic in $\R^+$, 
such that $\la\phi_k, B\phi_k\ra_{L^2}=F_k(L_1)$. Each 
$F_{j,k,l,m}(\cdot):=F_j(\cdot)-F_k(\cdot)-F_l(\cdot)+F_m(\cdot)$ is non-constant and analytic in $\R^+$, the set 
of its positive zeros $V_{j,k,l,m}$ is discrete. Now, we introduce the countable set: $$V:=\bigcup_{(j,k),(l,m)\in 
I\ : \ (j,k)\neq(l,m)}V_{j,k,l,m}.$$ For $\{L_l\}_{l\leq N}\in \AL\LL(N)$ so that $L_1\not\in V\cup \tilde V$, the 
point {\bf 2.\,\,} of Assumptions I$(3+\epsilon)$ with $\epsilon>0$ is satisfied.
	
\smallskip
	\noindent
	{\bf Validation of Assumptions II($\mathbf{3+\epsilon_1,\epsilon_2}$) with $\mathbf{\epsilon_1,\epsilon_2>0}$ so that $\mathbf{\epsilon_1+\epsilon_2\in\big(0,\frac12\big)}$.} Let 
$$P(x):=(5x^6-24x^5 L_1 +45 x^4 L_1^2-40 x^3 L_1^3 +15 x^2 L_1^4 -L_1^6).$$ For
$m> 0$, we notice $B :H^m\longrightarrow H^m$ and $\dd_x(B\psi)(\widetilde v)=0$ for every $\widetilde v\in V_e$ 
since $\dd_x P(0)=0$. Now, $\dd_x(B\psi)(v)=(B\psi)(v)=0$ with $v\in V_i$ since $\dd_x P(L_1)=P(L_1)=0$. Then, 
$B:H_{\Gi}^2\rightarrow H_{\Gi}^2$. Moreover, $\dd_x^2P(L_1)=\dd_x^3P(L_1)=0,$ which implies
$B:H^m_{\NN\KK}\longrightarrow H^m_{\NN\KK}$ for every $m\in \big(0,\frac{9}{2}\big).$ For 
$d\in\big[3+\epsilon_1+\epsilon_2,\frac72\big)$ and $d_1\in\big(d,\frac72\big)$, there follow $$Ran(B|_{ 
H^{d_1}_{\NN\KK}})\subseteq H^{d_1}_{\NN\KK},\ \ \ \ \ \ \ Ran(B|_{H_{\Gi}^{2+d}})\subseteq Ran(B|_{H^{2+d}\cap 
H^{1+d}_{\NN\KK}\cap H^2_{\Gi}})\subseteq H^{2+d}\cap H^{1+d}_{\NN\KK}\cap H^2_{\Gi}.$$ 
	The point {\bf 2.\,\,}of Assumptions II($3+\epsilon_1,\epsilon_2$) with $\epsilon_1,\epsilon_2>0$ so that $\epsilon_1+\epsilon_2\in\big(0,\frac12\big)$ is valid.

\smallskip
	\noindent
	{\bf Conclusion.}
	The couple $(A,B)$ satisfies Assumptions I$(3+\epsilon)$ and Assumptions II$(3+\epsilon_1,\epsilon_2)$ with $\epsilon_1,\epsilon_2>0$ so that $\epsilon_1+\epsilon_2\in\big(0,\frac12\big)$. Theorem $\ref{final}$ guarantees the global exact controllability of the $(\ref{mainx1})$ in $H^s_\Gi$ with $s=2+d$ and $d\geq 3+\epsilon_1+\epsilon_2$.\qedhere
\end{proof}

\section{Energetic controllability}\label{controlloenergia}
We recall that $(\varphi_{k})_{k\in\N^*}\subseteq(\phi_{k})_{k\in\N^*}$ indicates an orthonormal 
system (not necessarily complete) of $\Hi$ made by some eigenfunctions of $A$ and $(\mu_{k})_{k\in\N^*}$ the 
ordered sequence of corresponding eigenvalues. Let $$\widetilde\Hi:=\overline{\spn\{\varphi_k\ |\ 
k\in\N^*\}}^{\ L^2}.$$ We refer to Definition $\ref{energiaaa}$ for the formal 
definition of energetic controllability.

\begin{teorema}\label{globalenergetic}
Let $\Gi$ be a compact quantum graph and one of the following points be verified.

\begin{enumerate} 
	\item There exists an entire function $G$ such that  $G \in L^\infty(\R,\R)$ and there exist $J,I>0$ so that $|G(z)|\leq J e^{I|z|}$ for every $z\in\C.$
The eigenvalues $\{{\mu_k}\}_{k\in\N^*}$ are simple, the numbers $\{\pm\sqrt{\mu_k}\}_{k\in\N^*}$ are simple zeros 
of $G$ and there exist $\tilde d\geq 
0$ and $C>0$ so that $|G'(\pm\sqrt{\mu_k})|\geq \frac{C}{k^{1+\tilde d}}$ for every $k\in\N^*.$

\item There exist $C>0$ and $\tilde d\geq 0$ so that $|\mu_{k+1}-\mu_k|\geq {C}{k^{-\frac{\tilde d}{\MM-1}}}$ for each $k\in\N^*$ with $\MM$ from \eqref{g13}.

\end{enumerate}
If $(A,B)$ satisfies Assumptions I$(\upvarphi,\eta)$ and Assumptions II$(\upvarphi,\eta,\tilde d)$ for $\eta>0$, then the $(\ref{mainx1})$ is globally exactly controllable in $H^{s}_{\Gi}\cap\widetilde\Hi$ for $s=2+d$ with $d$ from Assumptions II$(\upvarphi,\eta,\tilde d)$ and energetically controllable in $(\mu_k)_{{k\in\N^*}}.$
\end{teorema}
\begin{proof}
	From Remark $\ref{wellene}$, the $(\ref{mainx1})$ is well-posed in $H^s_{\Gi}\cap\widetilde\Hi$ with $s=2+d$ 
and $d$ from Assumptions II$(\upvarphi,\eta,\tilde d)$. The statement of Theorem $\ref{global}$ holds in 
$\widetilde\Hi$ when the point {\bf 1.\,\,}is valid, while the validity of $\cite[Theorem\ 3.2]{mio3}$ in 
$\widetilde\Hi$ is guaranteed by {\bf 2.\,\,}. The global exact controllability is provided in 
$H_{\Gi}^s\cap\widetilde\Hi$ and the energetic controllability follows as $\varphi_k\in H^s_{\Gi}\cap\widetilde\Hi$ 
for every $k\in\N^*$.\qedhere
\end{proof}
Let $\Gi$ be a compact quantum graph.
By watching the structure of the graph and the boundary conditions of $D(A)$, it is possible to construct some 
eigenfuctions $(\ffi_k)_{k\in\N^*}$ of $A$ corresponding to some eigenvalues $(\mu_k)_{k\in\N^*}$. For instance, we 
consider $\Gi$ containing one loop $e_1$ of length $1$ connected to the graph in a vertex $v$. In such case, we 
point out that the Neumann-Kirchhoff boundary conditions in $v$ valid for a function $\psi\in D(A)$ yield that 
$\sum_{j\in N(v)\setminus \{1\}}\dd_x\psi^j(v)+\dd_x\psi^1(0)-\dd_x\psi^1(1)=0.$

\begin{figure}[H]
	\centering
	\includegraphics[width=\textwidth-150pt]{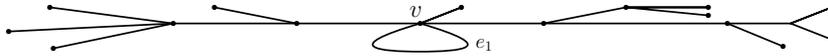}
	\caption{Example of compact graph containing a loop.}
\end{figure}

\noindent
We define $\upvarphi:=(\ffi_k)_{k\in\N^*}$ such that $\ffi_k=\big(\sqrt{{2}}\sin({2k\pi}x),0,...,0\big)$
and the corresponding eigenvalues $(\mu_k)_{k\in\N^*}=(4k^2\pi^2)_{k\in\N^*}\subseteq (\lambda_k)_{k\in\N^*}$, satisfying the gap condition
$$\inf_{k\in\N^*}|\mu_{k+1}-\mu_k|=12\pi^2>0.$$ The spectral hypotheses of Theorem \ref{globalenergetic} 
are guaranteed and the energetic controllability can be ensured by choosing a suitable $B$. In 
particular, if $(A,B)$ satisfies Assumptions I$(\upvarphi,\eta)$ and Assumptions II$(\upvarphi,\eta,0)$ for 
$\eta>0$, then Theorem $\ref{globalenergetic}$ implies the energetic controllability in $(\mu_k)_{k\in\N^*}$. As we 
will show in the proof of Theorem $\ref{bim1.1}$, this approach is also valid when $\Gi$ contains more loops ({\it 
e.g.} Figure $\ref{cappio}$).

\begin{oss} The idea described above can be adopted when $\Gi$ contains suitable sub-graphs denoted  
\virgolette{uniform chains}. A {uniform chain} is a sequence of edges of equal length $L$ connecting $M\in\N^*$ 
vertices $\{v_j\}_{j\leq M}$ such that $v_2,...,v_{M-1}\in V_i$ when $M\geq 3$. Moreover, one of the following 
conditions holds: either $v_1,v_{M}\in V_e$ are equipped with ($\Di$), or $v_1=v_{M}$ belong to $V_i$, or 
$M\in\{2,3\}$ and $v_1,v_M\in V_e$ are equipped with ($\NN$). 
	\begin{figure}[H]
		\centering
		
		\includegraphics[width=\textwidth-150pt]{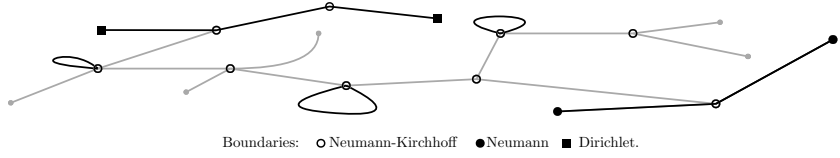}
		\caption{The figure underlines the uniform chains in a compact graph.}\label{bas3}
	\end{figure}
	
	\noindent
	Let $\Gi$ contain ${\widetilde N}\in\N^*$ {uniform chains} $\{\widetilde\Gi_j\}_{j\leq\widetilde N}$, composed 
by edges of lengths $\{L_j\}_{j\leq \widetilde N}\in\AL\LL(\widetilde N)$. Let $I_1\subseteq\{1,...,\widetilde N\}$ 
and $I_2\subseteq\{1,...,\widetilde N\}\setminus I_1$ be respectively the sets of indices $j$ such that the 
external vertices of $\widetilde\Gi_j$ are equipped with $(\NN)$ and $(\Di)$, while $I_3:=\{1,...,\widetilde 
N\}\setminus (I_1\cup I_2)$. 
	We consider the eigenvalues $(\mu_k)_{k\in\N^*}$ obtained by reordering 
$$\Big(\frac{(2k-1)^2\pi^2}{4L_j^2}\Big)_{\underset{j\in I_1}{k,j\in\N^*}}\cup 
\Big(\frac{k^2\pi^2}{L^2_j}\Big)_{\underset{j\in I_2}{k,j\in\N^*}}\cup 
\Big(\frac{4k^2\pi^2}{L^2_j}\Big)_{\underset{j\in I_3}{k,j\in\N^*}}.$$
As in the proof of $\cite[Lemma\ 2.6]{mio3}$, the Roth's Theorem $\cite[Proposition\ A.1]{mio3}$ ensures that, if 
$\{L_j\}_{j\leq \widetilde N}\in\AL\LL(\widetilde N)$, then for every $\epsilon>0,$ there exists $C>0$ so that 
$$|\mu_{k+1}-\mu_k|\geq {C}{k^{-\epsilon}},\ \ \ \ \ \ \ \ \ \forall k\in\N^*.$$ 
Thus, the spectral hypotheses of Theorem \ref{globalenergetic} 
are guaranteed and the energetic 
controllability can be ensured by choosing a suitable control operator $B$. If $(A,B)$ satisfies Assumptions 
I$(\upvarphi,\eta)$ and Assumptions II$(\upvarphi,\eta,\epsilon)$ with $\eta>0$, then Theorem 
$\ref{globalenergetic}$ implies the energetic controllability in $(\mu_k)_{k\in\N^*}$.
\end{oss}

\subsection{Proof of Theorem \ref{bim2} and some applications of Theorem \ref{globalenergetic}}\label{energia}
\begin{proof}[Proof of Theorem $\ref{bim2}$]
	Let us assume $N=3$. The ($\Di$) conditions to the external vertices $V_e$ imply $\phi_k=(a^1_k\sin(\sqrt{\mu_k} x),a^2_k\sin(\sqrt{\mu_k} x),a^3_k\sin(\sqrt{\mu_k} x))$ with suitable $(a^1_k, a^2_k, a^3_k )\in\C^3$. From the ($\NN\KK$) in $v\in V_i$, there follow $\sum_{l\leq 3}a^l_{k}\cos(\sqrt{\mu_k}L)=0$ and $a^m_{k}\sin(\sqrt{\mu_k}L)=c\in\R$ for every $m\leq 3.$ 	When $c\neq 0$, we have the eigenvalues $\big(\frac{(2k-1)^2\pi^2}{4L^2}\big)_{k\in\N^*}$ corresponding to the eigenfunctions $(g_k)_{k\in\N^*}$ so that $$g_k=\Big(\sqrt{\frac{2}{3L}}\sin\Big(\frac{(2k-1)\pi}{2L} x\Big),\sqrt{\frac{2}{3L}}\sin\Big(\frac{(2k-1)\pi}{2L}x\Big),\sqrt{\frac{2}{3L}}\sin\Big(\frac{(2k-1)\pi}{2L}x\Big)\Big),\ \ \ \ \ \forall k\in\N^*.$$
	When $c=0$, we obtain the eigenvalues $\big(\frac{k^2\pi^2}{L^2}\big)_{k\in\N^*}$ of multiplicity two that we associate to the couple of sequences of eigenfunctions  $(f^1_k)_{k\in\N^*}$ and $(f^2_k)_{k\in\N^*}$ such that, for every $k\in\N^*$, $$f^1_k:=\Big(-\sqrt{\frac{4}{3L}}\sin\Big(\frac{k\pi}{L}x\Big),\sqrt{\frac{1}{3L}}\sin\Big(\frac{k\pi}{L}x\Big),\sqrt{\frac{1}{3L}}\sin\Big(\frac{k\pi}{L}x\Big)\Big),$$
	$$f^2_k:=\Big(0,-\sqrt{\frac{1}{L}}\sin\Big(\frac{k\pi}{L}x\Big),\sqrt{\frac{1}{L}}\sin\Big(\frac{k\pi}{L}x\Big)\Big).$$
	Moreover,  $(f^1_k)_{k\in\N^*}\cup(f^2_k)_{k\in\N^*}\cup(g_k)_{k\in\N^*}$ is a Hilbert basis of $\Hi$ and $\big(\frac{k^2\pi^2}{L^2}\big)_{k\in\N^*}\cup\big(\frac{(2k-1)^2\pi^2}{4L^2}\big)_{k\in\N^*} $ are the eigenvalues of $A$ (not considering their multiplicity).

	\smallskip
	\noindent
	{\bf Validation of Assumptions I($\mathbf{\upvarphi,1}$).} We reorder $(f^1_k)_{k\in\N^*}\cup(g_k)_{k\in\N^*}$ 
in $\upvarphi=(\varphi_k)_{k\in\N^*}$. The point {\bf 1.\,\,}of Assumptions I($\upvarphi,1$) is verified since 
there exists $C_1,C_2>0$ such that
	$$|\la\varphi_1,B\varphi_k\ra_{L^2}|\geq \frac{C_1\sqrt{\mu_k}\sqrt{\mu_1}}{({\mu_k}-{\mu_1})^2}\geq 
\frac{C_2}{k^{3}},\ \ \ \ \ \ \forall k\in\N^*.$$ Subsequently, there exist $C_3,C_4>0$ so that 
$B_{k,k}:=\la\varphi_k,B\varphi_k\ra_{L^2}=C_3+{C_4}{k^{-2}}$ for every $k\in\N^*$ and $\mu_k=\frac{\pi^2k^2}{4 
L^2}$. Now, if $\mu_j-\mu_k-\mu_l+\mu_m=\frac{\pi^2}{4L^2}(j^2-k^2-l^2+m^2)=0$ with $(k,j),(m,n)\in I$ and 
$(k,j)\neq(m,n)$, then $$B_{j,j}-B_{k,k}- B_{l,l}+ 
B_{m,m}=C_4(j^{-2}-k^{-2}-l^{-2}+m^{-2})\neq 0,$$
	which implies the point {\bf 2.\,\,}of Assumptions I($\upvarphi,1$).
	
	\smallskip
	\noindent
	{\bf Validation of Assumptions II($\mathbf{\upvarphi,1,0}$) and conclusion.} The operator $B$ stabilizes the 
spaces $H^m$ with $m>0$ and $\overline{\spn\{\ffi_k:\ k\in\N^*\}}^{ L^2}\cap H^2_\Gi$, ensuring the point {\bf 
1.\,\,}of Assumptions II($\upvarphi,1,0)$.
	Since $$\inf_{j,k\in\N^*}|\mu_k-\mu_j|=\frac{\pi^2}{4L^{2}},$$ the point {\bf 2.\,\,}of Theorem $\ref{globalenergetic}$ holds and the global exact controllability is proved in $H^3_\Gi\cap\widetilde\Hi$. As $\ffi_k\in H^3_\Gi\cap\widetilde\Hi$ for every $k\in\N^*$, the energetic controllability follows in $\big(\frac{k^2\pi^2}{4L^2}\big)_{k\in\N^*}.$

	\smallskip
		When $N>3$, the spectrum contains simple eigenvalues relative to some eigenfunctions $(g_k)_{k\in\N^*}$ and multiple eigenvalues, each one corresponding to $N-1$ eigenfunctions $\{f_{k;j}\}_{l\leq N-1}$ with $k\in\N^*$. For each $k\in\N^*$, we construct $\{f_{k;j}\}_{l\leq N-1}$ such that only the functions $\{f_{k;j}\}_{l\leq N-2}$ vanish in $e_1$. We reorder $(f_{k;N-1})_{k\in\N^*}\cup(g_k)_{k\in\N^*}$ in $\upvarphi=(\varphi_k)_{k\in\N^*}$ and the proof is achieved as done for $N=3$.\qedhere
\end{proof}
\begin{teorema}\label{bim2.1}
	Let $\Gi$ be a star graph equipped with ($\Di$/$\NN$). Let $\Gi$ contain two edges $e_1$ and $e_2$ of length 
$1$ and 
connected to two external vertices both equipped with ($\Di$). Let $B$ be such that
$$B\psi=\big(x^2(\psi^{1}(x)-\psi^{2}(x)),x^2(\psi^{2}(x)-\psi^{1}(x)),0,...,0\big),\ \ \ \ \ \forall\psi\in\Hi.$$
	There exists $(\varphi_{k})_{k\in\N^*}\subset(\phi_k)_{{k\in\N^*}}$ such that the $(\ref{mainx1})$ is globally exactly controllable in $H^{3}_\Gi\cap\widetilde\Hi$ and energetically controllable in $({k^2\pi^2})_{k\in\N^*}.$ 
\end{teorema}
	\begin{figure}[H]
		\centering
		\includegraphics[width=\textwidth-150pt]{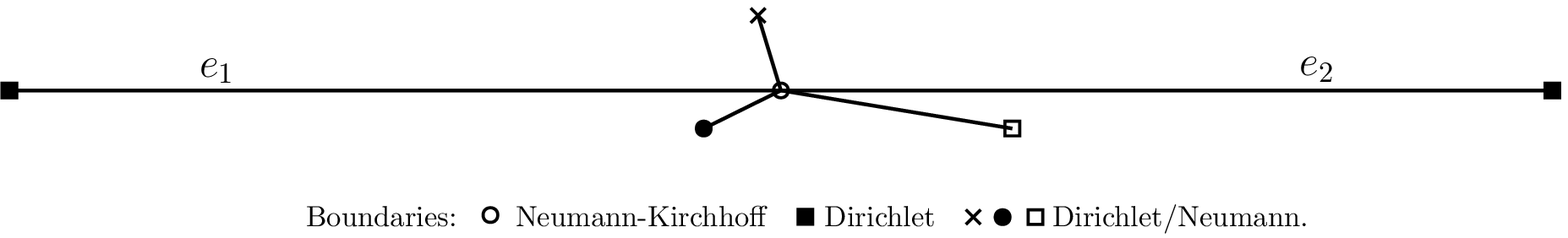}
		\caption{Example of star graph described by Theorem $\ref{bim2.1}$ with $N=5$.}
	\end{figure}

\begin{proof}
	Let $(\mu_{k})_{k\in\N^*}$ and $(\ffi_k)_{k\in\N^*}$ be so that $\mu_{k}=k^2\pi^2$, $\ffi_{k}^1=-\ffi_{k}^2=\sin({k\pi}x)$ and $\ffi_{k}^l=0$ with $k\in\N^*$ and $3\leq l\leq N.$
The claim follows from the point {\bf 2.\,\,}of Theorem $\ref{globalenergetic}$ with $\tilde d=0$ as Theorem 
$\ref{bim2}$. \qedhere
\end{proof}

\begin{teorema}\label{bim2.2}
	Let $\Gi$ be a star graph equipped with $(\Di)$ and composed by $\frac{N}{2}$ couples of edges $\{e_{2j-1}, 
e_{2j}\}_{j\leq \frac{N}{2}}$ of lengths $\{L_j\}_{j\leq \frac{N}{2}}\in\AL\LL(\frac{N}{2})$ with $N\in 2\N^*$. 
Let 
$B$ be such that $B\psi=((B\psi)^1,...,(B\psi)^N)$ for $\psi\in\Hi$ and 
$$(B\psi)^{2j}=-(B\psi)^{2j-1}=\sum_{l=1}^{N/2}{\color{black}\frac{L_l^{{3}/{2}}}{L_j^{{3}/{2}}}}x^2\Big(\psi^{2l} 
\Big(\frac{L_l}{L_j}x\Big)-\psi^{2l-1}\Big(\frac{L_l}{L_j}x\Big)\Big),\ \ \  \ \ \forall j\leq \frac{N}{2}.$$
	 There exists $\CC\subset (\R^+)^N$ countable so that, for every $\{L_j\}_{j\leq N}\in\AL\LL(N)\setminus \CC$, there exists $(\varphi_{k})_{k\in\N^*}\subseteq(\phi_k)_{{k\in\N^*}}$ such that $(\ref{mainx1})$ is globally exactly controllable in $H^{3+\epsilon}_\Gi\cap\widetilde\Hi$ with $\epsilon>0$ and energetically controllable in $\big(\frac{k^2\pi^2}{L_j^2}\big)_{\underset{j\leq N/2}{k,j\in\N^*}}.$
\end{teorema}

	\begin{figure}[H]
		\centering
		\includegraphics[width=\textwidth-150pt]{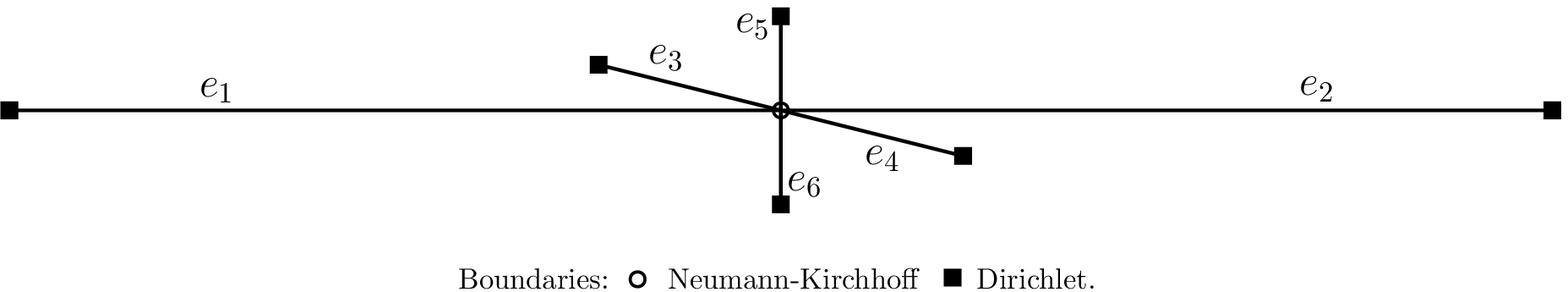}
		\caption{Example of star graph described by Theorem $\ref{bim2.2}$ with $N=6$.}
	\end{figure}
	
\begin{proof}
Let $(\mu_k)_{k\in\N^*}\subset(\lambda_k)_{k\in\N^*}$ be eigenvalues obtained by reordering $\big(\frac{k^2\pi^2}{L_j^2}\big)_{{k\in\N^*}}$ for every $j\leq N/2$ and $(\ffi_k)_{k\in\N^*}$ be an orthonormal system of $\Hi$ made by corresponding eigenfunctions.
	For $k\in\N^*$, there exist $m(k)\in\N^*$ and $l(k)\leq N/2$ so that $\ffi_k^{n}\equiv 0$ for $n\neq 2l(k),$ $2l(k)-1$ and
	$$\mu_k=\frac{m(k)^2\pi^2}{L^{2}_{l(k)}},\ \  \ \ \ffi_k^{2l(k)-1}(x)=-\ffi_k^{2l(k)}(x)=\sqrt{\frac{1}{L_{l(k)}}}\sin{(\sqrt{\mu_k} x)}.$$
	Let $[r]$ be the entire part of $r\in\R^+$. For $k\in\N^*$ and $C=4\min_{l\leq N}L_l$, we have
\begin{equation*} \begin{split}
&|\la\ffi_1, B\ffi_k\ra_{L^2}|=\Bigg|\sum_{l=1}^N\Bigg\la \ffi_k^{l}(x),\sum_{n=1}^{N/2}{\color{black}\frac{L_{n}^\frac{3}{2}x^2}{L_{[(l+1)/2]}^\frac{3}{2}}}\big(\ffi_1^{2n-1}\Big(\frac{L_{n}}{L_{[(l+1)/2]}}x\Big)-\ffi_1^{2n}\Big(\frac{L_{n}}{L_{[(l+1)/2]}}x\Big)\Big)\Bigg\ra_{L^2(e_{l})}\Bigg|\\
	&=\Big|\int_0^{L_{l(k)}} 
{\color{black}\frac{4x^2L_{l(1)}^\frac{3}{2}}{L_{l(k)}^\frac{3}{2}}}\sin\Big(\frac{m(1)\pi 
x}{L_{l(k)}}\Big)\sin\Big(\frac{m(k)\pi x}{L_{l(k)}} \Big) dx\Big|\geq C\Big|\int_0^1 x^2\sin(m(1)\pi 
x)\sin(m(k)\pi x) dx\Big|.\\
	\end{split}\end{equation*} 
	Assumptions I($\upvarphi,1$) and Assumptions II($\upvarphi,1,\epsilon$) with $\epsilon\in(0,\frac{1}{2})$ hold as in the proofs of Theorem $\ref{bim.2}$ and Theorem $\ref{bim2}$.
We consider the techniques adopted in the proof of $\cite[Lemma\ 2.6]{mio3}$ which are due to the Roth's Theorem 
$\cite[Proposition\ A.1]{mio3}$. For every $\epsilon>0,$ there exists $C>0$ so that $$|\mu_{k+1}-\mu_k|\geq 
{C}{k^{-\epsilon}},\ \ \ \  \ \ \forall k\in\N^*.$$
	The claim follows from the hypotheses {\bf 2.\,\,}of Theorem $\ref{globalenergetic}$ with $\tilde 
d=\epsilon>0$.\qedhere
\end{proof}

\begin{teorema}\label{bim1.1}
	Let $\Gi$ be a compact quantum graph. Let the first $\widetilde N\leq N$ edges $\{e_j\}_{j\leq\widetilde N}$ of the graph be loops of lengths $\{L_j\}_{j\leq \widetilde N}$ ({\it e.g} Figure $\ref{cappio}$). For $\psi=(\psi^1,...\psi^N)$, let $B$ be such that $$(B\psi)^l=\sum_{j\leq \widetilde N}{\color{black}\frac{L_j^{{3}/{2}}}{L_l^{{3}/{2}}}}x^2\Big(\frac{x}{L_l}-1\Big)\psi^j\Big(\frac{L_j}{L_l}x\Big),\ \ \ \ \ \ \  \ \ (B\psi)^m\equiv 0,\ \ \ \ \ \forall l\leq \widetilde N,\ \ \widetilde N<m\leq N.$$ 
There exists $\CC\subset (\R^+)^{\widetilde N}$ countable so that, if $\{L_j\}_{j\leq {\widetilde 
N}}\in\AL\LL({\widetilde N})\setminus \CC$, then there exists 
$(\varphi_{k})_{k\in\N^*}\subseteq(\phi_k)_{{k\in\N^*}}$ such that $(\ref{mainx1})$  
	is {globally exactly controllable} in
	$H^{3+\epsilon}_{\Gi}\cup\widetilde\Hi$ with $\epsilon>0$ and energetically controllable in $\big(\frac{k^2\pi^2}{L_j^2}\big)_{\underset{j\leq \widetilde N}{k,j\in\N^*}}.$
\end{teorema}
\begin{proof}
	Let $(\ffi_k)_{k\in\N^*}$ be such that, for each $k\in\N^*$, there exist $m(k)\in\N^*$ and $l(k)\leq \widetilde N$ such that $\mu_k=\frac{4m(k)^2\pi^2}{L^{2}_{l(k)}},$ $\ffi_k^{l(k)}(x)=\sqrt{\frac{2}{L_{l(k)}}}\sin{(\sqrt{\mu_k} x)}$ and $\ffi_k^{n}\equiv 0$ for every $n\neq l(k)$ and $n\leq N$.
	Now, $(\ffi_k)_{k\in\N^*}$ is an orthonormal system made by eigenfunctions of $A$ and the claim yields as Theorem $\ref{bim2.2}$.\qedhere
	\end{proof}

\smallskip
\noindent
{\bf Acknowledgments.} The author would like to thank the
referees for the constructive comments which improved the quality
of redaction. He is also grateful to Olivier Glass and Nabile Boussa\"id 
for having carefully reviewed 
this work. Finally, he thanks Ka\"is Ammari for suggesting him the problem and the colleagues Andrea 
Piras, Riccardo 
Adami, Enrico Serra and Paolo Tilli for the fruitful conversations.

\appendix\section{Appendix: Global approximate controllability}\label{approximate}
In the current appendix, for the sake of completeness, we propose the global approximate controllability result 
provided in \cite[Section\ 5.2]{mio3}. The outcome is adopted in the proof of Theorem \ref{global}.
\begin{defi}
	The \eqref{mainx1} is said to be globally approximately controllable in $H_{\Gi}^{s}$ with $s>0$ when, for every $\psi\in H^{s}_{\Gi}$, $\widehat\G\in U(\Hi)$ such that $\widehat\G\psi\in H^{s}_{\Gi}$ and $\varepsilon>0$, there exist $T>0$ and $u\in L^2((0,T),\R)$ such that $\|\widehat\G\psi-\G^u_T\psi\|_{(s)}<\varepsilon$.
\end{defi}
\begin{prop}\label{approx}
Let $(A,B)$ satisfy Assumptions I$(\eta)$ and Assumptions II$(\eta,\tilde d)$ for $\eta>0$ and $\widetilde d\geq 0$. The (\ref{mainx1}) is globally approximately controllable in $H^{s}_{\Gi}$ for $s=2+d$ with $d$ from Assumptions II$(\eta,\tilde d)$ .
\end{prop}
\begin{proof}
In the point {\bf 1)\,\,}of the proof, we suppose that $(A,B)$ admits a non-degenerate chain of connectedness (see \cite[Definition\ 3]{nabile}). We treat the general case in the point {\bf 2)\,\,}.

\smallskip

\needspace{3\baselineskip}

\noindent
{\color{black}{\bf 1) (a) Preliminaries.} Let $\pi_m$ be the orthogonal projector onto 
$\Hi_m:=\spn{\{\phi_j\ :\ j\leq m\}}$ with $m\in\N^*.$
Up to reordering $(\phi_k)_{k\in\N^*}$, the couples $(\pi_{m}A\pi_{m},\pi_{m}B\pi_{m})$ for $m\in\N^*$ admit 
non-degenerate chains of connectedness in $\Hi_{m}$. Let $\|\cdot\|_{BV(T)}=\|\cdot\|_{BV((0,T),\R)}$ and 
$\iii\cdot\iii_{(s)}:=\iii\cdot\iii_{L(H^s_{\Gi},H^s_{\Gi})}$ for $s>0.$ For $N\in\N^*$, we denote 
$SU(\Hi_{N})=\big\{\Gamma\in U(\Hi_{N})\ :\ (\la\phi_j,\Gamma\phi_k\ra_{L^2})_{j,k\leq N}\in 
SU(N)\big\}.$
\begin{itemize}
	\item[]	 {\bf Claim.} $\forall\ \widehat\G\in U(\Hi),\ \forall \varepsilon>0,\ \exists N_1\in\N^*,\ \widetilde\G_{N_1}\in U(\Hi)\ :\ \pi_{N_1}\widetilde\G_{N_1}\pi_{N_1}\in SU(\Hi_{N_1}),$
	\begin{equation}\label{grammo}\|\widetilde\G_{N_1}\phi_1-\widehat\Gamma\phi_1\|_{L^2}<\varepsilon.\end{equation}
\end{itemize}
Let $N_1\in\N^*$ and $\widetilde\phi_1:=\|\pi_{N_1}\widehat\G\phi_1\|_{L^2}^{-1}\pi_{N_1}\widehat\G\phi_1$. We 
define $(\widetilde\phi_j)_{2\leq j\leq N_1}$ such that $(\widetilde\phi_j)_{j\leq N_1}$ is an orthonormal basis of 
$\Hi_{N_1}$ and we complete it to $(\widetilde\phi_j)_{j\in\N^*}$, a Hilbert basis of $\Hi$. The operator 
$\widetilde\G_{N_1}$ is the unitary map such that $\widetilde\G_{N_1}\phi_j=\widetilde\phi_j$ for every $j\in\N^*$.	
The provided definition implies 
$\lim_{N_1\rightarrow\infty}\|\widetilde\G_{N_1}\phi_1-\widehat\Gamma\phi_1\|_{L^2}=0$. Thus, for every 
$\varepsilon>0$, there exists $N_1\in\N^*$ large enough satisfying the claim.}

\medskip
\needspace{3\baselineskip}

\noindent
{\bf 1) (b)  Finite dimensional controllability.} Let $M_{ad}$ be the set of $(j,k)\in\{1,...,N_1\}^2$ such that $B_{j,k}:=\la\phi_j,B\phi_k\ra_{L^2}\neq 0$ and $|\lambda_j-\lambda_k|=|\lambda_m-\lambda_l|$ with $m,l\in\N^*$ implies $\{j,k\}= \{m,l\}$ for $B_{m,l}=0$.
For every $(j,k)\in\{1,...,N_1\}^2$ and $\theta\in[0,2\pi)$, we define $E_{j,k}^{\theta}$ the $N_1\times N_1$ matrix with elements $(E_{j,k}^\theta)_{l,m}=0$, $(E_{j,k}^\theta)_{j,k}=e^{i\theta}$ and $(E_{j,k}^\theta)_{k,j}=-e^{-i\theta}$ for $(l,m)\in\{1,...,N_1\}^2\setminus\{(j,k),(k,j)\}.$ Let $E_{ad}=\big\{E_{j,k}^{\theta}\ :\ (j,k)\in M_{ad},\ \theta\in[0,2\pi)\big\}.$ 
Let $L_1:=E_{ad}$. We define by iteration $L_m:=[E_{ad},L_{m-1}]+L_{m-1}$ with $m\in\N^*\setminus\{1\}$ and there 
exists $\widetilde m\in\N^*$ such that $L_{\widetilde m}=L_{\widetilde m+1}$ as $E_{ad}$ is composed by $N_1\times 
N_1$ matrices. We denote $Lie(E_{ad})=L_{\widetilde m}$. Fixed $v$ a piecewise constant control taking value in 
$E_{ad}$ and $\tau>0$, we introduce the control system on $SU({N_1})$
\begin{equation}\label{formulapprox3}\begin{split}\begin{cases}
\dot{x}(t)=x(t)v(t),\ \ \ \ \ \ t\in(0,\tau),\\
x(0)=Id_{SU({N_1})}.\\
\end{cases}\end{split}\end{equation}

\begin{itemize}
	\item[] {\bf Claim.} $(\ref{formulapprox3})$ is controllable, {\it i.e.} for $R\in SU({N_1})$, there exist 
$p\in\N^*$, $M_1,...,M_p\in E_{ad}$, $\alpha_1,...,\alpha_p\in\R^+$ such that $R=e^{\alpha_1 M_1}\circ...\circ 
e^{\alpha_p M_p}.$
\end{itemize}
For every $(j,k)\in\{1,...,N_1\}^2$, we define the $N_1\times N_1$ matrices $R_{j,k}$, $C_{j,k}$ and $D_{j}$ as follow. For $(l,m)\in\{1,...,N_1\}^2\setminus\{(j,k),(k,j)\},$ we have $(R_{j,k})_{l,m}=0$ and $(R_{j,k})_{j,k}=-(R_{j,k})_{k,j}=1,$ while $(C_{j,k})_{l,m}=0$ and $(C_{j,k})_{j,k}=(C_{j,k})_{k,j}=i.$ Moreover, for $(l,m)\in\{1,...,N_1\}^2\setminus\{(1,1),(j,j)\},$ $(D_{j})_{l,m}=0$ and $(D_{j})_{1,1}=-(D_{j})_{j,j}=i.$
We consider the basis of $su({N_1})$, the Lie algebra of 
$SU({N_1})$, $$\{R_{j,k}\}_{j,k\leq N_1}\cup\{C_{j,k}\}_{j,k\leq 
N_1}\cup\{D_{j}\}_{j\leq N_1}.$$ Thanks to \cite[Theorem\ 6.1]{sac}, the controllability of (\ref{formulapprox3}) 
is equivalent to prove that $Lie(E_{ad})\supseteq su({N_1})$. 
The claim is valid since it is possible to obtain the matrices $R_{j,k}$, $C_{j,k}$ and $D_j$ for every $j,k\leq N_1$ by iterated Lie brackets of elements in $E_{ad}$ as follows.

\begin{itemize}
	\item For every $(j,k)\in M_{ad}$, we have $R_{j,k}=E_{j,k}^0$ and $C_{j,k}=E_{j,k}^{\frac{\pi}{2}}$. For 
every 
$(j,k)\not\in M_{ad}$ such that there exists $j_1\leq N_1$ so that $(j,j_1),(j_1,k)\in M_{ad}$, we have 
$R_{j,k}=[E_{j,j_1}^0,E_{j_1,k}^0]$ and $C_{j,k}=[E_{j,j_1}^0,E_{j_1,k}^{\frac{\pi}{2}}]$.
	\item If $(1,j)\in M_{ad}$, then $2D_j=[E_{1,j}^0,E_{1,j}^{\frac{\pi}{2}}]$, while if $(1,j)\not\in M_{ad}$ 
and there exists $j_1\leq N_1$ such that $(1,j_1),(j_1,j)\in M_{ad}$, then
	$-2D_j=\Big[[E_{1,j_1}^{\frac{\pi}{2}},E_{j_1,j}^{\frac{\pi}{2}}],[E_{1,j_1}^{0},E_{j_1,j}^{\frac{\pi}{2}}]\Big].$
	\item For every $(j,k)\not\in M_{ad}$, there exist $m\leq N_1$ and $\{j_l\}_{l\leq m}$ such that 
$(j,j_1),...,(j_m,k)\in M_{ad}.$ We call $S=\{(j,j_1),...,(j_m,k)\}.$ By repeating the previous point, we can 
generate each $R_{j,k}$, $C_{j,k}$ and $D_j$ with $j,k\in\{1,...,N_1\}$ by iterating Lie brackets of 
$E_{l,m}^{\theta}$ for $(l,m)\in S$ and $\theta\in[0,2\pi)$.	
	
	\end{itemize}

	\smallskip

	\needspace{3\baselineskip}
	
	\noindent
	{\bf 1) (c) Finite dimensional estimates.} The previous claim and the fact that
$(\la\phi_j,\widetilde\G_{N_1}\phi_k\ra_{L^2})_{j,k\leq N_1}\in SU({N_1})$ ensure the existence of $p\in\N^*$, 
$M_1,...,M_p\in E_{ad}$ and $\alpha_1,...,\alpha_p\in\R^+$ such that
	\begin{equation}\label{decomposition_rotations}(\la\phi_j,\widetilde\G_{N_1}\phi_k\ra_{L^2})_{j,k\leq 
N_1}=e^{\alpha_1 
M_1}\circ...\circ e^{\alpha_p M_p}.\end{equation}
For every $l\leq p$, we call
$\widehat\Gamma_l$ the operator in $SU(\Hi_{N_1})$ such that $(\la\phi_j,\widehat\Gamma_l\phi_k\ra)_{j,k\leq 
N_1}=e^{\alpha_l M_l}$. The identity \eqref{decomposition_rotations} yields
\begin{equation}\label{dexk}\pi_{N_1}\widetilde\G_{N_1}\pi_{N_1}
=\widehat\Gamma_1\circ...\circ\widehat\Gamma_p.\end{equation}
	
	\begin{itemize}\item[] {\bf Claim.} For every $l\leq p$, there exist $\{T_n^l\}_{l\in\N^*}\subset\R^+$ and 
$\{u_n^l\}_{n\in\N^*}$ such that $u_n^l\in L^2((0,T_n^l),\R)$ for every $n\in\N^*$ and
		
\begin{equation}\label{sorde}\lim_{n\rightarrow\infty}\|\G_{T_n^l}^{u_n^l}\phi_k-\widehat\Gamma_l\phi_k\|_{L^2}=0,\ 
\ \ \  \ \ \forall k\leq N_1, \end{equation}
		\begin{equation}\label{sorde1}\begin{split}\sup_{n\in\N^*}\|u_n^l&
		\|_{BV(T_n^l)}<\infty,  \ \ \  \  \ \ 
\sup_{n\in\N^*}\|u_n^l\|_{L^\infty((0,T_n^l),\R)}<\infty,  \ \ \  \  \ \ \sup_{n\in\N^*} 
T_n^l\|u_n^l\|_{L^\infty((0,T_n^l),\R)}<\infty.\end{split}\end{equation}\end{itemize}
	We consider the results developed in $\cite[Section\ 3.1\ \&\ Section\ 3.2]{chambrion2}$ by Chambrion and 
leading to $\cite[Proposition\ 6]{chambrion2}$ since $(A,B)$ admits a non-degenerate chain of connectedness 
(\cite[Definition\ 3]{nabile}). Each $\widehat\Gamma_l$ is a rotation in a two dimensional space for every 
$l\in\{1,...,p\}$ and this work explicits $\{T_n^l\}_{n\in\N^*}\subset\R^+$ and $\{u_n^l\}_{n\in\N^*}$ satisfying 
$(\ref{sorde1})$ such that $u_n^l\in L^2((0,T_n^l),\R)$ for every $n\in\N^*$ and 
\begin{equation*}\lim_{n\rightarrow\infty}\|\pi_{N_1}\G_{T_n^l}^{u_n^l}\phi_k-\widehat\Gamma_l\phi_k\|_{L^2}=0,\ \ 
\ \ \ \ \forall k\leq N_1.\end{equation*}
Now, $\lim_{n\rightarrow\infty}\|\pi_{N_1}\G_{T_n^l}^{u_n^l}\phi_k\|_{L^2}=1$ since $\widehat\Gamma_l\in 
SU(\Hi_{N_1})$. Thus, $\lim_{n\rightarrow\infty}\|(1-\pi_{N_1})\G_{T_n^l}^{u_n^l}\phi_k\|_{L^2}=0$ and 
$$\lim_{n\rightarrow\infty}\|\G_{T_n^l}^{u_n^l}\phi_k-\widehat\Gamma_l\phi_k\|_{L^2}\leq\lim_{n\rightarrow\infty}
\Big(\|\pi_{N_1}\G_{T_n^l}^{u_n^l}\phi_k-\widehat\Gamma_l\phi_k\|_{L^2}+\|(1-\pi_{N_1})\G_{T_n^l}^{u_n^l}\phi_k\|_{
L^2}\Big) =0.$$

	\needspace{3\baselineskip}
	{\color{black}
	\noindent
	{\bf 1) (d) Infinite dimensional estimates.} 
	
	\begin{itemize}
		\item[] {\bf Claim.} Let $\widehat\G\in U(\Hi)$. There exist $K_1,K_2,K_3>0$ such that for every 
$\varepsilon>0$, there exist $T>0$ and $u\in L^2((0,T),\R)$ such that 
$\|\G_{T}^{u}\phi_1-\widehat\G\phi_1\|_{L^2}<\varepsilon$ and
		\begin{align}\label{casinooo}\|u\|_{BV(T)}\leq K_1,  \ \  \ \ \ \ \ \|u\|_{L^\infty((0,T),\R)}\leq K_2,\ \ 
 \ \ \ \ \ T\|u\|_{L^\infty((0,T),\R)}\leq K_3.\end{align}
	\end{itemize}
	
	Let us assume that {\bf 1) (c)} be valid with $p=2$. Although, the following result is valid for any 
$p\in\N^*$. As $\widehat\Gamma_2\in SU(\Hi_{N_1})$, there exist $ l\leq N_1$ and $\alpha\in\C$ such 
that $|\alpha|=1$ and $\widehat\Gamma_2\phi_1=\alpha\phi_l$. Thanks to $(\ref{sorde})$, there exists $n\in\N^*$ 
large enough such that,
	\begin{equation*}\begin{split}&\| \G_{T_n^1}^{u_n^1}\G_{T_n^2}^{u_n^2}\phi_1-\widehat\Gamma_1 
\widehat\Gamma_2\phi_1\|_{L^2}
\leq \iii 
\G_{T_n^1}^{u_n^1}\iii\|\G_{T_n^2}^{u_n^2}\phi_1-\widehat\Gamma_2\phi_1\|_{L^2}+\|\G_{T_n^1}^{u_n^1}\alpha
\phi_l-\widehat\Gamma_1\alpha\phi_l\|_{L^2}< \varepsilon.\end{split}\end{equation*}
Thanks to $(\ref{dexk})$, there exist $K_1,K_2,K_3>0$ such that for every 
$\varepsilon>0$, there 
exist $T>0$ and $u\in L^2((0,T),\R)$ such that $\|\G_{T}^{u}\phi_1-\widetilde\G_{N_1}\phi_1\|_{L^2}<\varepsilon$ 
and satisfying \eqref{casinooo}. The 
relation $(\ref{grammo})$ and the triangular inequality achieve the claim.

\medskip
	\noindent
	{\bf 1) (e) Approximate controllability with respect to the $L^2$-norm.} Let $\psi\in \Hi$ and $\widehat \G\in U(\Hi)$.

	\begin{itemize}
		\item[] {\bf Claim.} There exist $K_1,K_2,K_3>0$ such that for every $\varepsilon>0$, there exist $T>0$ and 
$u\in L^2((0,T),\R)$ such that $\|\G_{T}^{u}\psi-\widehat\G\psi\|_{L^2}<\varepsilon$ and satisfying 
\eqref{casinooo}.
	\end{itemize}

We assume that $\|\psi\|_{L^2}=1$ but the same proof is also valid in the general case. We consider the 
unitary propagator $\widetilde \G_t^u$ describing the reversed dynamics of the \eqref{mainx1} introduced in 
Remark 
\ref{timereversibility}. We also recall the validity of the relation \eqref{inversepropagator}. Now, the 
results from \cite{chambrion2}, which are adopted in the point {\bf 1) (c)}, are also valid 
for the reversed dynamics. Thus, such as in {\bf 1) (d)}, it is also true that, for every $\widehat\G\in U(\Hi)$, 
there exist $K_1,K_2,K_3>0$ such that for every $\varepsilon>0$, there exist $T>0$ and $u\in L^2((0,T),\R)$ such 
that the relations \eqref{casinooo} are satisfied and $ 
\|\widetilde\G_{T}^{u}\phi_1-\widehat\G\phi_1\|_{L^2}<\varepsilon.$ By keeping in mind that $\widetilde 
\G_{T}^{u}=(\G_{T}^{u})^{-1}$ and $\iii\G_{T}^{u}\iii=1$, we have
$$\|\phi_1-\G_{T}^{u}\widehat\G\phi_1\|_{L^2}=\|\G_{T}^{u}\widetilde\G_{T}^{u}\phi_1-\G_{T}^{u}
\widehat\G\phi_1\|_{L^2} \leq 
\iii\G_{T}^{u}\iii\|\widetilde\G_{T}^{u}\phi_1-\widehat\G\phi_1\|_{L^2}<\varepsilon.$$
The last relation guarantees that, for 
every $\psi\in \Hi$ such that $\|\psi\|_{L^2}=1$, there exist $T_1>0$ and $u_1\in 
L^2((0,T_1),\R)$ such that
	$$\|\G_{T_1}^{u_1}\psi-\phi_1\|_{L^2}<{\varepsilon}.$$ Now, {\bf 1) (d)} ensures that, for 
every $\widehat \G\in U(\Hi)$, there exist $ T_2>0$ and 
$u_2\in L^2((0,T_2),\R)$ such that
	$$\|\G_{T_2}^{u_2}\phi_1-\widehat\G\psi\|_{L^2}<{\varepsilon}.$$
	The chosen controls $u_1$ and $u_2$ satisfy \eqref{casinooo}. In conclusion, the claim is proved as
	\begin{align*}\|\G_{T_2}^{u_2}\G_{T_1}^{u_1}\psi-\widehat\G\psi\|_{L^2}&\leq \|\G_{T_2}^{u_2}\G_{T_1}^{u_1}\psi-\G_{T_2}^{u_2}\phi_1\|_{L^2}+\|\G_{T_2}^{u_2}\phi_1-\widehat\G\psi\|_{L^2}
	< 2{\varepsilon}.\end{align*}

	\noindent
	{\bf 1) (f) Global approximate controllability in higher regularity norms.} Let $\psi\in H^{s}_{\Gi}$ with 
$s\in[s_1,s_1+2)$ and $s_1\in\N^*$. Let $\widehat \G\in U(\Hi)$ be such that $\widehat \G\psi\in H^{s}_{\Gi}$ and 
$B:H^{s_1}_{\Gi}\longrightarrow H^{s_1}_{\Gi}$.
	\begin{itemize}
		\item[] {\bf Claim.} There exist $T>0$ and $u\in L^2((0,T),\R)$ such that 
$\|\G_{T}^{u}\psi-\widehat\G\psi\|_{(s)}<\varepsilon$.
	\end{itemize}
	
We notice that the 
operator $-i(A+u(t)B-ic)$ is dissipative in $H^{s_1}_{\Gi}$ for
$c:=\|u\|_{L^\infty((0,T),\R)}\iii B\iii_{(s_1)}$. Indeed, for every 
$\lambda>0$ and $\psi\in H^{s_1+2}_{\Gi}$, we have
$$\|(\lambda+i(A+u(t)B-ic))\psi\|_{(s_1)}\geq 
\|(\lambda+c+iA)\psi\|_{(s_1)}-\|u\|_{L^\infty((0,T),\R)}\iii B\iii_{(s_1)}\|\psi\|_{(s_1)}.$$
Now, the operator $A$ with domain $H^{s_1+2}_\Gi$ is self-adjoint in the Hilbert space $H^{s_1}_{\Gi}$ and we have 
the inequality $\|(\lambda+c+iA)\psi\|_{(s_1)}\geq (\lambda+c)\|\psi\|_{(s_1)}$. By recalling that
$c=\|u\|_{L^\infty((0,T),\R)}\iii B\iii_{(s_1)}$, we obtain
\begin{align*}\|(\lambda-i(-A+u(T-t)B-ic))\psi\|_{(s_1)}\geq 
 (\lambda+c)\|\psi\|_{(s_1)}-c\|\psi\|_{(s_1)}= 
\lambda\|\psi\|_{(s_1)}.\end{align*}
Thus, $-i(A+u(t)B-ic)$ is dissipative and the Kato-Rellich's Theorem 
yields that it is also maximal dissipative. We consider the propagation of regularity developed by Kato in the 
work $\cite{kato1}$.
Let $\lambda>c$ and $\widehat H^{s_1+2}_{\Gi}:=D(A^\frac{s_1}{2}(i\lambda-A))\equiv  H^{s_1+2}_{\Gi}$. We know 
that $B:\widehat H^{s_1+2}_{\Gi}\subset  H^{s_1}_{\Gi}\rightarrow H^{s_1}_{\Gi}$ and the arguments of 
\cite[Remark\ 2.1]{mio1} imply that $B\in L(\widehat H^{s_1+2}_{\Gi},H^{s_1}_{\Gi})$. For $T>0$ and $u\in 
BV((0,T),\R)$, we have $$
	M:=\sup_{t\in [0,T]}\iii(i\lambda-A-u(t)B)^{-1}\iii_{L(H^{s_1}_{\Gi},\widehat H^{s_1+2}_{\Gi})}<+\infty.$$
We know $\|k+f(\cdot)\|_{BV((0,T),\R)}=\|f\|_{BV((0,T),\R)}$ for $f\in BV((0,T),\R)$ and $k\in\R$. Equivalently, $$N:=\iii i\lambda-A-u(\cdot)B\iii_{BV\big([0,T],L(\widehat H^{s_1+2}_{\Gi},H^{s_1}_{\Gi})\big)}=\|u\|_{BV(T)} \iii B\iii_{L(\widehat H^{s_1+2}_{\Gi},H^{s_1}_{\Gi})}<+\infty.$$ We call $C_1:=\iii A(A+u(T)B-i\lambda)^{-1}\iii_{(s_1)}<\infty$ and $U_t^{u}$ the propagator generated by $A+uB-ic$ such that $U_t^u\psi=e^{-ct}\G_t^u\psi$. Thanks to $\cite[Section\ 3.10]{kato1}$, for every $\psi\in H^{s_1+2}_{\Gi}$, it follows 
\begin{equation*}\begin{split}\|(A+u(T)B-i\lambda) U_t^u \psi\|_{(s_1)}\leq Me^{MN}\|(A-i\lambda) \psi\|_{(s_1)}\ \ \ \Longrightarrow \ \ \ \|\G_{T}^{u} \psi\|_{(s_1+2)}&\leq C_1 Me^{MN+cT}\|\psi\|_{(s_1+2)}.\end{split}\end{equation*}
For every $T>0$, $u\in BV((0,T),\R)$ and $\psi\in H^{s_1+2}_{\Gi}$, there exists $C=C(K)>0$ depending on $K=\big(\|u\|_{BV(T)},\|u\|_{L^\infty((0,T),\R)},T\|u\|_{L^\infty((0,T),\R)}\big)$ such that \begin{equation}\label{diid1}\|\G_{T}^{u} \psi\|_{(s_1+2)}\leq  C\|\psi\|_{(s_1+2)}.\end{equation}
Now, we notice that, for every $\psi\in H^{6}_{\Gi}$, we have 
$\|A\psi\|_{L^2}^2\leq\|\psi\|_{L^2}\|A^2\psi\|_{L^2}$ from the Cauchy-Schwarz inequality and there exists $C_2>0$ 
such that $\|A^2\psi\|_{L^2}^4\leq\|A\psi\|_{L^2}^2\|A^3\psi\|_{L^2}^2\leq C_2\|\psi\|_{L^2}\|A^3\psi\|_{L^2}^3$. 
By following the same idea, for every $\psi\in H^{s_1+2}_{\Gi}$, there exist $m_1,m_2\in\N^*$ and $C_3,C_4>0$ such 
that
\begin{equation}\label{diid2}\|A^\frac{s}{2}\psi\|_{L^2}^{m_1+m_2}\leq C_3\|\psi\|_{L^2}^{m_1}\|A^\frac{s_1+2}{2}\psi\|_{L^2}^{m_2} \ \ \ \ \ \Longrightarrow \ \ \ \ \|\psi\|_{(s)}^{m_1+m_2}\leq C_4\|\psi\|_{L^2}^{m_1}\|\psi\|_{(s_1+2)}^{m_2}.\end{equation}
	In conclusion, the point {\bf 1) (e)}, the relation $(\ref{diid1})$ and the relation $(\ref{diid2})$ ensure the claim.
	}

\smallskip

\noindent
{\bf 1) (g)  Conclusion.} Let $d$ be defined in Assumptions II$(\eta,\tilde d)$. 
If $d<2$, then $B:H^{2}_{\Gi}\rightarrow H^{2}_{\Gi}$ and the global approximate controllability is verified in $H^{d+2}_{\Gi}$ since $d+2<4.$ 
If $d\in [2,5/2)$, then $B:H^{d_1}\rightarrow H^{d_1}$ with $d_1\in(d,5/2)$ from Assumptions II$(\eta,\tilde d)$. 
Now, $H^{d_1}_{\Gi}=H^{d_1}\cap H^2_{\Gi}$, thanks to $\cite[Proposition\ 4.2]{mio3}$, and 
$B:H^{2}_{\Gi}\rightarrow H^{2}_{\Gi}$ implies $B:H^{d_1}_{\Gi}\rightarrow H^{d_1}_{\Gi}$. 
The global approximate controllability is verified in $H^{d+2}_{\Gi}$ since $d+2<d_1+2.$
If $d\in [5/2,7/2)$, then $B:H_{\NN\KK}^{d_1}\rightarrow H_{\NN\KK}^{d_1}$ for $d_1\in(d,7/2)$ and 
$H^{d_1}_{\Gi}=H^{d_1}_{\NN\KK}\cap H^2_{\Gi}$ from $\cite[Proposition\ 4.2]{mio3}$. Now, $B:H^{2}_{\Gi}\rightarrow 
H^{2}_{\Gi}$ that implies $B:H^{d_1}_{\Gi}\rightarrow H^{d_1}_{\Gi}$. 
The global approximate controllability is verified in $H^{d+2}_{\Gi}$ since $d+2<d_1+2.$

\smallskip
\noindent
{\bf 2) Generalization.} Let $(A,B)$ do not admit a non-degenerate chain of connectedness. We decompose 
$$A+u(\cdot)B=(A+u_0B)+u_1(\cdot)B,\ \ \ \ \ \ \  \ \ \ \ \ \  \ u_0\in \R,\ \ \ \ u_1\in L^2((0,T),\R).$$ We 
notice that, if $(A,B)$ satisfies Assumptions I$(\eta)$ and Assumptions II$(\eta,\tilde d)$ for $\eta>0$ and 
$\tilde d\geq 0$, then $\cite[Lemma\ C.2\ \&\ Lemma\ C.3]{mio3}$ are valid. We consider $u_0$ belonging to the 
neighborhoods provided by $\cite[Lemma\ C.2\ \&\ Lemma\ C.3]{mio3}$ and we denote $(\phi_k^{u_0})_{k\in\N^*}$ a 
Hilbert basis of $\Hi$ made by eigenfunctions of $A+u_0B$. The steps of the point {\bf 1)\,\,}can be repeated by 
considering the sequence $(\phi_k^{u_0})_{k\in\N^*}$ instead of $(\phi_k)_{k\in\N^*}$ and the spaces 
$D(|A+u_0B|^\frac{s_1}{2})$ in substitution of $H^{s_1}_\Gi$ with $s_1>0$. Thanks to the mentioned results, the 
claim is proved since $(A+u_0B,B)$ admits a non-degenerate chain of connectedness and 
$\big\||A+u_0B|^\frac{s_1}{2}\cdot\big\|_{L^2}\asymp\|\cdot\|_{(s_1)}$ with $s_1\in [s,s+2)$, $s=2+d$ and $d$ from 
Assumptions II$(\eta,\tilde d)$. \qedhere

\end{proof}

By referring to Remark 
\ref{timereversibility}, we consider the reversed dynamics of the \eqref{mainx1} and the unitary 
propagator 
$\widetilde \G_t^u$ generated by the time-dependent 
Hamiltonian $-A-u(T-t)B$ 
with $u\in L^2((0,T),\R)$. We also recall the validity of the relation \eqref{inversepropagator}. In the following 
corollary, we present the global approximate 
controllability for the reversed dynamics.

\begin{coro}\label{approxreversed}
Let $(A,B)$ satisfy Assumptions I$(\eta)$ and Assumptions II$(\eta,\tilde d)$ for $\eta>0$ and $\widetilde d\geq 
0$. Let $s=2+d$ be defined by $d$ from 
Assumptions II$(\eta,\tilde d)$. For 
every $\psi\in H^{s}_{\Gi}$, $\widehat\G\in U(\Hi)$ such that $\widehat\G\psi\in H^{s}_{\Gi}$ and $\varepsilon>0$, 
there exist $T>0$ and $u\in L^2((0,T),\R)$ such that $\|\widehat\G\psi-\widetilde \G^u_T\psi\|_{(s)}<\varepsilon$.
\end{coro}
\begin{proof}
First, we consider $K_1,K_2,K_3>0$ such that Proposition 
\ref{approx} holds with times and controls satisfying the 
identities \eqref{casinooo}. Second, the point {\bf 1) (f)} of the proof of 
Proposition \ref{approx} is also valid for the propagator $\widetilde\G_{T}^{u}$ as 
the operator $i(A+u(T-t)B+ic)$ with $c:=\|u\|_{L^\infty((0,T),\R)}\iii B\iii_{(s_1)}$ is maximal dissipative. By 
interpolation, there 
exists a constant $C>0$ depending on $K_1$, $K_2$ and $K_3$ such that $\iii\widetilde\G_{T}^{u} \iii_{(s)}\leq 
C$. Third, for every $\psi\in \Hi$ and $\widehat\G\in 
U(\Hi)$, we call $\widetilde\psi=\widehat\G\psi$ and 
$\widetilde\G=\big(\widehat\G\big)^{-1}\in U(\Hi)$. For every $\varepsilon>0$, we set 
$\widetilde\varepsilon=\varepsilon \,{C}^{-1}$. Proposition \ref{approx} ensures the existence of
$T>0$ and $u\in L^2((0,T),\R)$ 
such 
that $\|\G_{T}^{u}\widetilde\psi-\widetilde\G\widetilde\psi\|_{(s)}
<\widetilde\varepsilon$
and the 
identities \eqref{casinooo} are verified. In conclusion, as $\widetilde 
\G_{T}^{u}=(\G_{T}^{u})^{-1}$ and $\widetilde\G\widetilde\psi=\psi$,
$$\|\widetilde \G_{T}^{u}\psi-\widehat\G\psi\|_{(s)}=\|\widetilde \G_{T}^{u}\psi-\widetilde 
\G_{T}^{u}\G_{T}^{u}\widehat\G\psi\|_{(s)}\leq\iii \widetilde\G_{T}^{u}\iii_{(s)}
\|\psi-\G_{T}^{u}\widehat\G\psi\|_{(s)}\leq C\|\widetilde\G\widetilde\psi-\G_{T}^{u}\widetilde\psi\|_{(s)}
<C\widetilde\varepsilon=\varepsilon.\qedhere$$
\end{proof}

\section{Appendix: Some diophantine
approximation results}\label{numeri}
For $x\in\R$, we denote $E(x)$ the closest integer number to $x$, $\iii x\iii=\min_{z\in\Z}|x-z|$ and 
$F(x)=x-E(x).$ We notice $|F(x)|=\iii x\iii$ and $-\frac{1}{2}\leq F(z)\leq \frac{1}{2}.$ Let $\{L_j\}_{j\leq 
N}\in(\R^+)^N$ and $i\leq N$. We also define
$$n(x):=E\Big(x-\frac{1}{2}\Big),\ \ \ \  r(x):=F\Big(x-\frac{1}{2}\Big),\ \ \ \  d(x):=\iii x-\frac{1}{2}\iii ,\ \ \ \ \widetilde m^i(x) :=n\Big(\frac{L_i}{\pi}x\Big).$$In this appendix, we pursue $\cite[Appendix\ A]{mio3}$, which is based on the techniques from $\cite[Appendix\ A]{wave}$.

\begin{lemma}\label{zuazua1}
	Let $\{L_k\}_{k\leq N}\subset \R^+$, $I_1\subseteq\{1,...,N\}$, $I_2:=\{1,...,N\}\setminus I_1$ and
$$a(\cdot):=\prod_{{i\in I_2}}|\sin((\cdot) L_i)|\sum_{i\in I_1}\prod_{\underset{j\neq i}{j\in 
I_1}}|\cos((\cdot)L_j)|+\prod_{{i\in I_1}}|\cos((\cdot) L_i)|\sum_{i\in I_2}\prod_{\underset{j\neq i}{j\in 
I_2}}|\sin((\cdot) L_j)|.$$ Let $\{\widetilde L_j\}_{j\leq N}\subset\R^+$ be such that $\widetilde L_j=2 L_j$ when 
$j\in I_1$, while $\widetilde L_j=L_j$ when $j\in I_2$. There exists $C>0$ such that, for every $x\in\R$, there 
holds
	\begin{equation*}\begin{split}a(x)\geq C \min\Big( \min_{i\leq N}  \prod_{{j\neq i}} \iii \Big(\widetilde m^i(x)+\frac{1}{2}\Big)\frac{\widetilde L_j}{L_i}\iii ,\
	\min_{i\leq N}  \prod_{{j\neq i}} \iii  m^i(x)\frac{\widetilde L_j}{L_i}\iii\Big).\end{split}\end{equation*}\end{lemma}

\begin{proof}
From $\cite[relation\ (A.3)]{wave}$, for every $x\in\R$, there follows
\begin{equation}\label{primom}2d({x})\leq|\cos(\pi x)|\leq\pi d({x}).\end{equation}
As $2d\big( \big(\widetilde m^i(x)+\frac{1}{2}\big)\frac{L_j}{L_i}\big)  \leq \big|\cos\big(\big(\widetilde 
m^i(x)+\frac{1}{2}\big)\frac{L_j}{L_i}\pi\big)\big|$ 
and $\widetilde m^i(x)+\frac{1}{2}=\frac{L_i}{\pi}x-r\big(\frac{L_i}{\pi}x\big)$ for $x\in\R$ and $i,j\leq N$,
\begin{equation}\begin{split}\label{bingo11}
&2d\Big( \Big(\widetilde m^i(x)+\frac{1}{2}\Big)\frac{L_j}{L_i}\Big)  \leq 
|\cos(L_jx)|+\left|\sin\left(\pi\frac{L_j}{L_i}\Big|r\Big(\frac{L_i}{\pi}x \Big)\Big|\right)\right|.\\
\end{split}\end{equation}
Now, $|\sin(\pi|r(\cdot)|)|\leq\pi\iii | r(\cdot)|\iii\leq \pi|r(\cdot)|= \pi d(\cdot)\leq \frac{\pi}{2}|\cos(\pi (\cdot))|$ thanks to $\cite[relation\ (A.3)]{wave}$ and $(\ref{primom})$. For every $x\in\R$, it holds
\begin{equation}\label{rico}\left|\sin\left(\pi\frac{L_j}{L_i}\Big|r\Big(\frac{L_i}{\pi}x \Big)\Big|\right)\right|\leq \pi\frac{L_j}{L_i}\Big|r\Big(\frac{L_i}{\pi}x \Big)\Big|\leq \frac{\pi L_j}{2 L_i}|\cos(L_ix)|. \end{equation}
From $(\ref{bingo11})$ and $(\ref{rico})$, there exists $C_1>0$ such that, for every $i,j\leq N$,
\begin{equation}\begin{split}\label{bin1}
2d\Big(\Big(\widetilde m^i(x)+\frac{1}{2}\Big)\frac{L_j}{L_i}\Big)  &\leq 
|\cos(L_jx)|+\frac{\pi L_j}{2 L_i}|\cos(L_ix)|,\ \ \ \ \ \ \ \forall x\in\R^+,\\
\end{split}\end{equation}
$$\Longrightarrow\ \ \ C_1\prod_{\underset{j\neq i}{j\in I_1}}d\Big( \Big(\widetilde m^i(x)+\frac{1}{2}\Big)\frac{L_j}{L_i}\Big) \leq \prod_{\underset{j\neq i}{j\in I_1}}|\cos(L_jx)|+|\cos(L_ix)|.$$
From $\cite[relation\ (A.3)]{wave}$, as done in $(\ref{bingo11})$ and $(\ref{rico})$, there exists $C_2>0$ such that
\begin{equation}\label{bin2}\begin{split}
2\iii \Big(\widetilde m^i(x)+\frac{1}{2}\Big)\frac{L_j}{L_i}\iii 
&\leq 
|\sin(L_jx)|+\frac{\pi L_j}{2 L_i}|\cos(L_ix)|,\ \ \ \ \ \ \forall x\in\R,\\
\end{split}\end{equation}
\begin{equation*}\begin{split}
\Longrightarrow\ \ \  \ \  &C_2\prod_{\underset{j\neq i}{j\in I_1}}d\Big( \Big(\widetilde m^i(x)+\frac{1}{2}\Big)\frac{L_j}{L_i}\Big)\prod_{\underset{j\neq i}{j\in I_2}}\iii \Big(\widetilde m^i(x)+\frac{1}{2}\Big)\frac{L_j}{L_i}\iii \leq  \prod_{\underset{j\neq i}{j\in I_2}}|\sin(L_jx)|\prod_{\underset{j\neq i}{j\in I_1}}|\cos(L_jx)|+|\cos(L_ix)|.
\end{split}\end{equation*}
Now, it is satisfied $d(x)=\iii \frac{1}{2}(2x-1)\iii\geq\frac{1}{2}\iii 2x-1\iii =\frac{1}{2}\iii 2x\iii$ for 
every $x\in\R$ and then we have $d\big( \big(\widetilde m^i(x)+\frac{1}{2}\big)\frac{L_j}{L_i}\big)\geq 
\frac{1}{2}\iii \Big(\widetilde m^i(x)+\frac{1}{2}\Big)\frac{2L_j}{L_i}\iii$, which imply
\begin{equation}\label{rico1}\begin{split}
&C_2\prod_{\underset{j\neq i}{j\leq N}}\frac{1}{2}\iii \Big(\widetilde m^i(\cdot)+\frac{1}{2}\Big)\frac{\widetilde L_j}{L_i}\iii \leq a(\cdot)+|\cos(L_i(\cdot))|.
\end{split}\end{equation}
Equivalently, from the proof of $\cite[Proposition\ A.1]{wave}$, for every $x\in\R$,
\begin{equation*}\begin{split}
2\iii m^i(x)\frac{\ L_j}{L_i}\iii &\leq 
|\sin(L_jx)|+\frac{\pi L_j}{2 L_i}|\sin(L_ix)|,\ \ \ \ \ \ \ 
2d\Big( m^i(x)\frac{\ L_j}{L_i}\Big) 
\leq 
|\cos(L_jx)|+\frac{\pi L_j}{2 L_i}|\sin(L_ix)|,\\
\end{split}\end{equation*}
\begin{equation}\label{rico2}\begin{split}
&\Longrightarrow\ \ \ \ C_2\prod_{\underset{j\neq i}{j\leq N}}\frac{1}{2}\iii m^i(\cdot)\frac{\widetilde 
L_j}{L_i}\iii
\leq a(\cdot)+|\sin(L_i(\cdot))|.
\end{split}\end{equation}
The claim follows as $\cite[Proposition\ A.1]{wave}$. Indeed, if $(\lambda_k)_{k\in\N^*}\subset\R^+$ is so that $a(\lambda_k)\xrightarrow{k\rightarrow\infty}0$, then there exist some $i_0\leq N$ such that $|\sin(\lambda_k L_{i_0})|\xrightarrow{k\rightarrow\infty}0$ or $|\cos(\lambda_k L_{i_0})|\xrightarrow{k\rightarrow\infty}0$. By considering $(\ref{rico1})$ and $(\ref{rico2})$ with $i=i_0$, we have
	\begin{equation*}\begin{split}z(\lambda_k):=\min\Big( &\min_{i\leq N}  \prod_{{j\neq i}} \iii \Big(\widetilde 
m^i(\lambda_k)+\frac{1}{2}\Big)\frac{\widetilde L_j}{L_i}\iii,
\min_{i\leq N}  \prod_{{j\neq i}}\iii  m^i(\lambda_k)\frac{\widetilde L_j}{L_i}\iii
\Big)\xrightarrow{k\longrightarrow\infty}0.\end{split}\end{equation*}The lemma is proved since $z(\lambda_k)$ 
converges to $0$ at least as fast as $a(\lambda_k)$ thanks to $(\ref{rico1})$ and $(\ref{rico2})$.\qedhere
\end{proof}

\begin{prop}\label{zuazua3}
	Let $\{L_j\}_{j\leq N}\subset\R$, $I_1\subseteq\{1,...,N\}$ and $I_2:=\{1,...,N\}\setminus I_1$. If $\{L_j\}_{j\leq N}\in\AL\LL(N)$, then, for every $\epsilon>0$, there exists $C_\epsilon>0$ such that, for every $ x>\max\{{\pi}/{2L_j}:\ j\leq N\}$, we have
	{{$$\prod_{{j\in I_2}}|\sin(x L_j)|\sum_{j\in I_1}\prod_{\underset{k\neq j}{k\in I_1}}|\cos(x L_k)|+\prod_{{j\in I_1}}|\cos(x L_j)|\sum_{j\in I_2}\prod_{\underset{k\neq j}{k\in I_2}}|\sin(x L_k)|\geq\frac{C_\epsilon}{x^{1+\epsilon}}.$$}}
\end{prop}

\begin{proof}
The claim is due to Lemma $\ref{zuazua1}$ and to the Schmidt's Theorem $\cite[Theorem\ A.8]{wave}$, which implies that, for every $\epsilon>0$ and $i\leq N$, there exist $C_1(i),C_2(i),C_3(i)>0$ such that, for every $x\in\R$,
\begin{equation}\label{last}\begin{split} \prod_{\underset{j\neq i}{j\leq N}} \iii \Big(\widetilde 
m^i(x)+\frac{1}{2}\Big)\frac{\widetilde L_j}{L_i}\iii
\geq \frac{C_1(i)}{(2\widetilde m^i(x)+1)^{1+\epsilon}}
\geq\frac{C_1(i)}{\big(\frac{2L_i}{\pi}x+1\big)^{1+\epsilon}}
\geq
 \frac{C_2(i)}{x^{1+\epsilon}}\end{split}\end{equation}
and $\prod_{\underset{j\neq i}{j\leq N}}\iii  m^i(x)\frac{\widetilde L_j}{L_i}\iii
\geq{C_3(i)}{x^{-1-\epsilon}}$ for every $x>\frac{\pi}{2}\max\{1/L_j\ :\ j\leq N\}$.
The statement follows with $C_\epsilon:=\min\big(\min_{i\leq N}C_2(i),\min_{i\leq N}C_3(i)\big)$.\qedhere 	\end{proof}

\begin{coro}\label{corozuazua}
	Let $\{L_k\}_{k\leq N}\in\AL\LL(N)$ with $N\in\N^*$. Let $\{\omega_n\}_{n\in\N^*}$ be the unbounded sequence of positive solutions of the equation 
	\begin{equation}\label{zizi}\sum_{l\leq N}\sin(x L_l)\prod_{m\neq l }\cos(x L_m)=0,\ \ \ \ \ \ \ \ \ x\in\R.\end{equation}
	For every $\epsilon>0$, there exists $C_\epsilon>0$ so that $|\cos(\omega_n L_l)|\geq \frac{C_\epsilon}{\omega_n^{1+\epsilon}}$ for every $l\leq N$ and $n\in\N.$
\end{coro}

\begin{proof}If there exists $\{\omega_{n_k}\}_{k\in\N^*}$, subsequence of $\{\omega_n\}_{n\in\N^*}$, such that 
$|\cos(L_j\omega_{n_k})|\xrightarrow{k\rightarrow\infty} 0$ for some $j\leq N$, then there exists $i\leq N$ such 
that $i\neq j$ and $|\cos(L_i\omega_{n_k})|\xrightarrow{k\rightarrow\infty} 0$ thanks to $(\ref{zizi})$. The last 
convergence is at least as fast as the first one. From $(\ref{bin1})$, we have $\prod_{j\neq i}d\Big( 
\Big(\widetilde m^i(\omega_{n_k})+\frac{1}{2}\Big)\frac{L_j}{L_i}\Big)\xrightarrow{k\rightarrow\infty} 0$ at least 
as fast as $|\cos(L_j\omega_{n_k})|\xrightarrow{k\rightarrow\infty} 0$. As in the proof of Lemma 
$\ref{zuazua1}$, there exists $C_2>0$ so that
	$$ C_2|\cos(L_i\omega_n)| \geq \prod_{j\neq i}d\Big( \Big(\widetilde m^i(\omega_n)+\frac{1}{2}\Big)\frac{L_j}{L_i}\Big) =\prod_{j\neq i}\iii \frac{1}{2}\Big(\Big(\widetilde m^i(\omega_n)+\frac{1}{2}\Big)\frac{2L_j}{L_i}-1\Big)\iii.$$
	The last identity and the techniques leading to the equation $(\ref{last})$ achieve the claim.\qedhere 
\end{proof}

\end{document}